\documentclass[letterpaper]{article}
\usepackage{aaai23} % DO NOT CHANGE THIS
\usepackage{times} % DO NOT CHANGE THIS
\usepackage{helvet} % DO NOT CHANGE THIS
\usepackage{courier} % DO NOT CHANGE THIS
\usepackage[hyphens]{url} % DO NOT CHANGE THIS
\usepackage{graphicx} % DO NOT CHANGE THIS
\usepackage{graphicx}  % DO NOT CHANGE THIS
\usepackage{natbib}  % DO NOT CHANGE THIS
\usepackage{caption}  % DO NOT CHANGE THIS
\usepackage{algorithm}
\usepackage{algorithmic}
\usepackage{stmaryrd}
\usepackage{newfloat}
\usepackage{listings}
\floatstyle{ruled}
\newfloat{listing}{tb}{lst}{}
\floatname{listing}{Listing}
\usepackage{xspace}
\usepackage{amsfonts}
\usepackage{thm-restate}
\usepackage{amsmath}
\usepackage{amsthm}
\usepackage{amssymb}
\usepackage{ulem}
\usepackage{misc}   %custom for logic

\usepackage{color}
\usepackage{graphicx}
\usepackage{tikz}
\usepackage{tikz-cd}
\usetikzlibrary{calc}

\usepackage{amsthm}
\usepackage[shortcuts]{extdash} %for \=/

\newcommand{\ourEL}{\ensuremath{\mathcal{ELIHF}}}
\newcommand{\editColor}{black}

\title{Efficient Answer Enumeration in Description Logics with
  Functional Roles -- Extended Version}
\author{
    Carsten Lutz,
    Marcin Przyby{\l}ko
}
\affiliations{
    Institute of Computer Science, Leipzig University, Germany\\
    \{clu,przybyl\}@informatik.uni-leipzig.de
}

\newcommand{\cMP}[1]{{\color{violet} MP:~#1}}

\newtheorem{theorem}{Theorem}
\newtheorem{lemma}{Lemma}
\newtheorem{proposition}{Proposition}
\newtheorem{claim}{Claim}
\newtheorem{example}{Example}

\newcommand{\fun}[3]{\ensuremath{#1\colon #2 \to #3}}

\newcommand{\func}{\mn{func}}
\newcommand{\fEL}{\ensuremath{\mathcal{ELIF}}}

\tikzstyle{quantified}=[circle, draw, scale=.6]
\tikzstyle{answer} = [circle, fill, scale=.6]
\begin{document}

\maketitle

\begin{abstract}
  We study the enumeration of answers to ontology-mediated queries
  when the ontology is formulated in a description logic that supports
  functional roles and the query is a CQ. In particular, we show that
  enumeration is possible with linear preprocessing and constant delay
  when a certain extension of the CQ (pertaining to functional roles)
  is acyclic and free-connex acyclic. This holds both for complete answers and
  for partial answers. We provide matching lower bounds for the
  case where the query is self-join free.
\end{abstract}

% Ontology-mediated querying is a well-established paradigm in which a
% query is enriched with an ontology to facilitate access to incomplete
% and heterogeneous data \cite{some,stuff}.
% In ontology-mediated querying, a query is enriched with an ontology to
% facilitate access to incomplete and heterogeneous data
% \cite{some,stuff}. This framework has been the subject of intense
% research
In ontology-mediated querying, a query is combined with an ontology to
enrich querying with domain knowledge and to facilitate access to incomplete and
heterogeneous data \cite{bienvenu-ontology-disjunctive-datalog,
  DBLP:conf/rweb/CalvaneseGLLPRR09,DBLP:journals/ws/CaliGL12}.
Intense
research has been carried out on the complexity of ontology-mediated
querying, considering in particular description logics and existential
rules
as the ontology languages and conjunctive queries (CQs) as the actual
queries.
% The complexity of ontology-mediated querying has been the subject of
% intense research, studying for example its combined, data, and
% parameterized complexity, and considering various combinations of
% ontology language and query language, see e.g.\ \cite{find,a,lot}.
% Relevant ontology languages include description logics and existential
% rules while conjunctive queries (CQs) are the most studied query
% language.
Most of the existing studies have concentrated on the basic
problem of single-testing which means to decide, given an
ontology-mediated query (OMQ)~$Q$, a database \Dmc, and a candidate
answer $\bar a$, whether $\bar a$ is indeed an answer to $Q$
on~\Dmc. From the viewpoint of many practical applications, however,
the assumption that a candidate answer is provided is hardly realistic
and it seems much more relevant to enumerate, given an OMQ $Q$ and a
database \Dmc, all answers to $Q$ on \Dmc.

The investigation of answer enumeration for OMQs has recently been
initiated in \cite{ourPODS} which also introduces useful new notions
of minimal partial answers; such answers may contain wildcards to
represent objects that are known to exist, but whose exact identity is
unknown. If, for example, the ontology stipulates that
$$
\begin{array}{rcl}
\mn{Researcher} &\sqsubseteq& \exists \mn{worksFor}. \mn{University} \\[1mm]
\mn{Unversity} &\sqsubseteq& \mn{Academia}
\end{array}
$$
and the database \Dmc is $\{ \mn{Researcher}(\mn{mary})\}$, then
there are no complete answers to the CQ
$$q(x,y) = \mn{worksFor}(x,y) \wedge \mn{Academia}(y),$$
but $(\mn{mary},{\ast})$ is a minimal partial answer that conveys information
which is otherwise lost.  The ontologies in \cite{ourPODS} are sets of
guarded existential rules which generalize well-known description
logics such as $\EL$ and $\mathcal{ELIH}$. An important feature of
description logics that is not captured by guarded rules are
functionality assertions on roles, making it possible to declare that
some relations must be
interpreted as partial functions.

The purpose of this paper is to study answer enumeration to
OMQs based on description logics with functional roles, in particular
\ourEL\xspace and its fragments. For the actual queries, we
concentrate on CQs. We consider both the traditional complete answers
and two versions of minimal partial answers that differ in which kind
of wildcards are admitted. In one version, there is only a single
wildcard symbol `$\ast$' while in the other version, multiple
wildcards `$\ast_1$', `$\ast_2$', etc are admitted and multiple
occurrences of the same wildcard represent the same unknown constant.

We study enumeration algorithms with a preprocessing phase
that takes time linear in the size of \Dmc and with constant delay,
that is, in the enumeration phase the delay between two answers must
be independent of \Dmc. Note that we assume the OMQ $Q$ to be fixed
and of constant size, as in data complexity.  If such an algorithm
exists, then enumeration belongs to the complexity class \dlc. If in
addition the algorithm writes in the enumeration phase only a
constant amount of memory, then it belons to the class \cdlin. It is
not known whether \dlc and \cdlin coincide
\cite{kazana-phd}. Enumeration algorithms with these properties have
been studied intensely, see
e.g.\ \cite{berkholz-enum-tutorial,segoufin-enum} for an
overview.

It is an important result for conjunctive queries $q$ without
ontologies that enumeration is possible in \cdlin if $q$ is acyclic and
free-connex acyclic \cite{bagan-enum-cdlin}, the latter meaning that
$q$ is acyclic after adding an atom that covers all answer
variables. If these conditions are not met and $q$ is self-join free
(i.e., no relation symbol occurs more than once), then enumeration is
not possible in \dlc unless
% If $q$ is self-join free (i.e., no relation
% symbol occurs more than once), then there is also a converse: when the
% stated conditions are not met, then enumeration is not possible in
% \dlc unless
certain algorithmic assumptions fail that pertain to the triangle
conjecture, the hyperclique conjecture, and Boolean matrix
multiplication \cite{bagan-enum-cdlin,BraultBaron}. In the presence of
functional dependencies, which in their unary version are identical to
functionality assertions in description logic, this characterization
changes: enumeration in \cdlin is then possible if a certain extension
$q^+$ of $q$ guided by the functional dependencies is acyclic and
free-connex acyclic \cite{carmeli-enum-func}. In particular, adding
functional dependencies may result in queries to become enumerable in
\cdlin which were not enumerable before.

The main results presented in this paper are as follows. We consider
OMQs $Q$ where the ontology \Omc is formulated in the description
logic \ourEL\xspace and the query $q$ is a CQ, and show that
enumerating answers to $Q$ is possible in \cdlin if $q^+$ is acyclic
and free-connex acyclic. For complete answers, this is done by using
carefully defined universal models and showing that they can be
constructed in linear time via an encoding as a propositional Horn
formula. For minimal partial answers (with single or multiple
wildcards), we additionally make use of an enumeration algorithm that
was given in \cite{ourPODS}. Here, we only attain enumeration in \dlc.

We also prove corresponding lower bounds for self-join free queries,
paralleling those in \cite{bagan-enum-cdlin,BraultBaron} and
\cite{carmeli-enum-func}. They concern only ontologies formulated in
the fragment $\mathcal{ELIF}$ of \ourEL\xspace that disallows role
inclusions. The reason is that lower bounds for OMQs with role
inclusions entails a characterization of enumerability in \cdlin for
CQs with self-joins, a major open problem even without ontologies. The
lower bounds apply to complete and (both versions of) minimal partial
answers and are conditional on the same algorithmic assumptions as
in the case without ontologies. Our constructions and
correctness proofs are more challenging than the existing ones in the
literature since, unlike in the upper bounds, we cannot directly use
the query~$q^+$. This is because the transition from $q$ to $q^+$
changes the signature, extending the arity of relation symbols beyond
two, and it is unclear how this can be reflected in the ontology.

We also study the combined complexity of single-testing for the new
notions of minimal partial answers, concentrating on the important
tractable description logics \EL and $\mathcal{ELH}$. It turns out
that tractability is only preserved for acyclic CQs and the single
wildcard version of minimal partial answers, while in all other cases
the complexity increases to \NPclass-complete or {\sc DP}-complete.

% An \emph{$f$-CQ} is a CQ that may contain atoms of the form
% $\nrleq 1 r(x)$. The semantics of $f$-CQs $q$ is defined as for CQs,
% with the additional condition that every homomorphism $h$ from
% $q$ to an interpretation \Imc must satisfy: if $\nrleq 1 r(x) \in q$,
% then $h(x) \in \nrleq 1 r^\Imc$. An \emph{$f$-decoration} of a
% CQ $q$ is any f-CQ that can be obtained from $q$ by choosing
% (any number of) atoms $r(x,y) \in q$ and adding $\nrleq 1 r (x)$.

% \medskip

% Now let $(\Omc,\Sigma,q)$ be an OMQ. We replace $q$ with the f-UCQ
% $q_f$ that contains all f-decorations $p$ of $q$ such that
% %
% \begin{enumerate}

% \item $q \subseteq_\Omc p$ and

% \item for any $f$-decoration $p'$ of $q$ with $p \subsetneq p'$ (this
%   means set inclusion, not containment; in other words: more
%   functionalities), $q \not\subseteq_\Omc p'$.

% \end{enumerate}
% %
% We then work with the f-decorations when increasing the arity in
% $q^+$-style, so we get CD$\circ$Lin when for every CQ $p$ in~$q_f$,
% $p^+$ is acyclic and fc-acyclic. It seems pretty clear that this is
% not tight, that is, there cannot be a matching lower bound, even for
% self-join free queries. We'd also have to argue that the containments
% in Points~1 and~2 above are decidable, which probably isn't too hard.

%Detailed proofs are in the appendix of the long version, made
%available at \cite{arxiv}.

\section{Preliminaries}
\label{sect:prelims}

\newcommand{\Dtree}{\Dmc_{\mn{tree}}}

% We reserve a countably infinite supply of relation symbols of every
% arity and a countably infinite set \Cbf of constants. As usual in
% description logic, we refer to unary relation symbols as \emph{concept
%   names} and to binary relation symbols as \emph{role names}.
Let \NC, \NR, and \NK be countably infinite sets of \emph{concept names},
\emph{role names}, and \emph{constants}.
A
\emph{role} $R$ is a role name $r \in \NR$ or an \emph{inverse role} $r^-$
with $r$ a role name. If $R=r^-$, then $R^- = r$.  An
\emph{\ELI-concept} is built according to the rule
$$C, D ::= A \mid C \sqcap D \mid \exists R.C$$ % \mid \qnrleq 1 r C$$
where $A$ ranges over concept names and $R$ over roles.
%{\color{orange}What about $\bot$?}
% An
% \emph{\EL-concept} is an \ELI-concept that does not use inverse
% roles.
An \emph{\ourEL-ontology} is a finite set of \emph{concept inclusions
  (CIs)} $C \sqsubseteq D$ \emph{role inclusions (RIs)} $R \sqsubseteq
S$, and \emph{functionality assertions}
$\mn{func}(R)$ where (here and in what follows) $C,D$ range over
\ELI concepts and $R,S$ over roles. An
\emph{$\mathcal{ELIF}$-ontology} is an \ourEL-ontology that does not
use RIs.

A \emph{database} is a finite set of \emph{facts} of the form
$A(c)$ or $r(c,c')$ where $A$ is a concept name or $\top$,
$r$ is a role name, and $c,c' \in \NK$. We use
$\mn{adom}(\Dmc)$ to denote the set of constants used in
database~\Dmc, also called its \emph{active domain}. We may write
$r^-(a,b) \in \Dmc$ to mean $r(b,a) \in \Dmc$.

A \emph{signature} is a set of concept and role names, uniformly
referred to as \emph{relation symbols}. For a syntactic object $O$
such as a concept or an ontology, we use $\mn{sig}(O)$ to denote the
set of relation symbols used in it and $||O||$ to denote its size, that
is, the number of symbols needed to write it as a word using a
suitable
encoding. %, with every occurrence of a concept and role name contributing one.
% {\color{violet} Note about size of tuples and sets? We use $|\cdot |$ instead of $||\cdot ||$.}

The semantics is given in terms of \emph{interpretations}
$\Imc=(\Delta^\Imc,\cdot^\Imc)$ where $\Delta^\Imc$ is a non-empty set
called the \emph{domain} and $\cdot^\Imc$ is the interpretation
function, see \cite{baader-introduction-to-dl} for details.
We take the liberty to identify interpretations with non-empty and
potentially infinite databases. The interpretation
function $\cdot^\Imc$ is then defined as $A^\Imc = \{ c \mid A(c) \in
\Imc\}$ for concept names $A$ and $r^\Imc = \{ (c,c') \mid r(c,c') \in
\Imc\}$ for role names~$r$.
An interpretation \Imc \emph{satisfies} a CI $C \sqsubseteq D$ if
$C^\Imc \subseteq D^\Imc$, a fact $A(c)$ if $c \in A^\Imc$, and a fact
$r(c,c')$ if $(c,c') \in r^\Imc$. We thus make the standard names
assumption, that is, we interpret constants as themselves. % This
% implies the unique name assumption. For $S \subseteq \Delta^\Imc$, we
% use $\Imc|_S$ to denote the restriction of \Imc to domain~$S$.
An interpretation \Imc is a \emph{model} of an ontology (resp.\ database) if it
satisfies all inclusions and assertions (resp.\ facts) in it.

A database~$\Dmc$ is \emph{satisfiable} w.r.t.\ an ontology \Omc if
there is a model \Imc of~\Omc and~\Dmc. Note that functionality
assertions in an ontology \Omc can result in databases that are
unsatisfiable w.r.t.\ \Omc. %, due to the standard names assumption.
We write $\Omc \models \mn{func}(R)$ if every model of \Omc satisfies
the functionality assertion $\mn{func}(R)$. In \ourEL, this is
decidable and \ExpTime-complete; see {\color{\editColor} appendix}.
% This is the case iff \Omc contains RIs
% $R_0 \sqsubseteq S_1, \dots, R_n \sqsubseteq S_n$ such that
% $R_0 \in \{r,r^-\}$, $S_n \in \{s,s^-\}$,

\paragraph{Queries.}
A \emph{conjunctive query (CQ)} is of the form
$q(\bar x) = \exists \bar y\,\varphi(\bar x,\bar y)$, where $\bar x$
and $\bar y$ are tuples of variables and $\varphi(\bar x,\bar y)$ is a
conjunction of \emph{concept atoms} $A(x)$ and \emph{role atoms}
$r(x,y)$, with $A$ a concept name, $r$ a role name, and $x,y$ variables
from $\bar x \cup \bar y$.
We % require that all variables in $\bar x$ are
% used in $\varphi$,
call the variables in $\bar x$ the \emph{answer variables} of~$q$, and
use $\mn{var}(q)$ to denote $\bar x \cup \bar y$.  We may write
$\alpha \in q$ to indicate that $\alpha$ is an atom in $q$.
% and sometimes write $r^-(x,y) \in q$ in place of \mbox{$r(y,x) \in q$}.
% CQ $q$ is a \emph{tree} if $\Dmc_q$ is.
For
$V \subseteq \mn{var}(q)$, we use $q|_V$ to denote the restriction of
$q$ to the atoms that use only variables in~$V$.
A \emph{homomorphism} from $q$ to an interpretation \Imc is a function
$h:\mn{var}(q) \rightarrow \Delta^\Imc$ such that $A(x) \in q$ implies
$A(h(x)) \in \Imc$ and $r(x,y) \in q$ implies $r(h(x),h(y)) \in \Imc$.
A tuple $\bar d \in (\Delta^\Imc)^{|\bar x|}$, where $|\bar x|$
denotes
the length of the tuple $\bar x$,
% (with $|\bar x|$ the length of tuple $\bar x$)
is an \emph{answer} to $q$ on interpretation \Imc %, written
% $\Imc \models q(\bar d)$,
if there is a homomorphism $h$ from $q$ to
\Imc with $h(\bar x)=\bar d$.

Every CQ $q$ is associated with a canonical database $\Dmc_q$ obtained
from $q$ by viewing variables as constants and atoms as facts. We
associate every database, and via $\Dmc_q$ also every CQ $q$, with an
undirected graph
$G_\Dmc = (\mn{adom}(\Dmc), \{ \{ a,b \} \mid R(a,b) \in \Dmc \text{
  for some role } R \}$. It is thus clear what we mean by a
\emph{path} $c_0,\dots,c_k$ in a database and a path $x_0,\dots,x_k$
in a CQ.  A CQ $q$ is \emph{self-join free} if every relation
symbol occurs in at most one atom in $q$.

% For CQs $q_1,q_2$, we say that $q_1$ is \emph{contained} in $q_2$,
% written $q_1 \subseteq q_2$, if $q_1(\Dmc) \subseteq q_2(\Dmc)$ for
% every database \Dmc. We say that $q_1$ and $q_2$ are
% \emph{equivalent}, written $q_1 \equiv q_2$, if $q_1 \subseteq q_2$
% and $q_2 \subseteq q_1$.

\paragraph{Ontology-Mediated Queries.}
An \emph{ontology-mediated query (OMQ)} is a pair
$Q=(\Omc,\Sigma,q)$ with \Omc an ontology,
$\Sigma \subseteq \mn{sig}(\Omc) \cup \mn{sig}(q)$
%$\Sigma \subseteq \mn{sig}(\Omc) \cup \mn{sig}(q)$
a finite signature called
the \emph{data schema},
and $q$ a query. We write $Q(\bar
x)$ to indicate that the answer variables of $q$ are $\bar x$.
% {\color{blue}We use
% ELIQ$(\Sigma)$ to denote the language of all ELIQs that use only
% symbols from signature $\Sigma$.}
%
The signature $\Sigma$ expresses the promise that $Q$ is only
evaluated on $\Sigma$-databases.
 Let \Dmc be such a database.
A tuple $\bar a \in \mn{adom}(\Dmc)^{|\bar x|}$ is an \emph{answer} to
$Q(\bar x)$ on \Dmc, written $\Dmc \models Q(\bar a)$, if
$\Imc \models q(\bar a)$ for all models \Imc of \Omc and~\Dmc.
% When
% more convenient, we
We
might alternatively write
$\Dmc,\Omc \models q(\bar a)$. With $Q(\Dmc)$ we denote
the set of all answers to $Q$ on~\Dmc.
%For OMQs $Q_1(\bar
%x)$ and $Q_2(\bar x)$, $Q_i =(\Omc_i,\Sigma,q_i)$, we say that
% $Q_1$
% is \emph{contained} in $Q_2$ and write $Q_1 \subseteq Q_2$, if for
% every $\Sigma$-database \Dmc, %  and every $\bar a \in
% % \mn{adom}(\Dmc)^{|\bar x|}$
% $Q_1(\Dmc) \subseteq Q_2(\Dmc)$.
% We say that
% $Q_1$ is \emph{equivalent} to $Q_2$
% % and write $Q_1 \equiv Q_2$,
% if $Q_1(\Dmc)=Q_2(\Dmc)$ for every $\Sigma$-database~\Dmc.
An OMQ
$Q=(\Omc,\Sigma,q)$ is \emph{empty} if $Q(\Dmc)=\emptyset$ for every
$\Sigma$-database~\Dmc that is satisfiable w.r.t.~\Omc.

We may assume w.l.o.g.\ that \ourEL\xspace ontologies used in OMQs are
in \emph{normal form}, that is, all CIs in it are of one of the
following forms:
%
% $$
% \begin{array}{rclcrcl}
% \top &\sqsubseteq& A, & &
% A_1 \sqcap A_2 &\sqsubseteq& A \\[1mm]
% A_1 &\sqsubseteq& \exists R . A_2 &&
%                                      \exists R . A_1 &\sqsubseteq& A_2
%   % \\[1mm]
% % A_1 &\sqsubseteq& \qnrleq 1 R {A_2}
% \end{array}
% $$
$$
\top \sqsubseteq A, \quad
A_1 \sqcap A_2 \sqsubseteq A, \quad
A_1 \sqsubseteq \exists R . A_2, \quad
\exists R . A_1 \sqsubseteq A_2
$$
where $A_1, A_2, A$ range over concept names. Every \ourEL-ontology
\Omc can be converted into this form in linear time without affecting the
answers to OMQs \cite{baader-introduction-to-dl}.

% Clearly, the OMQs $Q=(\Omc,\Sigma,q)$ and $Q=(\Omc',\Sigma,q)$
% are equivalent if $\Omc'$ is the result of converting \Omc into normal
% form. We may thus generally assume that the ontologies in OMQs are
% in normal form.

With $(\Lmc,\Qmc)$, we denote the \emph{OMQ language} that
contains all OMQs~$Q$ in which $\Omc$ is formulated in DL \Lmc and $q$
in query language \Qmc, such as in $(\ourEL,\text{CQ})$.

\paragraph{Partial Answers.}
Fix a wildcard symbol `$\ast$' that is not in \NK. A \emph{wildcard
  tuple} for a database $\Dmc$ takes the form
$(c_1,\dots,c_n) \in (\mn{adom}(\Dmc) \cup \{\ast\})^n$, $n \geq 0$.
For wildcard tuples $\bar c = (c_{1},\dots,c_{n})$ and
$\bar c' = (c'_{1},\dots,c'_{n})$, we write $\bar c \preceq \bar c'$
if $c'_{i} \in \{ c_i, \ast \}$ for $1 \leq i \leq n$.  Moreover,
$\bar c \prec \bar c'$ if $\bar c \preceq \bar c'$ and
$\bar c \neq \bar c'$. For example, $(a,b) \prec (a,\ast)$ and
$(a,\ast) \prec (\ast,\ast)$ while $(a,\ast)$ and $(\ast,b)$ are
incomparable w.r.t.\ `$\prec$'. Informally, $\bar c \prec \bar c'$
expresses that tuple $\bar c$ is preferred over tuple~$\bar c'$ as it
carries more information.

A \emph{partial answer} to OMQ
$Q(\bar x)=(\Omc,\Sigma,q)$ on an \Sbf-database~$\Dmc$ is a wildcard
tuple $\bar c$ for $\Dmc$ of length $|\bar x|$ such that for each
model $\Imc$ of \Dmc and $\Omc$, there is a $\bar c' \in q(\Imc)$ such
that
% {\color{blue}$\bar c$ can be obtained from $\bar c'$ by replacing every constant
% from $\mn{adom}(I) \setminus \mn{adom}(D)$ with `$\ast$'.}
$\bar c' \preceq \bar c$. Note that some positions in $\bar c'$ may
contain constants from $\mn{adom}(\Imc) \setminus \mn{adom}(\Dmc)$,
and that the corresponding position in $\bar c$ must then have a
wildcard.
%{\color{red} Are we sure $c$ and $c'$ do not have to agree on domain
%  of $D$?}
%since $\bar c$ is a wildcard tuple for $D$.
%
A partial answer $\bar c$ to $Q$ on a $\Sigma$-database $\Dmc$ is a
\emph{minimal partial answer} if there is no partial answer $\bar c'$
to $Q$ on $\Dmc$ with $\bar c' \prec \bar c$.  We use $Q(\Dmc)^{\ast}$
to denote the set of all minimal partial answers to $Q$ on~$\Dmc$.  An
example is provided in the introduction.  Note that
$Q(\Dmc) \subseteq Q(\Dmc)^{\ast}$.  To distinguish them from partial
answers, we also refer to the answers in $Q(\Dmc)$ as \emph{complete}
answers.

\smallskip We also define a second version of minimal partial answers
where multiple wildcards are admitted, from a countably infinite set
$\Wmc = \{ \ast_1,\ast_2,\dots \}$ disjoint from \NK. Multiple
occurrences of the same wildcard then represent the same unknown
constant while different wildcards may or may not represent different
constants. We use $Q(\Dmc)^{\Wmc}$ to denote the set of minimal
partial answers with multiple wildcards. A precise definition is provided in
the {\color{\editColor} appendix}, here we only give an example.
\begin{example}
  Let $Q(x,y,z)=(\Omc,\Sigma,q)$ where \Omc contains
  $$
  \begin{array}{r@{\;}c@{\;}l}
    % \mn{TechWorker} & \sqsupseteq & \mn{Person} \sqcap \exists
    %                                 \mn{hasEmployer} . \mn{TechCompany} \\[1mm]
    % \mn{CarWorker} & \sqsupseteq & \mn{Person} \sqcap \exists
    %                                 \mn{hasEmployer} . \mn{CarCompany}
    % \\[1mm]
    \mn{Company} & \sqsubseteq & \exists \mn{hasEmployee}
                                 . \mn{Person} \\[1mm]
    \mn{TechCompany} &\sqsubseteq& \mn{Company},\ \
    \mn{CarCompany} \sqsubseteq \mn{Company}, \\[1mm]
    % \mn{TechCompany}  & \sqsubseteq & \mn{Company}\\[1mm]
    % \mn{CarCompany}  & \sqsubseteq & \mn{Company}\\[1mm]
    \mn{TechFactory} & \sqsubseteq & \exists \mn{hasOwner} . \mn{TechCompany}
    \\[1mm]
    \mn{CarFactory} & \sqsubseteq & \exists \mn{hasOwner} . \mn{CarCompany}
    %\\[1mm]
    % \mn{TechCompany} & \sqsupseteq &
    % \mn{Company} \sqcap \exists \mn{hasOwner}^- . \mn{TechFactory} \\[1mm]
    % \mn{CarCompany} & \sqsupseteq &
    % \mn{Company} \sqcap \exists \mn{hasOwner}^- . \mn{CarFactory}
    % \\[1mm]
   %  \end{array}
   % $$
   % $$
   % \begin{array}{c}
    % \mn{CarFactory} \sqsubseteq  \mn{Factory} \quad
    % \mn{TechFactory} \sqsubseteq  \mn{Factory} \\[1mm]
    % \mn{Factory} & \sqsubseteq & \exists \mn{hasOwner} . \mn{Company} \\[1mm]
    % \mn{Factory} & \sqsubseteq & \exists \mn{hasOwner} . \mn{Company} \\[1mm]
    % \mn{Tesla} & \sqsubseteq & \mn{TechCompany} \sqcap
    %                                  \mn{CarCompany} \\[1mm]
%    \multicolumn{3}{c}{\mn{func}(\mn{hasOwner}),}
    % \mn{TechWorker} & \sqsubseteq & \mn{Person} \sqcap \exists
    %                                 \mn{hasEmployer} . \mn{TechCompany} \\[1mm]
    % \mn{CarWorker} & \sqsubseteq & \mn{Person} \sqcap \exists
    %                                 \mn{hasEmployer} . \mn{CarCompany} \\[1mm]
    % \mn{TeslaWorker} & \sqsubseteq & \mn{TechWorker} \sqcap
    %                                  \mn{CarWorker} \\[1mm]
    % \multicolumn{3}{l}{\qquad \mn{func}(\mn{hasEmployer}),}
  \end{array}
  $$
  and $\mn{func}(\mn{hasOwner})$, $\Sigma$ contains all symbols from \Omc, and
  $$
  \begin{array}{r@{\;}c@{\;}l}
  q(x,y,z) &=& \mn{Person}(x) \, \wedge\\[1mm]
&&  \mn{hasEmployee} (y,x) \wedge
  \mn{TechCompany}(y) \, \wedge \\[1mm]
 &&   \mn{hasEmployee} (z,x) \wedge
  \mn{CarCompany}(z).
  \end{array}
  $$
 Further consider the database \Dmc with facts $$ %\Dmc = \{
 \mn{CarFactory}(\mn{gigafactory1}), \mn{TechFactory}(\mn{gigafactory1}).
 % \}.
 $$
 Then $Q^\Wmc(\Dmc)=\{ (\ast_1,\ast_2,\ast_2) \}$. If we extend
 \Dmc with $\mn{hasOwner}(\mn{gigafactory1},\mn{tesla})$, then this
 changes to $Q^\Wmc(\Dmc)=\{ (\ast_1, \mn{tesla}, \mn{tesla})\}$.
\end{example}

% \smallskip We shall also consider partial answers that may only use a
% single wildcard `$\ast$'. In this case, repeated occurrences of the
% wildcard do \emph{not} indicate occurrences of the same unknown
% constant.  Formally, all definitions are identical to the
% multi-wildcard case, except that we now use $\Wmc = \{ \ast \}$ and
% Condition~2 of `$\preceq$' is dropped. We use $Q(\Dmc)^{\ast}$ to
% denote the minimal partial answer evaluation of an OMQ $Q(\bar x)$ on
% a database~$\Dmc$ with a single wildcard. For distinction, we also
% speak about the answers in $Q(\Dmc)^\Wmc$ as the minimal partial
% answers \emph{with multi-wildcards}.

\paragraph{Enumeration.}  We are interested in
enumerating the complete and minimal partial answers to a given OMQ
$Q(\bar x) =(\Omc,\Sigma,q)\in (\Lmc,\Qmc)$ on a given
$\Sigma$-database $D$.  An \emph{enumeration algorithm} has a
\emph{preprocessing phase} where it may produce data
structures, but no output. In the subsequent \emph{enumeration phase},
it enumerates all tuples from $Q(D)$, without repetition.
%, followed by
%an end of enumeration signal. % An \emph{all-testing algorithm} for
% $\class{C}$ is defined similarly.  It takes the same two inputs, and
% has the same preprocessing phase, followed by a \emph{testing phase}
% where it repeatedly receives tuples
% $\bar a \in \mn{adom}(D)^{|\bar x|}$ and returns `yes' or 'no'
% depending on whether $\bar a \in Q(D)$.
Answer enumeration for an OMQ language $(\Lmc,\Qmc)$ is
possible \emph{with linear preprocessing and constant delay}, or
\emph{in \dlc}, if there is an enumeration algorithm for
$(\Lmc,\Qmc)$ in which preprocessing takes time
$f(||Q||) \cdot O(||D||)$, $f$ a computable function, while the
delay between the output of two consecutive answers depends only on
$||Q||$, but not on~$||D||$. % Note that this is indeed linear time in
% the sense of data complexity.
Enumeration in \cdlin is defined
likewise, except that
it can use a constant total amount of extra memory in the enumeration phase.
% amount of memory  written  in the
%enumeration phase must be independent of $||Q||$. %  It is not clear whether
% \dlc and \cdlin coincide or not, see e.g.\
% \cite{kazana-phd}.
% Note that this corresponds to data
% complexity: if $||Q||$ is bounded by a constant, then the time required
%   for preprocessing is linear and the delay is
%   $O(1)$.
  % Also note that the time requirement for
% the preprocessing phase can be described as
% fixed-parameter tractability with linear runtime, whence the term FPL.
% The definition of \emph{\dlc} and \emph{\cdlin} for all-testing is analogous,
% except that the enumeration delay is replaced with the time needed for
% testing.

The above definition only becomes precise when we fix a concrete
machine model. We use RAMs under a uniform cost measure
\cite{DBLP:journals/jcss/CookR73}, see \cite{Grandjean-RAM} for a
formalization. A RAM has a one-way read-only input tape, a write-only
output tape, and an unbounded number of registers that store
non-negative integers of $O(\log n)$ bits, $n$ the input size. In this
model, which is standard in the \dlc context, sorting is possible in
linear time and we can access in constant time lookup tables indexed
by constants from $\mn{adom}(D)$ %  or by tuples of constants of length
% $O(1)$
\cite{Grandjean-RAM}.

We also consider \emph{single-testing} which means to decide,
given an OMQ $Q(\bar x)=(\Omc,\Sigma,q)$, a $\Sigma$-database~$D$, and an
answer candidate $\bar c \in \mn{adom}(D)^{|\bar x|}$, whether
$\bar c \in Q(D)$. % We generally consider data complexity, where the
% OMQ $Q$ is fixed and thus of constant size and the only remaining
% inputs are $D$ and $\bar c$. % Unless otherwise

\paragraph{\bf Acyclic CQs.} Let $q(\bar x) =
\exists  \bar y \, \varphi(\bar x, \bar y)$ be a CQ. A \emph{join tree} for
$q(\bar x)$ is an undirected tree $T=(V,E)$ where $V$ is the set of
atoms in $\varphi$ and for each $x \in \mn{var}(q)$, the set $\{
\alpha \in V \mid x \text{ occurs in } \alpha \}$ is a connected
subtree of $T$. % Note that for constant occurring in $q$, there is no
% connectedness condition.
A CQ $q(\bar x)$ is \emph{acyclic} if it has a join tree.
If $q$ contains only unary and binary relations (which shall not
always be the case), then $q$ being acyclic is equivalent to $G_{\Dmc_q}$
being a tree, potentially with multi-edges and self-loops.
% and it is
% \emph{weakly acyclic} if $q|_{\bar y}$ is acyclic.
% SHORT
%\footnote{With `consistently', we mean that different
%  occurrences of the same answer variable must be replaced by the same
%  constant.}
%
A CQ $q(\bar x)$ is
\emph{free\=/connex acyclic} if adding a \emph{head atom}
$H(\bar x)$ that `guards' the answer variables, where $H$ is a
relation symbol of arity $|\bar x|$, results in an acyclic CQ.
% Note
% that other authors have called a CQ $q$ free\=/connex acyclic (or even
% just free\=/connex) if $q$ is both acyclic and (in our sense)
% free\=/connex acyclic
% \cite{berkholz-enum-tutorial}. % We do not
% require the former.
%
Acyclicity and free-connex acyclicity are independent
properties, that is, neither of them implies the other.
% Each of them
% implies weak acyclicity while the converse is
% false.

\section{Upper Bounds}
\label{sect:upperbounds}

We identify cases that admit answer enumeration in \cdlin and \dlc,
considering both complete answers and minimal partial answers. From
now on, we also use relation symbols of arity exceeding two,
identifying concept names and role names with relations symbols of arity one and two.

We start with some preliminaries. Let
$Q(\bar x)=(\Omc,\Sigma,q) \in (\ourEL,\text{CQ})$.  Fix a linear
order on the variables in $q$.  A \emph{path} $y_0, \dots, y_k$ in $q$
is \emph{functional} if for \mbox{$0\leq i <k$}, there is an atom
$R(y_i,y_{i+1}) \in q $ such that $\Omc \models \func(R)$. For a tuple
$\bar x$ of variables, we use $\bar x^+$ to denote the tuple
$\bar x \bar y$ where $\bar y$ consists of all variables~$y$, in the
fixed order, that are reachable from a variable in $\bar x$ on a
functional path in $q$ and that are not part of $\bar x$.

The \emph{FA-extension of $q$} is the CQ $q^+(\bar x^+)$ that contains
an atom $R'(\bar y^+)$ for every atom $R(\bar y)$ in $q$, where $R'$
is a fresh relation symbol of arity $|\bar y^+|$. Note that $q$ and
$q^+$ are over different signatures. %, and that
Moreover, $q^+$ may contain symbols
of high arity and is self-join
free. % Our FA-existensions are defined slightly differently to
% the FD-extensions of  \cite{carmeli-enum-func}, which is due to a
% technical glitch in that paper, see \cite{erratum}.
%
We will also consider the CQ $q^+(\bar x)$ with a non-extended
set of answer variables.
Our FA-extensions are a variation of the notion of
FD-extension used in \cite{carmeli-enum-func}.
% in the context of
% databases with functional dependencies.
% study the enumeration
% of answers to CQs in the presence of functional dependencies (FDs),
% which in their unary incarnation correspond to functionality
% assertions in \ourEL-ontologies. A central role is played by a certain
% extension of CQs called an FD-extension, which we adopt here in a
% slightly adapted form.
%

%\input{draws/tikzu/ex-query-and-fa}
%%%%%%%%%%%%%%%%%%%%%%%%%
%%%%%%%%%%%%%%%%%%%%%%%%%

\begin{figure}[t]
    %\begin{wrapfigure}{r}{0.45\textwidth}
    \centering
    \begin{tikzpicture}[scale=.7]

        %delimiters
        %%%%%%%%%%%%%%%grid
        \draw[dashed,gray] (4.25,-2)--(4.25,2.5);
        \draw[dashed,gray] (8.75,-2)--(8.75,2.5);
        %        \draw[dashed,gray] (15,-1)--(15,2);
        \node (end) at (14,0) {};

        \node (l0) at (2.5, -1.5) {\small CQ $q$};
        \node (x0) at (3.5,0)  {$x$};
        \node (y0) at (2.5,1)  {$y$};
        \node (z0) at (1.5,0)  {$z$};
        \node (t0) at (2.5,2)  {$t$};

        \draw [-latex] (x0)--(y0) node [midway, right] {\scriptsize $R_1$};
        \draw [-latex] (y0)to node [midway, left] {\scriptsize $R_2$} (z0);
        \draw [-latex] (z0)--(x0) node [midway, below] {\scriptsize $R_3$};
        \draw [-latex] (t0)--(y0) node [midway, right] {\scriptsize $R_4$};

        %        \node (l0) at (7.5, -1) {\small Graph of FAs in $Q$};
        %        \node (x1) at (8.5,0)  {$x$};
        %        \node (y1) at (7.5,1)  {$y$};
        %        \node (z1) at (6.5,0)  {$z$};
        %        \node (t1) at (7.5,2)  {$t$};
        %
        %        \draw [-latex] (x1)--(z1);
        %        \draw [-latex] (z1)--(y1);
        %        \draw [-] (z1)--(x1);
        %        \draw [-latex] (t1)--(y1);
        %

        \node (l0) at (6.5, -1.5) {\small Join tree for $q^+(x,t,y)$};
        \node (x2) at (7.5,1)  {\scriptsize ${R}_1'(x{,}y{,}z)$};
        \node (y2) at (5.5,2)  {\scriptsize ${R}_2'(y{,}z)$};
        \node (z2) at (7.5,2)  {\scriptsize ${R}_3'(z{,}x{,}y)$};
        \node (t2) at (5.5,0)  {\scriptsize ${R}_4'(t{,}y)$};

        \draw  (x2)--(z2);
        \draw  (z2)--(y2);
        \draw  (t2)--(y2);

        \node (l0) at (11.5, -1) {\small Join tree for $q^+(x,t,y)$};
        \node (l0) at (11.5, -1.5) {\small with added head atom ${H}$};
        \node (x3) at (12.5,1)  {\scriptsize ${R}_1'(x{,}y{,}z)$};
        \node (y3) at (10 ,2)   {\scriptsize ${R}_2'(y{,}z)$};
        \node (z3) at (12.5,2)  {\scriptsize ${R}_3'(z{,}x{,}y)$};
        \node (t3) at (10,0)    {\scriptsize ${R}_4'(t{,}y)$};
        \node (h3) at (12.5,0)  {\scriptsize ${H}(x{,}t{,}y)$};

        \draw  (x3)--(z3);
        \draw  (z3)--(y3);
        \draw  (t3)--(h3);
        \draw  (h3)--(x3);

    \end{tikzpicture}
    \caption{Illustration of Example~\ref{ex:query-and-fa}
    }

    \label{fig:ex-1}
\end{figure}
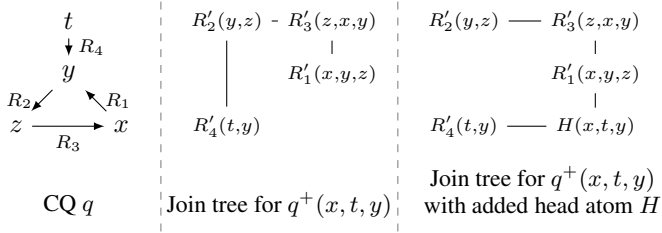

%%%%%%%%%%%%%%%%%%%%%%%%%
%%%%%%%%%%%%%%%%%%%%%%%%%
\begin{example}
    \label{ex:query-and-fa}
    Let $\Omc {=} \{\func(R_2^-), \func(R_3^-), \func(R_4)\}$ and
    $q(x,t) {=} R_1(x,y) \wedge R_2(y,z) \wedge R_3(z,x)\wedge
    R_4(t,y)$. Fix the variable order $x < y <z <t$.
    Then the CQ $q^+(({x,t})^+)$ is
    $$q^+(x,t,y) = {R}_1'(x,y,z), {R}_2'(y,z), \allowbreak
    {R}_3'(z,x,y), \allowbreak {R}_4'(t,y),$$ which is acyclic
    and free\=/connex acyclic while $q$ is neither.  See Figure~\ref{fig:ex-1} for a
    graphical representation.

    In contrast, the version of $q^+$ that has a non-extended set of
    answer variables $\{x,t\}$ (but is otherwise identical) is acyclic, but
    not free\=/connex acyclic.
\end{example}

We state  our upper bound for complete answers.
\begin{theorem}
\label{thm:firstupper}
  Let $Q(\bar x)=(\Omc,\Sigma,q) \in (\ourEL,\text{CQ})$ such that
  $q^+(\bar x^+)$ is acyclic and free-connex acyclic. Then the complete answers
  to $Q$ can be enumerated in \cdlin.
\end{theorem}
%
%In fact, this result is proved in \cite{carmeli-enum-func} for the
%special case that the ontology contains only functionality assertions,
%but no concept inclusions.
The proof of Theorem~\ref{thm:firstupper} relies on the following.
\begin{proposition}
\label{prop:japanesenoodles}
Given an OMQ $Q(\bar x)=(\Omc,\Sigma,q) \in (\ourEL,\text{CQ})$ and a
$\Sigma$-database \Dmc that is satisfiable w.r.t.~\Omc, one can
compute in time $2^{\mn{poly}(||Q||)} \cdot O(||\Dmc||)$ a database
$\Umc_{\Dmc,Q} \supseteq \Dmc$ that satisfies all functionality assertions in~\Omc
and such that $Q(\Dmc)=q(\Umc_{\Dmc,Q}) \cap \mn{adom}(\Dmc)^{|\bar x|}$.
\end{proposition}
The database $\Umc_{\Dmc,Q}$ from
Proposition~\ref{prop:japanesenoodles} should be thought of as a
univeral model for the ontology \Omc and database \Dmc that is
tailored specifically towards the query $q$. Such `query-directed'
universal models originate from \cite{DBLP:conf/ijcai/BienvenuOSX13}.
We compute $\Umc_{\Dmc,Q}$ in time
$2^{\mn{poly}(||Q||)} \cdot O(||\Dmc||)$ by constructing a suitable
propositional Horn logic formula~$\theta$, computing a minimal model
for $\theta$ in linear time \cite{dowling-gallier-horn}, and then
reading off $\Umc_{\Dmc,Q}$ from that model.  Details are provided in
the {\color{\editColor} appendix}.

\smallskip

To prove Theorem~\ref{thm:firstupper}, let
$Q(\bar x)=(\Omc,\Sigma,q) \in (\ourEL,\text{CQ})$ with
$q^+(\bar x^+)$ acyclic and free-connex, and let \Dmc be a
$\Sigma$-database.  We first replace $q$ by the CQ $q_0$ that is
obtained from $q$ by choosing a fresh concept name $D$ and adding
$D(x)$ for every answer variable $x$.  We also replace \Dmc by the
database $\Dmc_0=\Umc_{\Dmc,Q}$ from
Proposition~\ref{prop:japanesenoodles} extended with $D(c)$ for all
$c \in \mn{adom}(\Dmc)$. Clearly, $Q(\Dmc)=q_0(\Dmc_0)$ and thus it
suffices to enumerate the latter.

% We next replace $q$ with a self-join free CQ
% $q_1$ obtained by replacing every atom $S(\bar y)$ in $q$ with
% $S'(\bar y)$ where $S'$ is a fresh relation symbol of the same arity
% as $S$. We call $S'$ a \emph{copy} of $S$. We also replace $\Dmc_0$
% with a database $\Dmc_1$ in the signature of $q'$ obtained from
% $\Dmc_0$ by replacing every fact $S(\bar y)$ with a fact $S'(\bar y)$
% for every copy $S'$ of $S$. Moreover, we use $\Omc_1$ to denote the
% set of functionality assertions $\mn{func}(R')$ such that \Omc
% contains a functionality assertion $\mn{func}(R)$ and $R'$ is obtained
% from $R$ by replacing the role name $r$ in $R$ with a copy of $r$.  It
% should be clear that $q_1$ and $\Dmc_1$ can be constructed time linear
% in $||\Dmc||$, that $\Dmc_1$ satisfies all functionality assertions in
% $\Omc_1$, and that $q_0(\Dmc_0)=q_1(\Dmc_1)$. As a consequence of the
% latter, it suffices to enumerate $q_1(\Dmc_1)$.
%
We next replace $q_0(\bar x)$ by $q^+_0(\bar x^+)$ and $\Dmc_0$ by a database
$\Dmc_0^+$ that reflects the change in signature which comes with the
transition from $q_0$ to~$q_0^+$. More precisely if atom $R(\bar y)$ is
replaced with $R'(\bar y^+)$ in the construction of $q_0^+$ and $h$ is a
homomorphism from $q_0|_{\bar y^+}$ to $\Dmc_0$, then $\Dmc^+_0$
contains the fact $R'(h(\bar y^+))$. The following implies that
$q_0(\Dmc_0)$ is the projection of $q_0^+(\Dmc_0^+)$ to the
first $|\bar x|$ components.
\begin{restatable}{lemma}{lemsamehoms}
  \label{lem:samehoms}
 % $q(\Dmc_0)=q^+(\Dmc_0^+)|_{\bar x}$ and
  Every homomorphism from $q_0$ to $\Dmc_0$ is also a homomorphism from
  $q_0^+$ to $\Dmc_0^+$ and vice versa.  Moreover, $\Dmc^+_0$ can be
  constructed in time linear in $||\Dmc||$.
\end{restatable}
To enumerate $q_0(\Dmc_0)$, we may thus enumerate $q^+(\Dmc_0^+)$ and
project to the first $|\bar x|$ components. The former can be done in
\cdlin: since $q^+(\bar x^+)$ is acyclic and free-connex, so is
$q_0^+(\bar x^+)$, and thus we may apply the \cdlin
enumeration procedure
from~\cite{bagan-enum-cdlin,berkholz-enum-tutorial}.  Clearly,
projection can be implemented in constant time.  To argue that the
resulting algorithm produces no duplicates, it remains to observe that
the answers to $q_0(\Dmc_0)$ and to $q^+(\Dmc_0^+)$ are in a one-to-one
correspondence, that is, every $\bar c \in q_0(\Dmc_0)$ extends in a
unique way to a $\bar c' \in q^+(\Dmc_0^+)$. This, however, is an
immediate consequence of Lemma~\ref{lem:samehoms}, the definition of
$q_0^+$ and the fact that $\Dmc_0$ saties all functionality
assertions in \Omc.

\medskip We now turn to minimal partial answers. Here, we cannot
expect a result as general as Theorem~\ref{thm:firstupper}, a
counterexample is presented in Section~\ref{sect:lower}. We thus
resort to the stronger condition that $q^+(\bar x)$ rather than
$q^+(\bar x^+)$ is acyclic and free-connex acyclic. The difference
between the two conditions is related to the interplay of answer
variables and functional roles.  In particular, $q^+(\bar x)$ and
$q^+(\bar x^+)$ are identical for OMQs $Q=(\Omc,\Sigma,q)$ such that
answer variables have no functional edges to quantified variables,
that is, for every atom $R(x,y)$ in $q$, if
$\Omc \models \mn{func}(R)$ and $x$ is an answer variable, then $y$ is
also an answer variable. % We state our upper bound for minimal
% partial answers.
%
\begin{theorem}
  \label{thm:mpaupper}
  Let $Q(\bar x) =(\Omc,\Sigma,q) \in (\ourEL,\text{CQ})$ such that
  $q^+(\bar x)$ is acyclic and free-connex acyclic. Then the minimal
  partial answers to $Q$ can be enumerated in \dlc, both with
  multi-wildcards and with a single wildcard.
\end{theorem}
To prove Theorem~\ref{thm:mpaupper}, we make use of a recent result
regarding OMQs in which the ontologies are sets of guarded existential
rules. We refer to the class of such ontologies as~$\class{G}$. It was
shown in \cite{ourPODS} that minimal partial answers to OMQs
$Q=(\Omc,\Sigma,q) \in (\class{G},\text{CQ})$ can be enumerated in
\dlc if $q$ is acyclic and free-connex acyclic, both with
multi-wildcards and with a single wildcard.
% $\mathcal{ELIH}$-ontologies in normal form are a special case of
% $\class{G}$-ontologies whereas $\class{G}$-ontologies do not support
% functionality assertions.
The enumeration algorithms presented in
\cite{ourPODS} are non-trivial and we use them as a blackbox. To
achieve this, we need a slightly more `low-level' formulation of the
results from \cite{ourPODS}. In what follows, we restrict our
attention to minimal partial answers with a single wildcard. The
multi-wildcard case is analogous, details are in the {\color{\editColor} appendix}.

Fix a countably infinite set \NN of \emph{nulls} that is disjoint from
\NK and does not contain the wildcard symbol `$\ast$'. In what
follows, we assume that databases may use nulls in place of constants.
Let \Dmc be a database and $q(\bar x)$ a CQ.
For an answer $\bar a \in q(\Dmc)$, we use $\bar a^\ast_\Nbf$ to
denote the unique wildcard tuple for $\Dmc$ obtained from $\bar a$ by
replacing all nulls with `$\ast$'. We call $\bar a^\ast_\Nbf$
a \emph{partial answer} to $q$ on $\Dmc$ and say that it is
\emph{minimal} if there is no $\bar b \in q(\Dmc)$ with
$\bar b^\ast_\Nbf \prec \bar a^\ast_\Nbf$. With $q(\Dmc)^{\ast}_\Nbf$,
we denote the set of minimal partial answers to $q$ on~$\Dmc$.

A database $\Emc$ is \emph{chase-like} if there are
databases $\Dmc_{1},\dots,\Dmc_{n}$ such that
\begin{enumerate}

\item $\Emc=\Dmc_1 \cup \cdots \cup \Dmc_n$,

\item $\Dmc_i$ contains exactly one fact that uses no nulls,
  and that fact contains all constants in  $\mn{adom}(\Dmc_i) \setminus \Nbf$,

  % $\mn{adom}(D_i) \setminus N$ is a guarded
  % set in $D_i$ for $1 \leq i \leq n$, and

\item
 % $(\mn{adom}(D_i) \setminus N) \neq (\mn{adom}(D_j) \setminus N)$
%and
  $\mn{adom}(\Dmc_i) \cap \mn{adom}(\Dmc_j) \cap \Nbf = \emptyset$
    for
  $1 \leq i < j \leq n$.

\end{enumerate}
We call $\Dmc_{1},\dots,\Dmc_{n}$ a \emph{witness} for $\Emc$ being
chase-like.
The term `chase-like' refers to the chase, a well-known procedure for
constructing universal models \cite{DBLP:conf/pods/JohnsonK82}. The
query-directed universal models $\Umc_{\Dmc,Q}$ from
Proposition~\ref{prop:japanesenoodles} are chase-like when the
elements of $N= \mn{adom}(\Umc_{\Dmc,Q}) \setminus \mn{adom}(\Dmc)$
are viewed as nulls. A witness $\Dmc_1,\dots,\Dmc_n$ is obtained by
removing from $\Umc_{\Dmc,Q}$ all atoms $r(a,b)$ with
$a,b \in \mn{adom}(\Dmc)$ and taking the resulting maximally connected
components. The domain sizes $|\mn{adom}(\Dmc'_{i})|$ then only depend
on $Q$, but not on \Dmc.  The following is Proposition~E.1 in
\cite{DBLP:journals/corr/abs-2203-09288}.
\begin{theorem}
  \label{prop:enummultiwildcards}
  For every CQ $q(\bar x)$ that is acyclic and free\=/connex ayclic,
  enumerating the answers $q(\Dmc)^{\ast}_\Nbf$ is in \dlc for
  databases $\Dmc$ and sets of nulls $N \subseteq \mn{adom}(\Dmc)$ such that
  $\Dmc$ is chase-like with witness $\Dmc_{1},\dots,\Dmc_{n}$ where
  $|\mn{adom}(\Dmc_{i})|$ does not depend on $\Dmc$ for $1 \leq i \leq n$.
\end{theorem}
The strategy for proving Theorem~\ref{thm:mpaupper} is now similar to
the case of complete answers. Let
$Q(\bar x)=(\Omc,\Sigma,q) \in (\ourEL,\text{CQ})$ with $q^+(\bar x)$
acyclic and free-connex acyclic, and let \Dmc be a $\Sigma$-database.
It is shown in the {\color{\editColor} appendix} that the query-directed universal model
$\Umc_{\Dmc,Q}$ is also universal for partial answers with a single
wildcard in the sense that
$Q(\Dmc)^\ast=q(\Umc_{\Dmc,Q})^\ast_\Nbf$. We thus first replace \Dmc
with $\Umc_{\Dmc,Q}$, aiming to enumerate
$Q(\Dmc)^\ast=q(\Umc_{\Dmc,Q})^\ast_\Nbf$. We next replace $q(\bar x)$
with $q^+(\bar x)$ and $\Dmc_0=\Umc_{\Dmc,Q}$ with $\Dmc^+_0$.  It
follows from Lemma~\ref{lem:samehoms} that
$q(\Dmc_0) ^\ast_\Nbf=q^+(\Dmc^+_0) ^\ast_\Nbf$.  Note that no
projection is needed since $q^+$ has answer variables $\bar x$ here,
in contrast to answer variables $\bar x^+$ in the case of complete
answers.  It remains to invoke Theorem~\ref{prop:enummultiwildcards}.

\section{Lower Bounds}
\label{sect:lower}

The main aim of this section is to establish lower bounds that
(partially) match the upper bounds stated in
Theorems~\ref{thm:firstupper} and~\ref{thm:mpaupper}. First, however,
we show that Theorem~\ref{thm:mpaupper} cannot be strengthened
by using $q^+(\bar x^+)$ in place of $q^+(\bar x)$.
% by using , when compared to
% Theorem~\ref{thm:firstupper},
% cannot easily be avoided.
All results presented in this section are conditional on algorithmic
conjectures and assumptions. One of the conditions concerns Boolean
matrix multiplication.

A Boolean $n \times n$ matrix is a function
$M:[n]^2 \rightarrow \{0,1\}$ where $[n]$ denotes the set
$\{1,\dots,n\}$.  The \emph{product} of two Boolean $n \times n$
matrices $M_1,M_2$ is the Boolean $n \times n$ matrix
$M_1M_2 := \sum_{c=1}^{n} M_1(a,c)\cdot M_2(c,b)$ where sum and
product are interpreted over the Boolean semiring. In (non-sparse)
\emph{Boolean matrix multiplication (BMM)}, one wants to compute
$M_1M_2$ given $M_1$ and $M_2$ as $n \times n$ arrays.  In
\emph{sparse Boolean matrix multiplication (spBMM)}, input and output
matrices $M$ are represented as lists of pairs $(a,b)$ with
$M(a,b)=1$. The currently best known algorithm for BMM achieves
running time $n^{2.37}$ {\color{\editColor} \cite{williams-soda-laser}} and it is open whether running time $n^2$ can
be achieved; this would require dramatic advances in algorithm theory.
Regarding
spBMM, it is open whether running time $O(|M_1| + |M_2| + |M_1M_2|)$
can be attained, that is, time linear in the size of the input and the
output (represented as lists). This clearly implies BMM in time $n^2$,
but the converse is not known.

The following implies that Theorem~\ref{thm:mpaupper} cannot be
strengthened by using $q^+(\bar x^+)$ in place of $q^+(\bar x)$.
\begin{theorem}
    \label{thm:example-pa-head-extension}
    There is an OMQ $Q(\bar x)=(\Omc,\Sigma,q) \in (\fEL,\text{CQ})$
    such that $q^+(\bar x^+)$ is acyclic and free\=/connex acyclic, but the
    minimal partial answers to $Q$ cannot be enumerated in \dlc unless
    spBMM is possible in time $O(|M_1| + |M_2| + |M_1M_2|)$. This
    holds both for single wildcards and multi-wildcards.
     % The following holds:
     % \begin{enumerate}
     %     \item $q^+$ is acyclic and free\=/connex acyclic,
     %     \item the set of partial answers to $Q$ cannot be enumerated
     %     with linear preprocessing and constant delay
     %     unless sMM conjecture fails.
     % \end{enumerate}
\end{theorem}
\begin{proof}
  Let $Q(\bar x)=(\Omc,\Sigma,q)$ where
  $$
  \begin{array}{r@{\;}c@{\;}l}
        \Omc &=& \{A \sqsubseteq \exists f^-.\top, \func(f)\}\\[1mm]
        \Sigma &=& \{A, r_1,r_2,f\}\\[1mm]
        q(x,z,y) &=& r_1(x,u_1) \wedge f(z,u_1) \wedge f(z,u_2) \wedge r_2(u_2,y).
  \end{array}
  $$
  It is easy to see that $q^+(\bar x^+)$ is just
  $q$, except that now all variables are answer variables. Thus,
  $q^+(\bar x^+)$ is acyclic and free-connex acyclic, as required.

  Assume to the contrary of what is to be shown that there is an
  algorithm that given a database $\Dmc$, enumerates
  $Q(\Dmc)^{\ast}$ in \dlc (the case of multi-wildcards is identical).
  Then this algorithm can be used, given two Boolean
  matrices $M_1M_2$ in list representation, to compute
  $M_1M_2$ in time $O(|M_1| + |M_2| +
  |M_1M_2|)$. This is done as follows.

  Given $M_1$ and $M_2$, we construct a database $\Dmc$ by adding
  facts $r_1(a,c)$ and $A(c)$ for every $(a,c)$ with $M_1(a,c)=1$ and
  facts $r_2(c,b)$ and $A(c)$ for every $(c,b)$ with
  \mbox{$M_2(c,b)=1$}.  It is easy to verify that
  $Q(\Dmc)^{\ast} = \{(a,\ast,b) \mid (a,b)\in M_1M_2\}$.
%
    % \cMP{Let $h$ be a homomorphism from $q$ to $\Umc_{\Dmc, \Omc}$. By construction of
    % $\Umc_{\Dmc, \Omc}$ we have that $h(z) \notin \dom(\Dmc)$ and $h(v) \in \dom(\Dmc)$ for the remaining variables in $q$.
    % Thus, $r_1(h(x), h(u_1)) \in \Dmc$ and $r_2(h(u_2), h(y)) \in \Dmc$.
    % Hence, $(h(x), h(u_1)) \in M_1$ and $(h(u_2), h(y)) \in M_2$.
    % Since $\func(f^{-1})\in \Omc$, we have that $h(u_1) = h(u_2)$}
    % and, thus, $(h(x), h(y)) \in M_1M_2$.
    %
  We may thus construct a list representation of $M_1M_2$ by
  enumerating $Q(\Dmc)^\ast$. Since $|\Dmc| = |M_1| + |M_2|$ and we
  can enumerate in \dlc, the overall time spent is
  $O(|M_1| + |M_2| + |M_1M_2|)$.
\end{proof}
We now consider lower bounds for Theorems~\ref{thm:firstupper}
and~\ref{thm:mpaupper}. Here, we need two additional algorithmic
conjectures that are closely related, both from fine-grained
complexity theory. Recall that a \emph{$k$-regular
  hypergraph} is a pair $H=(V,E)$ where $V$ is a finite set of
vertices and $E \subseteq 2^V$ contains only sets of cardinality
$k$. Consider the following problems:
\begin{itemize}

\item The \emph{triangle detection problem} is to decide, given an
  undirected graph $G=(V,E)$ as a list of edges, whether $G$ contains
  a 3-clique (a ``triangle'').

\item The \emph{$(k+1,k)$-hyperclique problem}, for $k \geq 3$, is to
  decide whether a given $k$-uniform hypergraph $H$ contains a
  hyperclique of size $k+1$, that is, a set of $k+1$ vertices such
  that each subset of size $k$ forms a hyperedge in~$H$.
\end{itemize}
The \emph{triangle conjecture} states that there is no algorithm
for triangle detection that runs in linear time \cite{abboud-triangle}
and the \emph{hyperclique conjecture} states that every algorithm that
solves the $(k+1,k)$-hyperclique problem, for some $k \geq 3$,
requires running time at least $n^{k+1-o(1)}$ with $n$ the
number of vertices \cite{lincoln-soda-grain-complexity}. Note that triangle
detection is the same as $(k+1,k)$-hyperclique for $k=2$, but the
formulation of the two conjectures differs in that the former
refers to the number of edges and the latter to the number of
nodes. The following theorem summarizes our lower bounds.
\begin{restatable}{theorem}{lemmaenumerationlowerbound}
\label{lemma:enumeration-lower-bound}
Let $Q(\bar x)=(\Omc,\Sigma,q) \in ({\color{\editColor} \ELIF},\text{CQ})$ be non-empty with
$q$ self\=/join free and connected.
\begin{enumerate}

\item If $q^+$ is not acyclic, then enumerating complete
  answers to $Q$ is not in \dlc unless the triangle conjecture fails
  or the hyperclique conjecture fails.

\item If $q^+(\bar x^+)$ is acyclic, but not free-connex acyclic, then
  enumerating complete answers to $Q$ is not in \dlc unless spBMM
  is possible in time $O(|M_1| + |M_2| + |M_1M_2|)$.

\end{enumerate}
The same is true for least partial answers, both with a single
wildcard and with multi-wildcards.
\end{restatable}
Recall that we use different versions of $q^+$, namely $q^+(\bar x^+)$
and $q^+(\bar x)$ in Theorems~\ref{thm:firstupper}
and~\ref{thm:mpaupper}. The difference is moot for Point~1 of
Theorem~\ref{thm:mpaupper} as $q^+(\bar x^+)$ is acyclic if and only if
$q^+(\bar x)$ is.
%
% Note that Point~2 of
% Theorem1 speaks about
% $q^+(\bar x^+)$, just like Theorem that addresses complete answers.
% In contrast, Theorem~\ref{thm:mpaupper} concerned with partial answers
% speaks about $q^+(\bar x)$.  This difference is moot for Point~1 of
% Theorem~\ref{lemma:enumeration-lower-bound} since $q^+(\bar x^+)$ is
% acyclic if and only if $q^+(\bar x)$ is. This, however, is not truefor
% free-connex acyclicity and thus there remains a gap between
% Theorems~\ref{thm:mpaupper}
% and~\ref{lemma:enumeration-lower-bound}. In fact, it can be seen that
% $q^+(\bar x^+)$ cannot be replaced with $q^+(\bar x)$ in Point~2 of
% Theorem~\ref{lemma:enumeration-lower-bound} even in the case of
% partial answers.\footnote{Let $Q(\bar x)=(\Omc,\Sigma,q)$ be the OMQ
%   from the proof of Theorem~\ref{thm:example-pa-head-extension} with
%   atom $f(u_1,z)$ in $q$ replaced by $f'(u_2,z)$. Then $q'(\bar x)$ is
%   acyclic, but not free-connex acyclic. Yet, it can be shown that
%   minimal partial answers to $Q$ can be enumerated in \dlc.}

\smallskip The proof of Theorem~\ref{lemma:enumeration-lower-bound} is
inspired by proofs from
\cite{bagan-enum-cdlin,BraultBaron,carmeli-enum-func} and uses similar
ideas.  However, the presence of ontologies and the fact that we
want to capture minimal partial answers makes our proofs much more
subtle. In particular, the constructions in \cite{carmeli-enum-func} first
transition from $q$ to $q^+$ and then work purely on~$q^+$, but we
cannot do this due to the presence of the ontology, which is
formulated in the signature of $q$, not of $q^+$. We (partially)
present the proof of Point~1 and refer to the {\color{\editColor} appendix} for full
detail.

The proof of Point~1 of Theorem~\ref{lemma:enumeration-lower-bound}
splits into two cases.  Recall that the Gaifman graph of a CQ $q$
is the undirected graph that has the atoms of $q$ as its nodes and
an edge between any two nodes/atoms that share a variable. It is
known that if $q$ is not acyclic, then its Gaifman graph is not
chordal or not conformal~\cite{beeri-acyclic}.  Here, chordal means that every
cycle of length at least~4 has a chord and conformal means that for
every clique $C$ in the Gaifman graph, there is an atom in $q$ that
contains all variables in~$C$. The first case of the proof of Point~1
of Theorem~\ref{lemma:enumeration-lower-bound} is as follows.
\begin{restatable}{lemma}{lemmatrianglelower}
\label{lemma:trianglelower}
Let $Q(\bar x)=(\Omc,\Sigma,q) \in ({\color{\editColor}\ELIF},\text{CQ})$ be non-empty such
that $q$ is self\=/join free and connected and the hypergraph of $q^+$
is not chordal.  Then enumerating complete answers to $Q$ is not in
\dlc unless the triangle conjecture fails.  The same is true for least
partial answers, both with a single wildcard and with multi-wildcards.
\end{restatable}
The second case is formulated similarly, but refers to
non-conformality and the hyperclique conjecture.  We give the proof
of Lemma~\ref{lemma:trianglelower}.

\smallskip

Let $Q(\bar x)=(\Omc,\Sigma,q) \in ({\color{\editColor}\ELIF},\text{CQ})$ be as in
Lemma~\ref{lemma:trianglelower}, and let $y_0,\dots,y_k$ be a
chordless cycle in the Gaifman graph of $q^+$ that has length at
least~4.  Let $Y = \{y_0,\dots, y_k\}$ and for easier reference let
$y_{k+1}=y_0$. For every variable $x$ in $q$, we use $Y_x$ to denote
the set of variables $y \in Y$ such that $q$ contains a functional (possibly
empty) path from $x$ to~$y$.

Let $G = (V,E)$ be an undirected graph. % with $V = [n] := \{1,\dots,n\}$.
We may assume w.l.o.g.\ that $G$ does not contain isolated
vertices.
Our aim is to construct a database \Dmc, in time linear in $|E|$, such
that $G$ contains a triangle if and only if $Q(\Dmc) \neq
\emptyset$. Clearly, a \dlc enumeration algorithm for $Q$ lets us
decide the latter in linear time and thus we have found an algorithm
for triangle detection that runs in time linear in $|E|$, refuting the
triangle conjecture.

The construction proceeds in two steps. In the first step, we define a
database $\Dmc_0$ that encodes the graph~$G$. The constants in
$\Dmc_0$ are pairs $\langle x, f \rangle$ with $x\in \mn{var}(q)$ and
$f$ a partial function from $Y$ to $V$. For every
variable $x$ in $q$ and word \mbox{$w = a_0 \dots a_{k} \in V^{\ast}$}
we use $f^w_x$ denote the function that maps each variable
$y_i \in Y_x$ to $a_i$ and is undefined on all other variables. We may
treat $E$ as a symmetric (directed) relation, writing e.g.\
$(a,b) \in E$ and $(b,a) \in E$ if $\{ a,b \} \in E$. For every atom
$r(x,y)$ in $q$ with $r \in \Sigma$, add the following facts to
$\Dmc_0$:
\begin{enumerate}
            \item if $y_0 \in Y_x \cup Y_y$: %, then add
              $r(\langle x,f_x^{ab^k}\rangle,\langle y,f_y^{ab^k}\rangle)$ % and $r(f_y^{ab^k},
              % f_x^{ab^k})$
              for all $(a,b) \in E$,
            \item if $y_k \in Y_x \cup Y_y$: %, then add % facts
              $r(\langle x,f_x^{a^kb}\rangle,\langle y,f_y^{a^kb}\rangle)$ for all $(a,b) \in E$,
            \item if neither is true: %, then add
              $r(\langle x,f_x^{b^{k+1}}\rangle,\langle y,f_y^{b^{k+1}}\rangle)$ for all $a \in V$.
% otherwise, add facts
%               $r(f_x^{b^{k+1}},f_y^{b^{k+1}})$ for all $(a,b) \in E$
%               and all $(b,c) \in E$.
\end{enumerate}
In addition, we add the fact $A(c)$ for every concept name $A \in
\Sigma$ and every constant $c$ introduced above.

%
%We now give an intuition of the reduction. % , oversimplying the situation
 %for better comprehensibility.
%
To provide an intuition for the reduction, let us start with a
description that is relatively simple, but inaccurate. Consider a
homomorphism $h$ from $q$ to $\Dmc_0$. It can be shown that $h$ must
map every variable $y_i$ to a constant of the form
$\langle y_i,f_{y_i} \rangle$ and that the domain
of the function $f_{y_i}$ is $\{y_i\}$. Since $f_{y_i}(y_i)$ is a node
from $G$, the homomorphism $h$ thus identifies a sequence of nodes
$a_0,\dots,a_k$ from $G$, with $a_i=f_{y_i}(y_i)$. The construction of
$\Dmc_0$ ensures that $a_1=\cdots=a_{k-1}$ and $a_0,a_1,a_k$ forms a
triangle in $G$. Conversely, every triangle in $G$ gives rise to a
homomorphism from $q$ to $\Dmc_0$ of the described form. For other
variables $x$ from $q$, the use of the function $f_x$ in constants
$\langle x, f_x \rangle$ serves the purpose of ensuring that all
functionality assertions in \Omc are satisfied in $\Dmc_0$.  A
concrete example for the construction of $\Dmc_0$ is provided
%in Figure~\ref{fig:ex-cyc}
in the {\color{\editColor} appendix}.

The above description is inaccurate for several reasons. First,
instead of homomorphisms into $\Dmc_0$, we need to consider
homomorphisms into the universal model $\Umc_{\Dmc_0,\Omc}$ (defined
in the {\color{\editColor} appendix}).  Then variables $y_i$ need not be mapped to a
constant $\langle y_i,f_{y_i} \rangle$, but can also be mapped to
elements outside of $\mn{adom}(\Dmc_0)$. This does not break the
reduction but complicates the correctness proof. Another difficulty
arises from the fact that \Omc and $q$ may use symbols that do not
occur in $\Sigma$ and and we need these to be derived by \Omc at the
relevant points in~$\Dmc_0$. This is achieved in the second step of
the construction of $\Dmc_0$, described next.

Informally, we want \Omc to derive, at every constant
$c \in \mn{adom}(\Dmc_0)$, anything that it could possibly derive at
any constant in any database. This is achieved by attaching certain
tree-shaped databases to every constant in $\Dmc_0$. We next make this
precise.  Let $\Rmc_\Sigma$  be the set of all role names from $\Sigma$ and
their inverses. The infinite tree-shaped $\Sigma$-database~$D_\omega$
has as its active domain $\mn{adom}(\Dmc_\omega)$ the set of all
(finite) words over alphabet $\Rmc_\Sigma$ and contains the following facts:
\begin{itemize}

\item $A(w)$ for all $w \in \mn{adom}(\Dmc_\omega)$ and concept
  names $A \in \Sigma$;

\item $r(w,w')$ for all $w,w' \in \mn{adom}(\Dmc_\omega)$ with $w'=wr$;

\item $r(w',w)$ for all $w,w' \in \mn{adom}(\Dmc_\omega)$ with
  $w'=wr^-$.

\end{itemize}
We cannot directly use $\Dmc_\omega$ in the construction of \Dmc since
it is infinite. Consider all concept names $A$ such that
$\Dmc_\omega,\Omc \models A(\varepsilon)$. We prove in the {\color{\editColor} appendix}
that these are precisely the concept names $A$ that are
\emph{non-empty}, that is, $\Dmc,\Omc \models A(c)$ for some database
$\Dmc$ and some $c \in \mn{adom}(\Dmc)$. Clearly the number of such
concept names $A$ is finite. By compactness, there is thus a finite
database $\Dtree \subseteq \Dmc_\omega$ such that
$\Dtree,\Omc \models A(\varepsilon)$ for all non-empty concept names
$A$. We may w.l.o.g.\ assume that $\Dtree$ is the initial piece of
$\Dmc_\omega$ of some finite depth $k \geq 1$.

In principle, we would like to attach a copy of $\Dtree$ at every
constant in $\Dmc_0$. This, however, might violate functionality
assertions in \Omc and thus we have to be a bit more careful.  For a
role $R \in \{ r, r^-\}$ with $r \in \Sigma$, let
$\Dmc_R \subseteq \Dmc_{\mn{tree}}$ be the database that consists of
the fact $R(\varepsilon,R)$ and the subtree in $\Dmc_{\mn{tree}}$
rooted at $R$.  Now, the final database~\Dmc used in the reduction is
obtained from $\Dmc_0$ as follows: for every $c \in \mn{adom}(\Dmc_0)$
and every role $R \in \{ r, r^-\}$ with $r \in \Sigma$ such that there
is no fact $R(c, c') \in \Dmc_0$, add a disjoint copy of
$\Dmc_{R}$, glueing the copy of $\varepsilon$ to $c$.

It is easy to see that \Dmc can be computed in time $O(||E||)$. In
particular, the database $\Dtree$ can be constructed (in time
independent of \Dmc) by generating initial pieces of $\Dmc_\omega$ of
increasing depth and checking whether all non-empty concept names are
implied at $\varepsilon$.  In the {\color{\editColor} appendix}, we show that \Dmc
satisfies all functionality assertions in~$\Omc$ and is derivation
complete at $\mn{adom}(\Dmc_0)$.  We then use a rather subtle
analysis to prove the following.
    \begin{restatable}{lemma}{lemcyclecorr}
        \label{claim:lower-bound-partial-answers-and-cycle}
        ~\\[-4mm]
        \begin{description}
        \item[TD1] If there is a minimal partial answer to $Q$ on
          $\Dmc$ (with a single wildcard or with multiple wildcards),
          then there is a triangle in $G$.
            \item[TD2] If there is a triangle in $G$ then there is a complete answer to $Q$ on $\Dmc$.
        \end{description}
    \end{restatable}

\section{Combined Complexity of Single-Testing}
\label{sect:combined}

The results on enumeration provide a (mild) indication that partial
answers can be computationally more challenging than complete ones:
the condition used in Theorem~\ref{thm:firstupper} is weaker than that
in Theorem~\ref{thm:mpaupper}, % as it refers to $q^+(\bar x^+)$ rather
% than $q^+(\bar x)$,
and Theorem~\ref{thm:firstupper} achieves \cdlin while
Theorem~\ref{thm:mpaupper} achieves only \dlc.  Other cases in point
may be found in \cite{ourPODS}. This situation prompts us to study
the effect of answer partiality on the combined complexity of
single-testing.

We concentrate on the fragments \EL and $\mathcal{ELH}$ of
\ourEL\xspace that do not admit inverse roles and functionality
assertions and, in the case of \EL, also no role inclusions. These DLs
bear special importance as single-testing complete answers to OMQs
$Q=(\Omc,\Sigma,q) \in (\mathcal{ELH},\text{CQ})$ is in \PTime if $q$
acyclic and \NPclass-complete otherwise, both in combined complexity,
and thus no harder than without ontologies
\cite{DBLP:conf/semweb/KrotzschRH07,DBLP:conf/ijcai/BienvenuOSX13}. We
will show that making answers partial may have an adverse effect on
these complexities, starting, however, with a positive result. It is
proved by a Turing-reduction to single-testing complete answers.
\begin{restatable}{theorem}{propcombinedptime}
  \label{prop:combinedptime}
  For OMQs $Q=(\Omc,\Sigma,q) \in (\mathcal{ELH},\text{CQ})$ with $q$ acyclic,
  single-testing minimal partial answers with a single-wildcard is in
  \PTime in combined complexity.
\end{restatable}
Partial answers with multi-wildcards are less well-behaved. The
lower bound in the next result is
proved by a reduction from 1-in-3-SAT and only needs a very
simple ontology that consists of a single CI of the form
$A \sqsubseteq \exists r . \top$.
\begin{restatable}{theorem}{thmsinglemulti}
  \label{thm:singlemulti}
  For OMQs $Q=(\Omc,\Sigma,q) \in (\EL,\text{CQ})$ with $q$
  acyclic, single-testing minimal partial answers with
  multi-wildcards is \NPclass-complete in combined complexity.
  The same is true in $(\mathcal{ELH},\text{CQ})$
\end{restatable}
We now move from acyclic to unrestricted CQs. This makes the
complexity increase further, and the difference between single
and multi-wildcards vanishes.
\begin{restatable}{theorem}{thmDP}
  \label{thm:DP}
  For OMQs $Q=(\Omc,\Sigma,q) \in (\EL,\text{CQ})$ single-testing
  minimal partial answers is \DPclass-complete in combined
  complexity. This is true both for single wildcards and
  multi-wildcards, and the same holds also in
  $(\mathcal{ELH},\text{CQ})$.
\end{restatable}
For most other OMQ languages, we %  are complete in combined complexity for
% large complexity classes closed under complement such as \ExpTime, we
do not expect a difference in complexity between single-testing
complete answers and single-testing partial answers.  As an example,
we consider \ourEL\xspace where single-testing complete answers is
\ExpTime-complete \cite{DBLP:conf/jelia/EiterGOS08}.
\begin{restatable}{theorem}{thmourelcomb}
  In $(\ourEL,\text{CQ})$, single-testing minimal partial answers is
  \ExpTime-complete in combined complexity, both with single wildcards
  and multi-wildcards.
\end{restatable}
We remark that the \emph{data} complexity of single-testing minimal
partial answers in $(\ourEL,\text{CQ})$ is in \PTime, both with a
single wildcard and with multi-wildcards. This can be shown by using
essentially the same arguments as in the proof of
Theorem~\ref{prop:combinedptime}.

\section{Conclusion}

It would be interesting to extend our results to
\ourEL\xspace with local functionality assertions, that is, with
concepts of the form $\nrleq 1 R$ or even $\qnrleq 1 R C$. This is
non-trivial as it is unclear how to define the CQ extension~$q^+$. It
would also be interesting and non-trivial to get rid of self-join
freeness in the lower bounds, see
\cite{berkholz-enum-tutorial,DBLP:journals/corr/abs-2206-04988}. Another
natural question is whether answers can be enumerated in some given
order, see e.g.\ \cite{DBLP:conf/pods/CarmeliTGKR21}.
Note that it was observed in \cite{ourPODS} that when enumerating $Q(\Dmc)^\ast$
or $Q(\Dmc)^\Wmc$, it is possible to enumerate the complete answers
before the truely partial ones.

\section*{Acknowledgements}

The authors were supported by the DFG project LU \mbox{1417/3-1} QTEC.

\bibliography{enum}

%\end{document}

\cleardoublepage

\appendix

\section{Additional Details for Section~\ref{sect:prelims}}

We formally define the notion of partial answers with multi-wildcards.
Fix a countably infinite set of \emph{multi-wildcards}
$\Wmc = \{ \ast_1,\ast_2,\dots \}$, disjoint from \NK. A
\emph{multi-wildcard tuple} for a database \Dmc is a tuple
$(c_1,\dots,c_n) \in (\mn{adom}(\Dmc) \cup \Wmc)^n$, $n \geq 0$, such
that if $c_i = \ast_j$ with $j>1$, then there is an $i' < i$ with
$c_{i'}=\ast_{j-1}$. For example, $(a,\ast_1,b,\ast_2,\ast_1,a)$ is a
multi-wildcard tuple. Occurrences of the same wildcard represent
occurrences of the same unknown constant while different wildcards
represent constants that may or may not be
different. % In this sense, multi-wildcard extend
% wildcard-tuples with equality, but not with inequality. We disallow
% tuples such as $(\ast_2,\ast_1)$ to avoid redundancy, as they
% represent the same situation as the (valid) tuple $(\ast_1,\ast_2)$.
%
For multi-wildcard tuples $\bar c = (c_{1},\dots,c_{n})$ and
$\bar c' = (c'_{1},\dots,c'_{n})$, we write $\bar c \preceq \bar c'$ if
\begin{enumerate}

\item $c_{i}=c'_{i}$ or % $c_i \in \mn{adom}(\Dmc)$ and
  $c'_{i} \in \Wmc$ for
  $1 \leq i \leq n$ and

\item %$c_{i},c_{j} \in \Wmc$ and
  $c'_{i} = c'_{j}$
  implies $c_{i} = c_{j}$ for \mbox{$1 \leq i,j \leq n$}.

% \item one of the following is true:
%   %
%   \begin{itemize}
    %
    %   \item $c_{i,1} \notin \Wmc$ and $c_{i,2} \in \Wmc$ for some $i$;
    %
    %   \item $c_{i,1}=c_{j,1} \in \Wmc$ and $c_{i,1} \neq c_{j,2}$ for some $i,j$.
    %
    %   \end{itemize}
%   %
\end{enumerate}
Then, $\bar c \prec \bar c'$ if $\bar c \preceq \bar c'$ and
$\bar c \neq \bar c'$. For example, $(\ast_1,a) \prec (\ast_1,\ast_2)$
and $(a,\ast_1,\ast_2,\ast_1) \prec (a,\ast_1,\ast_2,\ast_3)$.
Informally, $\bar c \prec \bar c'$ means that $\bar c$ is
strictly more informative than $\bar c'$.

A \emph{partial answer with multi-wildcards} to OMQ
$Q(\bar x)=(\Omc,\Sigma,q)$ on a $\Sigma$-database $\Dmc$ is a
multi-wildcard tuple $\bar c$ for $\Dmc$ of length $|\bar x|$ such
that for each model $\Imc$ of $\Omc$ and \Dmc, there is a
$\bar c' \in q(\Imc)$ such that
% {\color{blue}$\bar c$ can be obtained from $\bar c'$ by replacing every constant
    % from $\mn{adom}(I) \setminus \mn{adom}(D)$ with `$\ast$'.}
$\bar c' \preceq \bar c$. % Note that some positions in $\bar c'$ may
% contain elements of $\Delta^\Imc \setminus \mn{adom}(\Dmc)$, and that the
% corresponding position in $\bar c$ must then have a wildcard.
%{\color{red} Are we sure $c$ and $c'$ do not have to agree on domain
    %  of $D$?}
%since $\bar c$ is a wildcard tuple for $D$.
%
A partial answer with multi-wildcards $\bar c$ to $Q$ on a
$\Sigma$-database $\Dmc$ is a \emph{minimal partial answer} if there
is no partial answer with multi-wildcards $\bar c'$ to $Q$ on $\Dmc$
with $\bar c' \prec \bar c$.  The {\em minimal partial answer
      evaluation of $Q(\bar x)$ on~$\Dmc$ with multi-wildcards}, denoted
$Q(\Dmc)^{\ast}$, is the set of all minimal partial answers to $Q$
on~$\Dmc$ with multi-wildcards.  Note that
$Q(\Dmc) \subseteq Q(\Dmc)^{\Wmc}$.

\medskip

Throughout the paper, we assume that given an \ourEL\xspace ontology
\Omc and roles $R,S$, it is decidable whether
$\Omc \models R \sqsubseteq S$ and whether
$\Omc \models \mn{func}(R)$. It is in fact not hard to see that these
problems are \ExpTime-complete. We describe the upper bounds.

\smallskip

To decide whether $\Omc \models R \sqsubseteq S$, one constructs
the database $\Dmc = \{ R(a,b) \}$ and then checks whether $\Dmc,\Omc
\models S(x,y)$, which is possible in \ExpTime
\cite{DBLP:conf/jelia/EiterGOS08}.

\smallskip

To decide whether $\Omc \models \mn{func}(R)$, one constructs the
database $\Dmc = \{ R(a,b_1),R(a,b_2) \}$ and then checks whether
\Dmc is satisfiable w.r.t.\ \Omc which is the case if and only if
$\Dmc,\Omc \not\models \exists x \, A(x)$ with $A$ a fresh concept
name.

\medskip Throughout the appendix, we say that a CQ $q$ is a
\emph{tree} if the undirected graph $G_{\Dmc_q}$ is a tree and $q$
contains no reflexive atoms $r(x,y)$.  Note that multi-edges
$r(x,y),s(x,y) \in q$, with $r \neq s$, are admitted.

\section{Simulations}

We introduce the notion of a simulation, which is closely linked to
the expressive power of \ELI and will be used throughout the appendix.
Let %$\Sigma$ be a schema and let
$\Imc$ and $\Jmc$ be interpretations.
% $\Sigma$\=/
A \emph{simulation} from $\Imc$ to $\Jmc$ is a relation
$S \subseteq \mn{adom}(\Imc) \times \mn{adom}(\Jmc)$ such that
\begin{description}

    \item[Sim1] $A(c) \in \Imc$ and $(c,c') \in S$ implies $A(c') \in \Jmc$ and

    \item[Sim2] $R(c_1,c_2) \in \Imc$  and $(c_1,c'_1) \in S$ implies that
    there is a $c'_2 \in \mn{adom}(\Jmc)$ such that $R(c'_1,c'_2) \in \Jmc$
    and \mbox{$(c_2,c'_2) \in S$}.

%    \item[Sim3] and $R(c_2,c_1) \in I$  and $(c_1,c'_1) \in S$ implies that
%    there is a $c'_2 \in \mn{adom}(J)$ such that $R(c'_2,c'_1) \in J$
%    and $(c_2,c'_2) \in S$.

\end{description}
If there is a simulation $S$  from $\Imc$ to $\Jmc$ such that $(c,c') \in S$, then
we write $(\Imc,c) \preceq (\Jmc,c')$.

\section{Universal Models}
\label{sect:univmod}

We show how to construct a universal model $\Umc_{\Dmc,\Omc}$ of a
database \Dmc and \ourEL-ontology \Omc, using a somewhat unusual
combination of chase and `direct definition' that is tailored towards
the needs of our subsequent proofs and constructions.

In a nutshell, we use chase-like rule applications for the database
part of the universal model and a `direct definition' for the parts of
the universal model generated by existential quantifiers in \Omc.  One
advantage over a pure chase approach is that this avoids ever having
to identify generated objects. To
see why such identifications might be necessary, consider an ontology
that includes
$$
\begin{array}{rcl}
 \Omc &=& \{ A_1 \sqsubseteq \exists r_1 . B_1, A_2 \sqsubseteq \exists
          r_2 . B_2, A \sqsubseteq \exists s . B \\[1mm]
  && \ \ \  r_1 \sqsubseteq s, r_2 \sqsubseteq s, \mn{func}(s) \}
\end{array}
$$
and a database that includes $A_2(c)$. Chasing would generate
$r_1(c,c_1), B_1(c_1)$ and $r_2(c,c_2), B_2(c_2)$ with with $c_1,c_2$
fresh constants. If now other parts of the chase generate $A(c)$, then
$c_2$ has to be merged into $c_1$ and $s(c,c_1)$ has
to be added.

Let \Omc be an \ourEL-ontology in normal form and $\Dmc$ a database
that is satisfiable w.r.t.~\Omc.
% We write  $S \sqsubseteq_\Omc^* R$ if \Omc contains role
% inclusions $R_1 \sqsubseteq R_2, R_2 \sqsubseteq R_3, \dots, R_{n-1}
% \sqsubseteq R_n$ with $R_1=S$ and $R_n=R$.
For each non-empty set $\rho$ of
roles and set $M$ of concept names, define a CQ
$$
q^1_{\rho,M}(x)= \exists y \bigwedge\{ R(x,y)
\mid R \in \rho \} \wedge \bigwedge \{A(y) \mid A \in M\}
$$
and let
$q^2_{\rho,M}(x,y)$ denote the same CQ, but with $y$ an additional answer
variable rather than a quantified variable.
For the chase part of our construction, we use the following chase
rules:
\newcommand{\chaserules}{
    \begin{description}

    \item[R1] If $A_1(c),\dots,A_n(c) \in \Dmc$,
      $\Omc \models A_1 \sqcap \cdots \sqcap A_n \sqsubseteq A$, and
      $A(c) \notin \Dmc$, then add $A(c)$ to \Dmc;

    \item[R2] If $A(c_1),R(c_2,c_1) \in \Dmc$, $\exists R . A
      \sqsubseteq B \in \Omc$, and \mbox{$B(c_2) \notin \Dmc$}, then
      add $B(c_2)$ to \Dmc;

    \item[R3] If $R(c_1,c_2) \in \Dmc$, $R \sqsubseteq S \in \Omc$, and
      $S(c_1,c_2) \notin \Dmc$, then add $S(c_1,c_2)$ to \Dmc;

    % \item[R4] If $A(c_1),R(c_1,c_2) \in \Dmc$,
    %   $A \sqsubseteq \exists S . B \in \Omc$,  $S \sqsubseteq_\Omc^* R$,
    %   $\mn{func}(R) \in \Omc$, and $S(c_1,c_2) \notin \Dmc$, then add
    %   $S(c_1,c_2)$ and $B(c_2)$ to~\Dmc.

    \item[R4] If $A_1(c_1),\dots,A_n(c_1), R(c_1,c_2) \in \Dmc$,
      \mbox{$\mn{func}(R) \in \Omc$},
      $\{A_1(c_1),\dots,A_n(c_n)\},\Omc \models q^1_{\rho,M}(c_1)$,
      $S,R \in \rho$, \mbox{$A \in M$}, and $S(c_1,c_2) \notin \Dmc$ or $A(c_2)
      \notin \Dmc$, then add
      $S(c_1,c_2)$ and $A(c_2)$ to~\Dmc.

    \end{description}
    }
\chaserules
In rule {\bf R1}, we may have $n=0$ to capture CIs of the form
$\top \sqsubseteq A$.  Rule {\bf R4} is best understood in the
light of the example above. For a database \Dmc and ontology \Omc, we use
$\mn{ch}_\Omc(\Dmc)$ to denote the result of exhaustively applying the
above rules to~\Dmc. Note that no rule introduces fresh constants, and
thus rule application terminates. It is also easy to see that the
final result does not depend on the order in which the rules are
applied.

To construct the universal model $\Umc_{\Dmc,\Omc}$, we first build
the chase $\mn{ch}_\Omc(\Dmc)$ and then proceed with the `semantic'
part of the construction as follows.
Let $c \in \mn{adom}(\Dmc)$. We use $M_c$ to denote the set of concept
names $A$ with $A(c) \in \mn{ch}_\Omc(\Dmc)$. Moreover, we use $\Dmc_M$
to denote the database $\{ A(\widehat c) \mid A \in M \}$ where
$\widehat c$ is an
arbitrary, but fixed constant. If $M=M_c$, then we take $\widehat c=c$.
 For a non-empty set
$\rho$ of roles and set $M$ of concept names, we write
$c\rightsquigarrow_{\Dmc,\Omc}^\rho M$ if
\begin{enumerate}

\item %$\{A(c) \mid A(c) \in \mn{ch}_\Omc(\Dmc)\},
  $\Dmc_{M_c},\Omc\models q^1_{\rho,M}(c)$,

\item $\rho$ and $M$ are maximal with this property, that is, there
  are no $\rho' \supseteq \rho$ and $M' \supseteq M$ with
  $(\rho',M') \neq (\rho,M)$ such that
  $\Dmc_{M_c},\Omc\models q^1_{\rho',M'}(c)$, and

% \item there are no $R \in \rho$ and $R(c,c') \in \mn{ch}_\Omc(\Dmc)$
%   such that
%   $S \sqsubseteq^*_\Omc R$ and
%   $\mn{func}(R) \in \Omc$.

  \item there is no $R(c,c') \in \mn{ch}_\Omc(\Dmc)$
  such that $R \in \rho$ and
  $\mn{func}(R) \in \Omc$.

% \item there is no $c' \in \mn{adom}(\Dmc)$ with
%   $\Dmc,\Omc\models q^2_{\rho,M}(c,c')$.

% \item $\Dmc_{M_i},\Omc \models q^1_{\rho_{i+1},M_{i+1}}(\widehat c)$ for
%   $1 \leq i < n$.

\end{enumerate}
Similarly, for a non-empty set $\rho$ of roles  and sets $M_1,M_2$ of
concept names, we write $M_1 \rightsquigarrow_{\Dmc,\Omc}^\rho M_2$
if for $1 \leq i < n$:
\begin{enumerate}

\item[4.] $\Dmc_{M_i},\Omc \models q^1_{\rho_{i+1},M_{i+1}}(\widehat
  c)$ and

\item[5.] $\rho_{i+1}$ and $M_{i+1}$ are maximal with this property.

\end{enumerate}
A \emph{trace for $\Dmc$ and} $\Omc$ is a sequence
$$t = c\rho_1M_1\rho_2M_2\dots \rho_nM_n,\ n \geq 0,$$
where $c \in \mn{adom}(\Dmc)$, $\rho_1,\dots,\rho_n$ are sets of roles
that occur in \Omc, and $M_1,\dots,M_n$ are sets of concept names that
occur in \Omc such that for $1 \leq i < n$:
\begin{enumerate}

\item[6.] $c\rightsquigarrow_{\Dmc,\Omc}^{\rho_1}M_1$,

\item[7.]
  $M_i \rightsquigarrow_{\Omc}^{\rho_{i+1}}M_{i+1}$,

\item[8.] there is no role name $R$ with $R^- \in \rho_i$,
  $R \in \rho_{i+1}$, and $\mn{func}(R) \in \Omc$, for
  $1 \leq i < n$.

% \item[8.] if   $A_1 \sqsubseteq \qnrleq 1 r {A_2} \in \Omc$, $A_1 \in
%   M_1$, $\Dmc,\Omc \models A_2(c)$, and $A_2 \in M_2$,
%   then $R^- \notin \rho_1$ or $R \notin \rho_2$, and

% \item[9.] if $A_1 \sqsubseteq \qnrleq 1 r {A_2} \in \Omc$, $A_1 \in
%   M_i$ with $1 < i < n$, and $A_2 \in M_{i-1} \cap M_{i+1}$,
%   then $R^- \notin \rho_{i}$ or $R \notin \rho_{i+1}$.

\end{enumerate}
Let $\Tbf$ denote the set of all traces for $\Dmc$
and $\Omc$. Then, the \emph{universal model} of \Dmc and \Omc is the
following database with active domain $\Tbf$, viewed as an
interpretation:
\begin{align*}
    \Umc_{\Dmc, \Omc} = {} & \mn{ch}_\Omc(\Dmc) \cup {}\\
%                           & \{ S(c,c') \mid R(c,c') \in \Dmc \text{ and
    %			   } \Dmc,\Omc \models S(c,c')\}\cup {}\\
                           & \{A(t\rho M) \mid t\rho M \in \Tbf \text{ and }
                             A \in M\}
   \cup {}\\
   &\{R(t, t\rho M) \mid t\rho M \in \Tbf \text{ and } R \in \rho \}.
\end{align*}
The following lemma summarizes the most important properties of
universal models.
\begin{lemma}
  \label{lem:univbasic}
  ~\\[-4mm]
  \begin{enumerate}

      \item $\Umc_{\Dmc, \Omc}$ is a model of \Dmc and \Omc.

      \item If \Imc is a model of \Dmc and \Omc, then there is
        a homomorphism from $\Umc_{\Dmc, \Omc}$ to \Imc that
        is the identity on $\mn{adom}(\Dmc)$.

      \item $Q(\Dmc) = q(\Umc_{\Dmc, \Omc}) \cap \mn{adom}(\Dmc)^{|x|}$.

      \end{enumerate}
\end{lemma}
To prove Lemma~\ref{lem:univbasic}, let us first give two
technical lemmas that will also be useful later on.
\begin{lemma}
  \label{lem:sim} Let $\Omc$ be an \ELIF ontology in normal form,
  $\Dmc_1,\Dmc_2$ be databases that are satisfiable w.r.t.\ \Omc, and
  $c_i \in \mn{adom}(\Dmc_i)$ for $i \in \{1,2\}$. Then
  $(\Dmc_1,c_1) \preceq (\Dmc_2,c_2)$ and
  $\Dmc_1, \Omc \models A(c_1)$ implies $\Dmc_2, \Omc \models A(c_2)$
  for every concept name $A$.
\end{lemma}
%
%It is enough to show that if $(\Dmc_1,c_1) \preceq (\Dmc_2,c_2)$
% then for every concept $C$ in \Omc,
% $C(c_1) \in \Umc_{\Dmc_1, \Omc}$ implies $C(c_2) \in \Umc_{\Dmc_2, \Omc}$.
%
% By construction of the universal model, since $c_1 \in \dom(\Dmc_1)$, then either $C(c_1) \in \Dmc_1$
% or $C(c_1)$ was added by the chase. Thus, $C(c_1) \in \mn{ch}_{\Omc}(\Dmc_1)$.
%
\begin{proof}
Consider the chase procedure used as the first part of the
construction of universal models.
Started on $\Dmc_1$, it constructs a finite sequence
$$\Dmc_1 = I_0 \subseteq I_1 \subseteq \cdots \subseteq I_k =
\mn{ch}_\Omc(\Dmc_1).$$
Since $(\Dmc_1,c_1) \preceq (\Dmc_2,c_2)$, there is a simulation $S$
from $\Dmc_1$ to $\Dmc_2$ with $(c_1,c_2) \in S$. We prove by induction
on $i$ that
\begin{itemize}
    \item [$(\dagger)$] %for $0 \leq i \leq k$, $S_i \subseteq
    % \dom(I_i) \times \dom(J_i)$
    $S$
    is a simulation from $I_i$ to $\mn{ch}_\Omc(\Dmc_2)$.
\end{itemize}
The induction start is trivial. For the induction step assume that we
have already shown that $S$ is a simulation from $I_i$ to
$\mn{ch}_\Omc(\Dmc_2)$. We need to verify that Conditions \textbf{Sim1}
and \textbf{Sim2} are also satisfied for $S$ as a simulation from
$I_{i+1}$ to $\mn{ch}_\Omc(\Dmc_2)$.

Recall that $I_{i+1}$ was obtained from $I_i$ by application of one of
the rules \textbf{R1} to \textbf{R4} from Section~\ref{sect:univmod}.
Since \Omc is formulated in \ELIF (and thus contains no role
inclusions) %, however,
rule \textbf{R3} cannot be used. %This also means that
Similarly,
no new binary facts $R(c,c')$ are added by the rules and, thus,
\textbf{Sim2} holds by inductive assumption.

For \textbf{Sim1}, let $(c,c') \in S$ and $A(c) \in I_{i+1}$.  If
$A(c) \in I_i$ then $A(c') \in \Umc_{\Dmc_2, \Omc}$ by the inductive
assumption.  Otherwise, $A(c)$ was added by one of the above rules.
Thus, $A(c)$ is a consequence of some facts
$S = \{A_1(c), \dots, A_\ell(c), R_1(c,c_1), \dots, R_k(c,c_k)\}
\subseteq I_i$ and the ontology, i.e. $S, \Omc \models A(c)$.  By
inductive assumption, there are constants $c'_i$, for $1\leq j \leq k$
such that
$S'= \{A_1(c'), \dots, A_\ell(c), R_1(c,c'_1), \dots, R_k(c,c'_k)\}
\subseteq \Umc_{\Dmc_2, \Omc}$.  Clearly, we have that
$S', \Omc \models A(c')$.  Since $\Umc_{\Dmc_2, \Omc}$ is a model of
$\Omc$, we infer that $A(c') \in \Umc_{\Dmc_2, \Omc}$.  This finishes
the proof of~($\dagger$).

\smallskip Now assume that $\Dmc_1, \Omc \models A(c_1)$.  Then
Point~3 of Lemma~\ref{lem:univbasic} yields
$A(c_1) \in \Umc_{\Dmc_1,\Omc}$. Since the construction of
$\Umc_{\Dmc_1,\Omc}$ from $\mn{ch}_\Omc(\Dmc_1)$ adds no
facts $A(c)$ with $c \in \mn{adom}(\Dmc_1)$, this yields $A(c_1) \in
\mn{ch}_\Omc(\Dmc_1)$. Now ($\dagger$) yields $A(c_2) \in
\mn{ch}_\Omc(\Dmc_2)$, thus $A(c_2) \in \Umc_{\Dmc_2,\Omc}$,
implying $\Dmc_2, \Omc \models A(c_2)$.
\end{proof}
For
$c \in \dom(\Dmc)$, we use $\Umc_{\Dmc, \Omc}^{\downarrow c}$ to
denote the restriction of $\Umc_{\Dmc, \Omc}$ to the domain that
consists of all traces starting with $c$. Note that, by construction
of $\Umc_{\Dmc, \Omc}$, $\Umc_{\Dmc, \Omc}^{\downarrow c}$ takes the
shape of a tree without multi-edges and reflexive loops. We also refer
to $\Umc_{\Dmc, \Omc}^{\downarrow c}$ as \emph{the tree in
      $\Umc_{\Dmc, \Omc}$ rooted in $c$}.
\begin{lemma}
  \label{lem:newintermediate}
  Let \Omc be an \ourEL-ontology and let $\Dmc_1,\Dmc_2$ be databases
  that are satisfiable w.r.t.\ \Omc and $c_1,c_2$ be constants such that
  $\Dmc_1, \Omc \models A(c_1)$ implies $\Dmc_2, \Omc \models A(c_2)$
  for every concept name $A$.  Further, let \Jmc be a model of $\Dmc_2$
  and~\Omc.  Then %, $(\Dmc_1,c_1) \preceq (\Dmc_2,c_2)$ implies that
  there is a homomorphism $h$ from $\Umc^{\downarrow c_1}_{\Dmc_1,\Omc}$
  to \Jmc with $h(c_1)=c_2$.
\end{lemma}
\begin{proof}
  % We will build the homomorphism inductively with respect to the the
  % length $n$ of the traces in the tree
  % $\Umc^{\downarrow c_1}_{\Dmc_1,\Omc}$.
%
  For brevity, let $\Imc = \Umc^{\downarrow c_1}_{\Dmc_1,\Omc}$.
  Let $\Tbf_n$, $n \geq 0$, be the set of traces in $\Imc$
  of length $n$. In particular $\Tbf_0 = \{c_1\}$.

  For all $n \geq 0$ we construct a homomorphism $h_n$ from
  $\Imc|_{\Tbf_n}$ to $\Jmc$ such that $h_n(c_1) =c_2$. The desired
  homomorphism $h$ is then obtained in the limit, that is, $h=
  \bigcup_{n \geq 0} h_n$.

  The homomorphism $h_0$ is simply the identity on $c_1$.  Since
  $\Dmc_1, \Omc \models A(c_1)$ implies $\Dmc_2, \Omc \models A(c_2)$
  for every concept name $A$, $h_0$ is a homomorphism.

  For the inductive step,
  let us assume that we have already defined $h_n$.
  We show how to use it to define $h_{n+1}$.

  We first take $h_{n+1}= h_n$. Now, let $w\rho M \in \Tbf_{n+1} \setminus \Tbf_n$ be a trace in the next layer of the tree.
  Since $w \in \Tbf_n$, $h_{n+1}(w)$ is already defined.
  Let $N = \{A \mid A(w) \in \Imc\}$ and $N' = \{A \mid A(w) \in \Jmc\}$. Since $h_n$
  is a homomorphism, $N \subseteq N'$.
  By construction of the universal model $\Umc_{\Dmc_1,\Omc}$ we have
  $M = \{A \mid A(w\rho M) \in \Imc \}$.
  Moreover, $D_N, \Omc \models q^1_{\rho,M}(w)$.
  Since $\Jmc$ is a model of $\Dmc_2$ and $\Omc$, and $N \subseteq N'$, we infer that
  $D_{N'}, \Omc \models q^1_{\rho,M}(h(w))$.
  Therefore, there is a constant $d \in \dom(\Jmc)$
  such that $\Jmc, \Omc \models q^2_{\rho,M}(h(w), d)$. We set
  \begin{itemize}
          \item[($\dagger$)] $h_{n+1}(w\rho M) = d$.
      \end{itemize}

  We now show that $h_{n+1}$ is a homomorphism from $\Imc|_{\Tbf_{n+1}}$ to $\Jmc$.
  For $w, w' \in \Tbf_n$,  we have $h_{n+1}(w) = h_n(w)$ and $h_{n+1}(w') = h_n(w')$.
  Thus, by inductive assumption, $A(w) \in \Imc$ implies $A(h_{n+1}(w)) \in \Jmc$
  and $R(w,w') \in \Imc$ implies $R(h_{n+1}(w), h_{n+1}(w')) \in \Jmc$.

  Thus, let $w\in \Tbf_{n+1} \setminus \Tbf_n$. By definition of $\Umc_{\Dmc_1, \Omc}$ there is a unique $w' \in \Tbf_n$ such that $w = w'\rho M$
  where $M = \{A \mid A(w) \in \Umc_{\Dmc_1, \Omc}\}$, and $\rho = \{R \mid R(w',w) \in \Umc_{\Dmc_1, \Omc}\}$.
  Moreover, since $\Imc$ is a tree, we have that for all $w'' \in \Tbf_{n+1}$ if $R(w'',w)$ for some role name $R$
  then $w''=w'$.

  Now, if $A(w) \in \Imc$ then $A \in M$ and $A(h_{n+1}(w)) \in \Imc$ by the choice ($\dagger$) of $h_{n+1}(w)$.
  Similarly, if $R(w,w'') \in \Imc$ for some $w'' \in \Tbf_{n+1}$ then $w''=w'$ and $R(h_{n+1}(w),h_{n+1}(w'')) \in \Imc$
  by, again, the choice ($\dagger$) of $h_{n+1}(w)$.
%
  % This ends the proof that $h_{n+1}$ is a homomorphism from
  % $\Imc|_{\Tbf_{n+1}}$ to $\Jmc$.
%
  % Finally, we define $h$ as $h = \bigcup_{n\geq 0} h_{n}$. It is easy to check that $h$ is
  % the desired homomorphism from $\Umc^{\downarrow c_1}_{\Dmc_1,\Omc}$ to $\Jmc$.
%
\end{proof}
We now give the proof of Lemma~\ref{lem:univbasic}.
\begin{proof}[Proof of Lemma~\ref{lem:univbasic}]
  We start with Point~1. It is clear by construction that
  $\Umc_{\Dmc, \Omc}$ is a model of \Dmc.
  %
  % It is straightforward
  % to show the following, essentially by induction on the number
  % of rule applications used to construct $\mn{ch}_\Omc(\Dmc)$:
  % %
  % \begin{itemize}
      % \item[(i)] $c \in A^{\mn{ch}_\Omc(\Dmc)}$ implies
      %   $\Dmc,\Omc \models A(c)$;

      % \item[(ii)] $(c,c') \in R^{\mn{ch}_\Omc(\Dmc)}$ implies
      %   $\Dmc,\Omc \models R(c,c')$.

      % \end{itemize}
  % %
  % Note that we might replace the precondition with
  % $c \in A^{\Umc_{\Dmc,\Omc}}$ in (i) since due to Rule~{\bf R1} we
  % have $c \in A^{\mn{ch}_\Omc(\Dmc)}$ iff
  % $c \in A^{\Umc_{\Dmc,\Omc}}$.  Likewise, the precondition in
  % (ii)~might be replaced with $(c,c') \in R^{\Umc_{\Dmc,\Omc}}$.

  %  We next note that for all
  % $c,c' \in \mn{adom}(\Dmc)$, concept names $A$ and roles $R$,
  % %
  % \begin{itemize}
      % \item[(i)] $c \in A^{\Umc_{\Dmc, \Omc}}$ iff
      %   $\Dmc,\Omc \models A(c)$;
    %
      % \item[(ii)] $(c,c') \in R^{\Umc_{\Dmc, \Omc}}$ iff
      %   $\Dmc,\Omc \models R(c,c')$.
    %
      % \end{itemize}
  % %
%
  We next show that $\Umc_{\Dmc, \Omc}$ also satisfies all CIs in
  \Omc, making a case distinction according to the form of the CI:
  \begin{itemize}

      \item $\top \sqsubseteq A \in \Omc$.

        For all $c \in \mn{adom}(\Dmc)$, Rule {\bf R1} ensures that
        $A(c) \in \mn{ch}_\Omc(\Dmc) \subseteq \Umc_{\Dmc, \Omc}$.  For
        traces $t=t'\rho M$, by construction of $\Umc_{\Dmc, \Omc}$ it
        suffices to show that $A \in M$. But this is a consequence of the
        definition of traces, in particular of Conditions~2 and~5 above
        which from now on we refer to as the `maximality conditions'.

      \item $A_1 \sqcap A_2 \sqsubseteq A \in \Omc$.

        Assume that
        $t \in A_1^{\Umc_{\Dmc, \Omc}} \cap A_2^{\Umc_{\Dmc, \Omc}}$. If
        $t=c \in \mn{adom}(\Dmc)$, then
        $t \in A_1^{\mn{ch}_\Omc(\Dmc)} \cap A_2^{\mn{ch}_\Omc(\Dmc)}$ and
        rule {\bf R1} ensures that $A(c) \in \Umc_{\Dmc,\Omc}$. If
        $t=c \rho M$ or $t=t' \rho M$, then
        $t \in A_1^{\Umc_{\Dmc, \Omc}} \cap A_2^{\Umc_{\Dmc, \Omc}}$
        implies $A_1,A_2 \in M$, thus also $A \in M$ by the maximality
        conditions, which yields $t \in A^{\Umc_{\Dmc, \Omc}}$ by
        construction of $\Umc_{\Dmc, \Omc}$.

      \item $A_1 \sqsubseteq \exists R . A_2 \in \Omc$.

        Let $t \in A_1^{\Umc_{\Dmc, \Omc}}$. First assume that
        $t=c \in \mn{adom}(\Dmc)$.  Further assume that
        \begin{itemize}

            \item[($*$)] there are $S(c,c_2) \in \Dmc$ and $\rho,M$ such
              that \mbox{$\mn{func}(S) \in \Omc$},
              $\Dmc_{M_{c}},\Omc \models q^1_{\rho,M}(c)$,
            $S,R \in \rho$, and $A_2 \in M$.

            % there is some $S(c,c') \in\Dmc$ with
            % $R \sqsubseteq^*_\Omc S$ and $\mn{func}(S) \in \Omc$.
          \end{itemize}
      Then Rule {\bf R4} makes sure that $R(c,c_2)$ and $A_2(c_2)$ are in
      $\Umc_{\Dmc,\Omc}$, and thus
      $c \in (\exists R . A_2)^{\Umc_{\Dmc,\Omc}}$ as required.  Now
      assume that $(*)$ does not hold. Clearly
      $t \in A_1^{\Umc_{\Dmc, \Omc}}$ implies
      $A_1(t) \in \mn{ch}_\Omc(\Dmc)$.  Thus $A_1 \in M_c$ and
      consequently there must be some set of roles $\rho$ and of concept
      names $M$ such that $R \in \rho$, $A_2 \in M$, and
      $\Dmc_{M_c},\Omc \models q^1_{\rho,M}(c)$. We may assume w.l.o.g.\
      that $\rho$ and $M$ are maximal with this property. Thus
      Conditions~1 and~2 of $c \rightsquigarrow^\rho_{\Dmc,\Omc} M$ are
      satisfied, and it is not hard to show that Condition~3 must also be
      satisfied. If, in fact, Condition~3 is violated then there is an
      $R'(c,c') \in \mn{ch}_\Omc(\Dmc)$ with $R \in \rho$ and
      $\mn{func}(R) \in \Omc$. But then ($*$) is satisfied for $\rho,M$
      and \mbox{$S=R$}.

        \smallskip

        Now assume that $t=c \rho M$.  Then $A_1 \in M$ and
        $\Dmc_{M_c},\Omc \models q^1_{\rho,M}(c)$. If
        $\Dmc_{M_c} \models A_2(c)$ and $R^- \in \rho$, then an easy
        semantic argument shows $\Dmc_{M_c},\Omc \models A_2(c)$ and thus
        rule {\bf R1} ensures that $A_2(c) \in \Umc_{\Dmc,\Omc}$ and the
        construction of $\Umc_{\Dmc,\Omc}$ yields
        $(t,c) \in R^{\Umc_{\Dmc, \Omc}}$, thus we are done. Now assume
        that $\Dmc_{M_c} \not\models A_2(c)$ or $R^- \notin \rho$. Since
        $A_1 \in M$, there must be some set of roles $\rho'$ and of
        concept names $M'$ such that $R \in \rho$, $A_2 \in M$, and
        $\Dmc_{M},\Omc \models q^1_{\rho',M'}(\widehat c)$.  We may assume
        w.l.o.g.\ that $\rho'$ and $M'$ are maximal with this
        property. Thus $M \rightsquigarrow^{\rho'}_{\Dmc,\Omc} M'$. Using
        the fact that
        $\Dmc_{M_c},\Omc \not\models A_2(c)$ or $R^- \notin \rho$ one
        can
        prove that Condition~8 of traces is satisfied for $t'=t \rho' M'$,
        based on a straightforward model-theoretic argument. Since all
        other conditions of traces are trivially satisfied for $t'$, we
        obtain $t' \in \Tbf$ and the construction of $\Umc_{\Dmc, \Omc}$
        yields $(t,t') \in R^{\Umc_{\Dmc, \Omc}}$ and
        $t' \in A_2^{\Umc_{\Dmc, \Omc}}$.

        \smallskip

        Now assume that $t=t' \rho M$ with $t' = t'' \rho' M'$.  Then
        $A_1 \in M$.  If $A_2 \in M'$ and $R^- \in \rho$, then
        $(t,t') \in R^{\Umc_{\Dmc, \Omc}}$ and
        $t' \in A_2^{\Umc_{\Dmc, \Omc}}$ and we are done. Now assume that
        $A_2 \notin M'$ or $R^- \notin \rho$. Since $A_1 \in M$, there
        must be some set of roles $\widehat \rho$ and of concept names
        $\widehat M$ such that $R \in \widehat \rho$,
        $A_2 \in \widehat M$, and
        $\Dmc_{M},\Omc \models q^1_{\widehat \rho,\widehat M}(\widehat
        c)$.  We may assume w.l.o.g.\ that $\widehat \rho$ and
        $\widehat M$ are maximal with this property. Thus
        $M \rightsquigarrow^{\widehat \rho}_{\Dmc,\Omc} \widehat M$. The
        fact that $A_2 \notin M'$ or $R^- \notin \rho$
        can be used to prove model-theoretically that Condition~8 of
        traces is satisfied for $\widehat t=t \widehat\rho \widehat
        M$. Since all other conditions of traces are trivially satisfied
        for $\widehat t$, we obtain $\widehat t \in \Tbf$ and the
        construction of $\Umc_{\Dmc, \Omc}$ yields
        $(t,\widehat t) \in R^{\Umc_{\Dmc, \Omc}}$ and
        $\widehat t \in A_2^{\Umc_{\Dmc, \Omc}}$.

      \item $\exists R . A_1 \sqsubseteq A_2 \in \Omc$.  Let
        $(t,t') \in R^{\Umc_{\Dmc, \Omc}}$ and
        $t' \in A_1^{\Umc_{\Dmc, \Omc}}$. We distinguish several cases.

        \smallskip \emph{$t,t'$ are both from $\mn{adom}(\Dmc)$}, then
        Rule~{\bf R2} yields $t \in A_2^{\Umc_{\Dmc, \Omc}}$, as required.

        \smallskip \emph{$t$ is from $\mn{adom}(\Dmc)$ and
              $t' = t \rho M$.} Then $A_1 \in M$, $R \in \rho$, and
        $\Dmc_{M_c},\Omc \models q^1_{\rho,M} (t')$. Consequently,
        $\Dmc_{M_c},\Omc \models A_2(t)$, implying $t \in A_2^{\Umc_{\Dmc, \Omc}}$.

        \smallskip \emph{$t'$ is from $\mn{adom}(\Dmc)$ and
              $t= t' \rho M$.}  Then $\Dmc,\Omc \models A_1(t')$,
        $R^- \in \rho$, and $\rho$ and $M$ are maximal with
        $\Dmc,\Omc \models q^1_{\rho,M} (t)$.  Thus clearly $A_2 \in M$,
        implying $t \in A_2^{\Umc_{\Dmc, \Omc}}$ due to Rule ${\bf R1}$.

        \smallskip \emph{$t,t'$ are both not from $\mn{adom}(\Dmc)$ and
              $t'=t \rho' M'$.} Let $t=c \rho M$ or $t=M_0 \rho M$. Then
        $A_1 \in M'$, $R \in \rho'$, and
        $\Dmc_{M},\Omc \models q^1_{\rho',M'} (\widehat c)$. It is easy to
        see that this implies
    $\Dmc_M,\Omc \models A_2(\widehat c)$.  Since $M$ is maximal
        with $\Dmc_{M_c},\Omc \models q^1_{\rho,M}(c)$ or
        $\Dmc_{M_0} \models q^1_{\rho,M}(\widehat c)$, it follows that
        $A_2 \in M$. This implies that $t \in A_2^{\Umc_{\Dmc, \Omc}}$.

        \smallskip \emph{$t,t'$ are both not from $\mn{adom}(\Dmc)$ and
              $t = t' \rho M$.} Let $t'=c \rho' M'$ or $t'=M_0 \rho' M'$.
        Then $A_1 \in M'$, $R^- \in \rho$, and $M$ is maximal with
        $\Dmc_{M'},\Omc \models q^1_{\rho,M} (\widehat c)$. Consequently,

        $A_2 \in M$, implying $t \in A_2^{\Umc_{\Dmc, \Omc}}$.

      \end{itemize}
  We next show that the RIs in \Omc are satisfied. Thus let
  $R \sqsubseteq S \in \Omc$. If $R(c,c') \in \Umc_{\Dmc, \Omc}$ with
  $c,c' \in \mn{adom}(\Dmc)$, then $R(c,c') \in \mn{ch}_\Omc(\Dmc)$
  and Rule~{\bf R3} yields $(c,c') \in S^{\Umc_{\Dmc, \Omc}}$. It
  remains to deal with the cases
  $(c,c{\rho}M) \in R^{\Umc_{\Dmc, \Omc}}$,
  $(t,t{\rho}M) \in R^{\Umc_{\Dmc, \Omc}}$, and
  $(t{\rho}M,t) \in R^{\Umc_{\Dmc, \Omc}}$. In the first two cases, we
  must have $R \in \rho$, implying $S \in \rho$ by the maximality
  conditions, and thus also $(c,c{\rho}M) \in S^{\Umc_{\Dmc, \Omc}}$
  and $(t,t{\rho}M) \in S^{\Umc_{\Dmc, \Omc}}$, respectively. In the
  last case, we must have $R^- \in \rho$. Clearly
  $R\sqsubseteq S \in \Omc$ implies
  $\Omc \models R^- \sqsubseteq S^-$. We can thus argue similarly that
  $(t{\rho}M,t) \in S^{\Umc_{\Dmc, \Omc}}$.

  It remains to show that all functionality assertions are
  satisfied. We first note that it is straightforward to show the
  following, essentially by induction on the number of rule
  applications used to construct $\mn{ch}_\Omc(\Dmc)$:
  \begin{itemize}
      \item[(i)] $c \in A^{\mn{ch}_\Omc(\Dmc)}$ implies
        $\Dmc,\Omc \models A(c)$;

      \item[(ii)] $(c,c') \in R^{\mn{ch}_\Omc(\Dmc)}$ implies
        $\Dmc,\Omc \models R(c,c')$.

      \end{itemize}
  We may replace the precondition in (ii) with the equivalent
  condition $(c,c') \in R^{\Umc_{\Dmc,\Omc}}$ (and likewise for (i),
  but this is not going to be important).

  Assume to the contrary of what we want to show that there are
  $R(t,t_1),R(t,t_2) \in \Umc_{\Dmc,\Omc}$ with $\mn{func}(R) \in \Omc$ and
  \mbox{$t_1 \neq t_2$}.  We distinguish several cases:

  \smallskip
  $t, t_1, t_2$ are all from $\mn{adom}(\Dmc)$. It then follows
  from (ii) (with precondition $(c,c') \in R^{\Umc_{\Dmc,\Omc}}$)
  that \Dmc is not satisfiable w.r.t.\ \Omc, a
  contradiction.

  \smallskip $t, t_1$ are from $\mn{adom}(\Dmc)$ and
  $t_2=t\rho_2M_2$. Then
  \mbox{$t \rightsquigarrow^{\rho_2}_{\Dmc,\Omc} M_2$} and
  $R \in \rho_2$. Since $R(t,t_1) \in \Umc_{\Dmc,\Omc}$, also
  $R(t,t_1) \in \mn{ch}_\Omc(\Dmc)$.  Together
  with \mbox{$\mn{func}(R) \in \Omc$}, we thus obtain a contradiction
  against Condition~3 of
  $t \rightsquigarrow^{\rho_2}_{\Dmc,\Omc} M_2$.

  \smallskip $t_1=t\rho_1M_1$ , and $t_2=t\rho_2M_2$.  Then
  $t \rightsquigarrow^{\rho_1}_{\Dmc,\Omc} M_1$ and
  $t \rightsquigarrow^{\rho_2}_{\Dmc,\Omc} M_2$. Moreover,
  $R \in \rho_1 \cap \rho_2$. The maximality conditions thus yield
  $\rho_1 = \rho_2$ and $M_1 = M_2$, in contradiction to
  $t_1 \neq t_2$.

  % \smallskip $t \notin \mn{adom}(\Dmc)$, $t_1=t\rho_1M_1$ , and $t_2=t\rho_2M_2$.
  % As previous case, but using maximality condition~5.

  \smallskip $t=t_1\rho_1M_1$ , and $t_2=t\rho_2M_2$. Then
  $R^- \in \rho_1$ and $R \in \rho_2$, contradicting Condition~8
  of paths for $t_2$.

  %   \smallskip $t_1 \in \mn{adom}(\Dmc)$, $t=t_1\rho_1M_1$ , and
  % $t_2=t\rho_2M_2$. Then
  % $t_1 \rightsquigarrow^{\rho_1}_{\Dmc,\Omc} M_1$ and
  % $t \rightsquigarrow^{\rho_2}_{\Dmc,\Omc} M_2$. Moreover,
  % $R^- \in \rho_1$ and $R \in \rho_2$. An easy semantic argumenbt and
  % maximality condition~2 then yield $M_2 \subseteq M_1$.

  \medskip
  For Point~2, let \Imc be a model of \Dmc and \Omc.
%  Let $\Jmc = \Umc_{\Dmc, \Omc}$.
  We construct a homomorphism from $\Umc_{\Dmc, \Omc}$ to \Imc that is
  the identity on $\mn{adom}(\Dmc)$ in two steps.

  We first show that $\mn{ch}_{\Omc}(\Dmc) \subseteq \Imc$.  This
  gives us a homomorphism $h_0$ from
  $\Umc_{\Dmc, \Omc}|_{\mn{adom}(\Dmc)} = \mn{ch}_{\Omc}(\Dmc)$
  to \Imc that is the identity on $\mn{adom}(\Dmc)$. Then, we extend $h_0$
  to a homomorphism from $\Umc_{\Dmc, \Omc}$ to \Imc using
  Lemma~\ref{lem:newintermediate}.

  Let $\Dmc = \Nmc_0, \Nmc_1, \dots, \Nmc_\ell = \mn{ch}_{\Omc}(\Dmc)$
  be the databases obtained by applying rules
  \textbf{R1}\=/\textbf{R4} in the construction of the chase.  To show
  that $\mn{ch}_{\Omc}(\Dmc) \subseteq \Imc$, it suffices to show that
  $\Nmc_i \subseteq \Imc$ for all $0 \leq i \leq \ell$.  The base
  case, $i=0$, is trivial as $\Imc$ is a model of $\Dmc= \Nmc_0$.

  For the inductive step, let us assume that $\Nmc_i \subseteq \Imc$, $i < \ell$. We show that $\Nmc_{i+1} \subseteq \Imc$.
  By definition, $\Nmc_{i+1}$ was constructed from $\Nmc_i$ by applying one of the rules \textbf{R1}\=/\textbf{R4}.

  Let $\textbf{R1}$ be the applied rule. It has added a single new fact $A(c)$.
  Since the rule has been used, there are facts $A_1(c),\dots,A_n(c) \in \Nmc_i$.
  By the inductive assumption we have that $A_1(c),\dots,A_n(c) \in \Imc$.
  Since $\Imc$ is a model of $\Omc$ and $\Omc \models A_1 \sqcap \cdots \sqcap A_n \sqsubseteq A$
  we infer that $A(c) \in \Imc$. This implies that $\Nmc_{i+1} = \Nmc_i \cup \{A(c)\} \subseteq \Imc$
  and ends the inductive step for rule \textbf{R1}.
%
  % The inductive step for the remaining rules is proven in the same manner.
%
  % We have thus shown that for all $0 \leq i \leq \ell$ we have that $\Nmc_i \subseteq \Imc$.
  % Since $\Nmc_\ell = \mn{ch}_{\Omc}(\Dmc)$ we infer that $\mn{ch}_{\Omc}(\Dmc) \subseteq \Imc$
  % which
 This ends the first step of the proof. Let $h_0$ be the identity homomorphism from $\mn{ch}_{\Omc}(\Dmc)$
  to $\Imc$.

\smallskip

For the second step note that $\mn{ch}_\Omc(\Dmc) \subseteq \Imc$
implies that $(\mn{ch}_\Omc(\Dmc),c) \preceq (\Imc,c)$ for every
$c \in \dom(\Dmc)$. It thus follows from Lemmas~\ref{lem:sim}
and~\ref{lem:newintermediate} that there is a homomorphism $h_c$ from
$\Umc_{\Dmc, \Omc}^{\downarrow c}$ to $\Imc$.  We show that
$h = h_0 \cup \bigcup_{c \in \dom(\Dmc)} h_c$ is the desired
homomorphism from $\Umc_{\Dmc, \Omc}$ to \Imc.

  The function $h$ is well defined as $h_0$ is identity on $\mn{adom}(\Dmc)$ and for all
  $c \in \mn{adom}(\Dmc)$ we have that $h_c(c)=c$. This also implies that
  $h$ is identity on $\mn{adom}(\Dmc)$.

  To show that $h$ is a homomorphism observe the following.
  For concept names, if $A(d) \in \Umc_{\Dmc, \Omc}$ then $h(d) = h_0(d)$
  or $h(d) = h_c(d)$ for some $c \in \dom(\Dmc)$.
  Thus, $A(h(c)) \in \Imc$ as $h_0$ and $h_c$ are homomorphisms.

  Similarly, for role names, if $R(c,d)$ is a fact in $\Umc_{\Dmc, \Omc}$
  then either $c,d \in \dom(\Dmc)$ and $h(c)=h_0(c),h(d)=h_0(d)$
  or there is $c' \in \dom(\Dmc)$ such that $c$ and $d$ are traces in $\Umc_{\Dmc, \Omc}^{\downarrow c}$.
  In the former case we have that $R(h(c), h(d)) \in \Imc$ as $h_0$ is a homomorphism from $\mn{ch}_\Omc(\Dmc)$ to $\Imc$,
  and in the latter we have that $R(h(c), h(d)) \in \Imc$ as $h_c$ is a homomorphism from $\Umc_{\Dmc, \Omc}^{\downarrow c}$ to $\Imc$.

  \medskip Finally, Point~3 is a consequence of Point~2 and the
  definitions of $Q(\Dmc)$ and $q(\Umc_{\Dmc,\Omc})$.
\end{proof}

\subsection{Query-Directed Universal Models}
\label{sect:univext}

Let $Q=(\Omc,\Sigma,q) \in (\ourEL,\text{CQ})$ with \Omc in normal
form, and let $n=|\mn{var}(q)|$. We use $\mn{cl}(Q)$ to denote the set of
all Boolean tree CQs that use at most $n$ variables, only symbols
from~$\mn{sig}(\Omc) \cup
\mn{sig}(q)$, % and {\color{orange} no constants},
and satisfy all functionality assertions in~\Omc. We may assume that
the CQs in $\mn{cl}(Q)$ use only variables from a fixed set of size $n$, and
thus $\mn{cl}(Q)$ is finite.

Let $\Dmc$ be a $\Sigma$-database that is satisfiable w.r.t.\ \Omc.
The \emph{$Q$-universal extension of \Dmc}, denoted $\Umc_{\Dmc,Q}$ is
constructed in three steps as follows.

\smallskip {\bf Step~1.} Start with $\mn{ch}_\Omc(\Dmc)$.

\smallskip {\bf Step~2.} Extend $\mn{ch}_\Omc(\Dmc)$ to the
restriction of $\Umc_{\Dmc,\Omc}$ to traces of length $n$.
Informally, this attaches trees of depth $n$ to constants
$c \in \mn{adom}(\Dmc)$. This can be done following exactly the
definition of traces, and it involves deciding questions of the form
$\Dmc_M,\Omc \models q^1_{\rho,M'}(c)$.

\smallskip {\bf Step~3.}  For every $c \in \mn{adom}(\Dmc)$, let again
$M_c$ denote the set of concept names $A$ with
$A(c) \in \mn{ch}_\Omc(\Dmc)$.  Then include, for every $p \in \mn{cl}(Q)$
such that $\Dmc_{M_c},\Omc \models p$ for some
$c \in \mn{adom}(\Dmc)$, a copy of $\Dmc_p$ that uses only fresh
constants.

\smallskip
It should be clear that $\Umc_{\Dmc,Q}$ satisfies all functionality
assertions in \Omc.
\begin{lemma}
  \label{lem:pseudochase}
  Let $Q(\bar x)=(\Omc,\Sigma,q) \in (\mathcal{ELIHF},\text{CQ})$ and
  let $\Dmc$ be a $\Sigma$-database. Then
  $Q(\Dmc)=q(\Umc_{\Dmc,Q}) \cap \mn{adom}(\Dmc)^{|\bar x|}$.  What is
  more, for every $Q'(\bar x')=(\Omc,\Sigma,q')$ where
  $|\mn{var}(q')| \leq |\mn{var}(q)|$, we have
  $Q'(\Dmc)=q'(\Umc_{\Dmc,Q}) \cap \mn{adom}(\Dmc)^{|\bar x|}$.
\end{lemma}
\begin{proof}
  % Let $Q(\bar x)=(\Omc,\Sigma,q)$, $D$, $q'$, and $\bar c$ be as in
  % Lemma~\ref{lem:pseudochase}. We have to show that $D \cup \Omc
  % \models q'(\bar c)$ iff $\bar c \in q'(\mn{ch}^q_\Omc(D))$.
  It clearly suffices to prove the `what is more' part of the lemma.

  \medskip

  The `$\supseteq$' direction is simple. Assume that
  $\bar c \in q'(\Umc_{\Dmc,Q})  \cap \mn{adom}(\Dmc)^{|\bar x'|}$.
  Using the construction of $\Umc_{\Dmc,Q}$, it is not hard to see
  that there is a homomorphism from $\Umc_{\Dmc,Q}$ to
  $\Umc_{\Dmc,\Omc}$ that is the identity on $\mn{adom}(\Dmc)$. Thus
  $\bar c \in q'(\Umc_{\Dmc,\Omc})$ and it follows from Point~3 of
  Lemma~\ref{lem:univbasic} that $\bar c \in Q'(\Dmc)$.

  \smallskip For the `$\subseteq$' direction, assume that
  \mbox{$\bar c \in Q'(\Dmc)$}. Then there is a homomorphism
  $h$ from $q'$ to $\Umc_{\Dmc,\Omc}$ such that
  \mbox{$h(\bar x')=\bar c$}. Let $q_1,\dots,q_k$ be the maximal
  connected components of $q'$ such that $h(y) \notin \mn{dom}(\Dmc)$
  for all variables $y$ that occur in them. Note that by definition,
  $q_1,\dots,q_k$ are Boolean.  Further let $q_0$ be the disjoint
  union of all remaining maximal connected components of $q'$.

  Let $n=|\mn{var}(q)|$ and let $\Umc_{\Dmc,\Omc,n}$ denote the
  restriction of $\Umc_{\Dmc,\Omc}$ to traces of length at most
  $n$. It is not hard to verify that
  $h(y) \in \mn{adom}(\Umc_{\Dmc,\Omc,n})$ for all
  $y \in \mn{var}(q_0)$, and thus the restriction $h'$ of $h$ to the
  variables in $q_0$ is a homomorphism from $q_0$ to
  $\Umc_{\Dmc,\Omc,n} \subseteq \Umc_{\Dmc,Q}$. To show that
  $\bar c \in q'(\Umc_{\Dmc,Q}) \cap \mn{adom}(\Dmc)^{|\bar x|}$, it
  remains to extend $h'$ to $q_1,\dots,q_k$.

  Let $1 \leq i \leq k$ and let $q'_i$ be the CQ obtained from $q_i$
  by identifying variables $x,y \in \mn{var}(q)$ if
  $h(x)=h(y)$. Clearly, it suffices to show that we can extend $h'$ to
  $q'_i$ in place of $q_i$. Since $h(y) \notin \mn{adom}(\Dmc)$ for
  all $y \in \mn{var}(q_i)$ and due to the shape of
  $\Umc_{\Dmc,\Omc}$, the CQ $q'_i$ must be a tree CQ. Consequently,
  it is in $\mn{cl}(Q)$.

  Since the restriction $h_i$ of $h$ to the variables in $q'_i$ is a
  homomorphism from $q'_i$ to $\Umc_{\Dmc,\Omc}$ that does not have
  any constants from $\mn{adom}(\Dmc)$ in its range, there is an
  $e_i \in \mn{adom}(\Dmc)$ such that the range of $h_i$ consists only
  of traces that start with~$e_i$. The construction of
  $\Umc_{\Dmc,\Omc}$ then yields $\Dmc_{M_{e_i}},\Omc \models
  q_i()$. Consequently, a copy $\Dmc'_{q_i}$ of $\Dmc_{q_i}$ was added
  during the construction of $\Umc_{\Dmc,Q}$.  It is thus easy to
  extend $h'$ to $q'_i$, as desired.
  % now easy to construct a
  % homomorphism $h'$ from $q$ to $\Umc_{\Dmc,Q}$ with
  % $h'(\bar x)=\bar c$, showing that
  % $\bar c \in q(\Umc_{\Dmc,Q}) \cap \mn{adom}(\Dmc)^{|\bar x|}$ as
  % required.
\end{proof}

We now show that the $Q$-universal database extension $\Umc_{\Dmc,Q}$
is also universal for answers with wildcards when we view all elements
of $\mn{adom}(\Umc_{\Dmc,Q}) \setminus \mn{adom}(\Dmc)$ as nulls.
\begin{restatable}{lemma}{proprestrictedchaseworks}
  \label{prop:restrictedchaseworks}
%  Let $q(\bar x)$ be a CQ with at most $n$ variables and
 % $Q=(\Omc,\Sigma,q)$.  Then
 % $Q(\Dmc)=q(\Umc_{\Omc,n}) \cap \mn{adom}(\Dmc)^{|\bar x|}$,
  ~\\[-4mm]
  \begin{enumerate}
      \item
      $Q(\Dmc)^{\ast}= q(\Umc_{\Dmc,Q})^{\ast}_\Nbf$ and
      \item
      $Q(\Dmc)^{\Wmc}= q(\Umc_{\Dmc,Q})^{\Wmc}_\Nbf$.
      \end{enumerate}
\end{restatable}
\begin{proof}
  We only consider Point~2. The proof of Point~1 is similar. When we
  speak of partial answers, we thus generally mean partial answers
  with multi-wildcards. We first observe that it suffices to prove the
  following:
  \begin{itemize}

      \item[(a)] every minimal partial answer to $Q$
        on \Dmc is a partial answer to $q$ on $\Umc_{\Dmc,Q}$;

      \item[(b)] every partial answer to $q$ on $\Umc_{\Dmc,Q}$
        if a partial answer to $Q$ on \Dmc.

      \end{itemize}
  For assume that (a) and (b) have been shown. First let
  $\bar a^\Wmc \in Q(\Dmc)^\Wmc$. By~(a), $\bar a^\Wmc$ is a partial
  answer to $q$ on $\Umc_{\Dmc,Q}$. Assume to the contrary of what we
  have to show that $\bar a^\Wmc$ is not minimal, that is, there is a
  partial answer $\bar b^\Wmc$ to $q$ on $\Umc_{\Dmc,Q}$ with
  $\bar b^\Wmc \prec \bar a^\Wmc$.  By~(b), $\bar b^\Wmc$ is
  a partial answer to $Q$ on \Dmc, contradicting minimality of $\bar
  a^\Wmc$.

  Conversely, let $\bar a^\Wmc \in q(\Umc_{\Dmc,Q})^{\Wmc}_\Nbf$.
  By~(b), $\bar a^\Wmc$ is a partial answer to $Q$ on \Dmc. Assume to
  the contrary of what we have to show that $\bar a^\Wmc$ is not
  minimal, that is, there is a partial answer $\bar b^\Wmc$ to $Q$ on
  $\Dmc$ with $\bar b^\Wmc \prec \bar a^\Wmc$. We may assume w.l.o.g.\
  that $\bar b^\Wmc$ is minimal. Thus~(a) implies that   $\bar b^\Wmc$ is
  a partial answer to $q$ on $\Umc_{\Dmc,Q}$, contradicting minimality of $\bar
  a^\Wmc$.

  \medskip

  It thus remains to prove (a) and (b). We start with the former.
  Assume that $\bar a^\Wmc \in Q(\Dmc)^\Wmc$.  Let
  $\bar x = x_1,\dots,x_n$, let $\bar a^\Wmc = a_1,\dots,a_n$, and let
  the wildcards from \Wmc that occur in $\bar a^\Wmc$ be
  $\ast_1,\dots,\ast_\ell$. Consider the CQ $q'(\bar x')$ obtained
  from $q$ in the following way:
  \begin{itemize}

      \item introduce fresh quantified variables $z_1,\dots,z_\ell$;

      \item if $a_i=\ast_j$, then replace in $q'$ the answer variable
        $x_i$ with quantified variable $z_j$.

      \end{itemize}
  Further let $\bar a$ be obtained from $\bar a^\Wmc$ by removing all
  wildcards.

  By a simple semantic argument, it follows from the fact that
  $\bar a^\Wmc$ is a partial answer to $Q$ on \Dmc that
  $\bar a \in Q'(\Dmc)$ where $Q'(\bar x')=(\Omc,\Sigma,q')$. Note
  that $|\mn{var}(q')| \leq |\mn{var}(q)|$. Applying Lemma~\ref{lem:pseudochase} thus
  yields
  $\bar a \in q'(\Umc_{\Dmc,Q}) \cap \mn{adom}(\Dmc)^{|\bar x'|}$.
  Consequently, there is a homomorphism $h'$ from $q'$ to
  $\Umc_{\Dmc,Q}$ such that $h'(\bar x')=\bar a$. We observe the
  following:
  \begin{description}

      \item[($*$)] $h'(z_i) \notin \mn{adom}(\Dmc)$ for $1 \leq i \leq \ell$.

      \end{description}
  For assume to the contrary that there is a $z_p$ with
  $h'(z_p) \in \mn{adom}(\Dmc)$. % Let $d_1,\dots,d_g$ be all
  % elements of $\mn{adom}(\Umc_{\Dmc,Q}) \setminus  \mn{adom}(\Dmc)$
  % that are in the range of $h$.
  Consider
  \begin{itemize}

      \item the CQ $q''(\bar x'')$ obtained from $q$ in exactly the same
        way as $q'$ except that if $a_i=\ast_p$, then $x_i$ is \emph{not}
        replaced with~$z_p$,

      \item the tuple $\bar b^\Wmc$ obtained from $\bar a^\Wmc$ in the
        following way: if $a_i = \ast_p$, % , and thus $x_i$ was replaced with
        % $z_p$ during the construction of $q'$
        then replace $a_i$ with
        $h'(z_p)$, and

      \item the tuple $\bar b$ obtained from $\bar b^\Wmc$ by removing
        all wildcards.

      \end{itemize}
  The homomorphism $h'$ witnesses that
  $\bar b \in q''(\Umc_{\Dmc,Q}) \cap \mn{adom}(\Dmc)^{|\bar x''|}$.
  From Lemma~\ref{lem:pseudochase} we thus obtain
  $\bar b \in Q''(\Dmc)$ where $Q''(\bar x'')=(\Omc,\Sigma,q'')$.
  An easy semantic argument shows that, consequently, $\bar b^\Wmc$
  is a partial answer to $Q$ on \Dmc, contradicting the minimality of
  $\bar a^\Wmc$.
%
  % the tuple
  % $\bar b^\Wmc \in (\mn{adom}(\Dmc) \cup \Wmc)^{|\bar x|}$ obtained
  % from $\bar a^\Wmc$ in the following way: if $a_i = \ast_j$, and thus
  % $x_i$ was replaced with $z_j$ during the construction of $q'$, then:
  % %
  % \begin{itemize}
    %
      % \item if $h'(z_j) \in \mn{adom}(\Dmc)$, then replace $a_i$ with
      %   $h'(z_j)$;
    %
      % \item if $h'(z_j) = d_f \notin \mn{adom}(\Dmc)$, then replace $a_i$
      %   with $\ast_f$.
    %
      % \end{itemize}
  % %
  % Finally, rename the wildcards so that they occur in $\bar b^\Wmc$ in
  % ascending order from left to right. It is easy to see that
  % $\bar b^\Wmc \preceq \bar a^\Wmc$ and in fact,
  % $h'(z_p) \in \mn{adom}(\Dmc)$ and the construction of $\bar b^\Wmc$
  % implies $\bar b^\Wmc \prec \bar a^\Wmc$. Moreover, we may use
  % $h'$ to show that $\bar b^\Wmc$ it can be
  % This contradicts the
  % minimality of $\bar a^\Wmc$ and
  Thus ($*$) is shown.

  \smallskip We may obtain from $h'$ a homomorphism $h$ from $q$ to
  $\Umc_{\Dmc,Q}$ by setting $h(x_i)=z_j$ whenever $x_i$ was replaced
  by $z_j$ in the construction of $q'$. Using $h$ and ($*$), it is now
  easy to show that $\bar a^\Wmc$ is a partial answer to $q$ on
  $\Umc_{\Dmc,Q}$, as required.

  \medskip For~(b), let $\bar a^\Wmc$ be a partial answer to $q$ on
  $\Umc_{\Dmc,\Omc}$. Let $q'$ and $\bar a$ be constructed exactly as
  above. Then
  $\bar a \in q'(\Umc_{\Dmc,Q}) \cap \mn{adom}(\Dmc)^{|\bar x'|}$.
  Lemma~\ref{lem:pseudochase} yields $\bar a \in Q'(\Dmc)$ where again
  $Q'(\bar x')=(\Omc,\Sigma,q')$. An easy semantic argument shows
  that, consequently, $\bar a^\Wmc$ is a partial answer to $Q$ on \Dmc,
  as required.
\end{proof}

\subsection{Query-Directed Universal Models in Linear Time}

\begin{restatable}{proposition}{propchaseinlineartime}
  \label{prop:chaseinlineartime} Let
  $Q=(\Omc,\Sigma,q) \in (\ourEL,\text{CQ})$ and
  $\Dmc$ a
  $\Sigma$-database that is satisfiable w.r.t.\ \Omc. Then the
  $Q$-universal extension
  $\Umc_{\Dmc,Q}$ can be computed in time linear in
  $||\Dmc||$, more precisely in time
$2^{\mn{poly}(||Q||)}\cdot
||\Dmc||$.
%$2^{O(||\Omc||^{n})} \cdot ||\Dmc||$.
\end{restatable}
To prove Proposition~\ref{prop:chaseinlineartime}, we make use of the
fact that minimal models for propositional Horn formulas can be
computed in linear time~\cite{dowling-gallier-horn}. More precisely,
we derive a satisfiable propositional Horn formula $\theta$ from
$\Omc$ and \Dmc, compute a minimal model of $\theta$ in linear
time, and read off $\mn{ch}_\Omc(\Dmc)$ from that model. It then
remains to construct $\Umc_{\Dmc,Q}$ from $\mn{ch}_\Omc(\Dmc)$,
exactly as described in the previous section.

% Let $\Qmc(\Omc,n)$ be the class of CQs defined
% during the definition of $\Umc_{\Dmc,\Omc,n}$.
An \emph{edge} in \Dmc is a pair $(c_1,c_2)$ such that
$R(c_1,c_2) \in \Dmc$ for some role $R$. Note that if $(c_1,c_2)$ is
an edge in \Dmc, them so is $(c_2,c_1)$. We introduce the following
propositional variables:
\begin{itemize}

% \item $x_{p()}$ for every Boolean $p \in \Qmc(\Omc,n)$;

% \item $x_{p(c)}$ for every unary $p \in \Qmc(\Omc,n)$ and every
%   $c \in \mn{adom}(\Dmc)$;

% \item $x_{R(c,c')}$ for every edge $\{c,c'\}$  in \Dmc and every role
%   $R$ in $\Sigma \cup \mn{sig}(\Omc)$.

%  \item $x_{p()}$ for every Boolean $p \in \Qmc(\Omc,n)$;

\item $x_{A(c)}$ for every concept name $A \in \mn{sig}(\Omc)$ and
  every $c \in \mn{adom}(\Dmc)$;

\item $x_{r(c,c')}$ for every edge $(c,c')$  in \Dmc and every role
  name $r \in \mn{sig}(\Omc)$.

\end{itemize}
We may write $x_{r^-(c',c)}$ on place of $x_{r(c,c')}$.  Clearly, the
number of introduced variables is bounded by $||\Dmc||\cdot||\Omc||$.
% Recall that
% the CQs in $\Qmc(\Omc,n)$ are tree CQs of degree at most $||\Omc||$
% and depth at most $n$ that use only symbols from~\Omc. There are at
% most $2^{O(||\Omc||^{n})}$ such CQs, and thus the cardinality of
% $\Qmc(\Omc,n)$ is bounded by this number. It follows that the number
% of variables introduced is bounded by
% $2^{O(||\Omc||^{n})} \cdot ||\Dmc||$.

The propositional Horn formula $\theta$ consists of the following conjuncts:
\begin{enumerate}

\item $x_{A(c)}$ for every $A(c) \in \Dmc$ and $x_{r(c,c')}$ for every
  \mbox{$r(c,c') \in \Dmc$};

\item
  $x_{A_1(c)} \wedge \cdots \wedge x_{A_n(c)} \rightarrow x_{A(c)}$
  for all $c \in \mn{adom}(\Dmc)$ and all concept names $A_1,\dots,A_n,A$
  such that $\Omc \models A_1 \sqcap \cdots \sqcap A_n \sqsubseteq A$;

\item $x_{A(c_1)} \wedge x_{R(c_2,c_1)} \rightarrow x_{B(c_2)}$
  for all edges $(c_1,c_2)$ of \Dmc and all
  $\exists R . A \sqsubseteq B \in \Omc$;

\item $x_{R(c_1,c_2)} \rightarrow x_{S(c_1,c_2)}$ for all edges
$(c_1,c_2)$ of \Dmc and all $R \sqsubseteq S \in \Omc$;

\item $x_{A(c_1)} \wedge x_{R(c_1,c_2)} \rightarrow x_{S(c_1,c_2)}$
  and
  $x_{A(c_1)} \wedge x_{R(c_1,c_2)} \rightarrow x_{B(c_2)}$
  for all $A \sqsubseteq \exists S . B \in \Omc$ such that
  $S \sqsubseteq_\Omc^* R$ and $\mn{func}(R) \in \Omc$.

% \item $\bigwedge_{B(e) \in \Dmc'} x_{B(e)} \rightarrow A(c)$ for all
%   databases $\Dmc'$, concept names $A$, and constants
%   $c \in \mn{adom}(\Dmc')$ such that
%   %
%   \begin{enumerate}

    %   \item  $\Dmc'$ uses only symbols from
    %     \Omc,

    %   \item $|\mn{adom}(\Dmc')| \leq 2$, and

    %   \item $\Dmc',\Omc \models A(c)$.
    %   \end{enumerate}

% \item {\color{blue}also need to compute roles!!}

% \item
%   $\bigwedge_{S(\bar d) \in D'} x_{S(\bar d)} \rightarrow x_{p(\bar
    %     c)}$ for every $\mn{sig}(\Omc)$-database $D'$, every CQ
%   $p(\bar y) \in \mn{cl}(Q)$, and every $\bar c \in
%   \mn{adom}(D')^{|\bar y|}$ such that $D' \cup \Omc \models p(\bar c)$
%   and
%   $\mn{adom}(D')$ is a guarded set $S$ in $D$.

\end{enumerate}
The size of $\theta$ is bounded by
$2^{O(||\Omc||)} \cdot ||\Dmc||$ and $\theta$ can be constructed
in time \mbox{$2^{O(||\Omc||)} \cdot ||\Dmc||$}.

Since $\theta$ contains no negative literals, it is clearly
satisfiable and thus has a unique minimal model. Let $V$ be the truth
assignment that represents this minimal model. We construct the
database $\Dmc_\theta$ by including all $A(c)$ such that
$V(x_{A(c)})=1$ and all $r(c_1,c_2)$ such that
$V(x_{r(c_1,c_2)})=1$. It is clear that the construction of
$\Dmc_\theta$ is possible in time $2^{O(||\Omc||)} \cdot
||\Dmc||$. Moreover, it is easy to verify that
$\Dmc_\theta = \mn{ch}_\Omc(\Dmc)$ given that the implications in
$\theta$ exactly parallel the rules of the chase given in
Section~\ref{sect:univmod}.

We may now construct $\Umc_{\Dmc,Q}$ from $\mn{ch}_\Omc(\Dmc)$ exactly
as in the previous section, using Steps~2 and~3 described there. In
Step~2, we may iterate over all $c \in \mn{adom}(\Dmc)$ and then add
the fragment of $\Umc_{\Dmc,\Omc}$ that consists of traces of length
at most $n$ starting with $c$, following exactly the definition of
such traces, that is, proceeding in the order of increasing trace
length. It is easy to see that the number of traces that start with
$c$ is bounded by $2^{\mn{poly}(||Q||)}$. This involves deciding
questions of the form $\Dmc_M,\Omc \models q^1_{\rho,M'}(c)$
which is possible in time $2^{\mn{poly}(||Q||)}$ \cite{DBLP:conf/jelia/EiterGOS08}.
Thus, Step~2 can be implemented in time $2^{\mn{poly}(||Q||)} \cdot ||\Dmc||$.

For Step~3, first note that the number of CQs in $\Qmc$ is bounded by
$2^{\mn{poly}(||Q||)}$. We may iterated over all $c \in \mn{adom}(D)$
and all $p \in \Qmc$ and check whether $\Dmc_{M_c},\Omc \models p$,
and if so, add a disjoint copy of $p$. Again the described check is
possible in time $2^{\mn{poly}(||Q||)}$
\cite{DBLP:conf/jelia/EiterGOS08} and the overall time needed is
$2^{\mn{poly}(||Q||)} \cdot ||\Dmc||$.

\section{Proofs for Section~\ref{sect:upperbounds}}

\lemsamehoms*
\begin{proof}
  % We first prove that $q(\Dmc_0)=q^+(\Dmc_0^+)|_{\bar x}$.
  % It suffices to show that every homomorphism from $q$ to
  % $\Dmc_0$ is also a homomorphism from $q^+$ to $\Dmc_0^+$
  % and vice versa.
%
  First assume that $h$ is a homomorphism from $q_0$ to $\Dmc_0$.  Let
  $R'(\bar y^+)$ be an atom in $q_0^+$, and let $R(\bar y)$ be the
  atom in $q_0$ that gave rise to it during the construction of
  $q_0^+$. The restriction $h|_{\bar y^+}$ of $h$ to the variables in
  $\bar y^+$ is a homomorphism from $q_0|_{\bar y^+}$ to $\Dmc_0$. By
  definition, $\Dmc^+_0$ thus contains the atom $R'(h(\bar y^+))$, as
  required.

  Conversely, assume that $h$ is a homomorphism from $q_0^+$ to
  $\Dmc_0^+$ and let $R(\bar y)$ be an atom in $q_0$. Then $q_0^+$
  contains a corresponding atom $R'(\bar y^+)$ and
  $R'(h(\bar y^+)) \in \Dmc_0^+$. By construction of $\Dmc^+_0$, there
  is a homomorphism $g$ from $q_0|_{\bar y^+}$ to~$\Dmc_0$ that is
  identical to the restriction of $h$ to the variables in $\bar y^+$.
  It follows that $R(h(\bar y)) \in \Dmc_0$, as required.

  \smallskip We sketch the computation of $||\Dmc_0^+||$ in linear
  time. We first construct a lookup table that allows us to find in
  constant time, given a $c \in \mn{adom}(\Dmc_0)$ and a role $R$ with
  $\mn{func}(R) \in \Omc$, the unique $c' \in \mn{adom}(\Dmc_0)$ with
  $R(c,c') \in \Dmc_0$ (if existant). This relies on the use of the
  RAM model and can be achieved by a single
  scan of $\Dmc_0$. We then iterate over every atom $R(\bar y)$ of
  $q_0$.  If $R(\bar y)$ is replaced with $R'(\bar y^+)$ in $q_0^+$, then
  we need to add to $\Dmc_0^+$ a fact $R'(h(\bar y^+))$ for every
  homomorphism $h$ from $q_0|_{\bar y^+}$ to $\Dmc_0$. To find these
  homomorphisms $h$, we consider all facts $R(\bar c) \in \Dmc_0$ as
  candidates for $R(h(\bar y))$. We can find them by a single scan of
  $\Dmc_0$. Every variable $z$ in $\bar y^+$ is reachable from a
  variable $y$ in $\bar y$ by a functional path in $q_0$. Following this
  functional paths in $\Dmc_0$, starting at $h(y)$ and using the
  lookup table, we can find the unique possible target $h(z)$ (if
  existant). It then remains to check whether the constructed $h$ is
  really a homomorphism. This is possible in constant time if we have
  previously created lookup tables that tell us for every $k$-ary
  $R \in \Sigma$ and $\bar c \in \mn{adom}(\Dmc_0)^k$ whether
  $R(\bar c) \in \Dmc_0$.
\end{proof}
We next establish Theorem~\ref{thm:mpaupper} for the multi-wildcard
case.  Let \Dmc be a database that may use nulls in place of
constants, and let $q(\bar x)$ be a CQ.  For an answer
$\bar a \in q(\Dmc)$, we use $\bar a^{\Wmc}_\Nbf$ to denote the
(unique) multi-wildcard tuple for $\Dmc$ obtained from $\bar a$ by
consistently replacing nulls from $\Nbf$ with wildcards using exactly
a prefix of the ordered set $\Wmc = \{ \ast_1, \ast_2,\dots \}$ and
respecting the order, that is, if the first occurrence of a null $c_1$
in $\bar a$ is before the first occurrence of a null $c_2$, $c_1$ is
replaced with $\ast_i$, and $c_2$ with $\ast_j$, then $i < j$.  With
$q(D)^{\Wmc}_N$. We call such an $\bar a^\Wmc_\Nbf$ a \emph{partial
      answer} to $q$ on~\Dmc (with multi-wildcards) and say that it is
\emph{minimal} if there is no $\bar a \in q(\Dmc)$ with
$\bar b^\Wmc_\Nbf \prec \bar a^\Wmc_\Nbf$.
We use $q(D)^{\Wmc}_N$ to denote the set of minimal partial answers
with multi-wildcards to $q$ on $D$.

As already mentioned, In the query-directed universal models
$\Umc_{\Dmc,Q}$ from Proposition~\ref{prop:japanesenoodles}, we may
view all elements of
$N= \mn{adom}(\Umc_{\Dmc,Q}) \setminus \mn{adom}(\Dmc)$ as nulls.  We
then have $Q(\Dmc)^\Wmc=q(\Umc_{\Dmc,Q})^\Wmc_\Nbf$.

To prove Theorem~\ref{thm:mpaupper} in the multi-wildcard case, we
replace Theorem~\ref{prop:enummultiwildcards} with the following
result, also from \cite{DBLP:journals/corr/abs-2203-09288} (Proposition~F.1).
\begin{theorem}
  \label{thm:arxivesecond}
  For every CQ $q(\bar x)$ that is acyclic and free\=/connex ayclic,
  enumerating the answers $q(\Dmc)^\Wmc_\Nbf$ is in \dlc for databases
  $\Dmc$ and sets of nulls $N \subseteq \mn{adom}(\Dmc)$ such that
  $\Dmc$ is chase-like with witness $\Dmc_{1},\dots,\Dmc_{n}$ where
  $|\mn{adom}(\Dmc_{i})|$ does not depend on $\Dmc$ for
  $1 \leq i \leq n$.
\end{theorem}
We now proceed as in the single-wildcard case. Let
$Q(\bar x)=(\Omc,\Sigma,q) \in (\ourEL,\text{CQ})$ with $q^+(\bar x)$
acyclic and free-connex acyclic, and let \Dmc be a $\Sigma$-database.
By Lemma~\ref{prop:restrictedchaseworks}, the query-directed universal
model $\Umc_{\Dmc,Q}$ is also universal for partial answers with
multi-wildcards in the sense that
$Q(\Dmc)^\Wmc=q(\Umc_{\Dmc,Q})^\Wmc_\Nbf$. We thus first replace \Dmc
with $\Umc_{\Dmc,Q}$, aiming to enumerate
$q(\Umc_{\Dmc,Q})^\Wmc_\Nbf$. We next replace $q(\bar x)$
with $q^+(\bar x)$ and $\Dmc_0=\Umc_{\Dmc,Q}$ with $\Dmc^+_0$.  It
follows from Lemma~\ref{lem:samehoms} that
$q(\Dmc_0)^\Wmc_\Nbf=q^+(\Dmc^+_0) ^\Wmc_\Nbf$ and thus we can apply
the algorithm promised by Theorem~\ref{thm:arxivesecond} to enumerate
$q^+(\Dmc^+_0) ^\Wmc_\Nbf$.

% The proof of this theorem is given (almost) completely in the main
% body of the paper. We only add here a remark on the discrepancy
% between our definition of FA-extension and the definition of
% FD-extension from \cite{carmeli-enum-func}. In fact, their definition
% of $q^+$ and ours are not identical. In a nutshell, the definition of
% $q^+$ in \cite{carmeli-enum-func} does not introduce fresh relation
% symbols, but

% This is related to a technical
% glitch in \cite{carmeli-enum-func} rectified in \cite{Erratum}.

% \subsection{Proof of Proposition~\ref{prop:CKwithwildcards}}

%\input{cq-fc}

\section{Proofs for Section~\ref{sect:lower}}

\newcommand{\DGamma}{\mn{adom}(\Dmc_0)}

We provide proofs for Section~\ref{sect:lower}, starting with several
preliminaries.

%%%%%%%%%%%%%Derivation

\subsection{Derivation completeness}
\label{app:sect:dercompl}

Recall that we have defined in the main part of the paper a database
$\Dmc_\omega$ and a finite subset $\Dmc_{\mn{tree}}$. We will
use the latter in all of our lower bound proofs. In this section, we
estbalish some fundamental properties related to it.

We say that a $\Sigma$-database \Dmc is \emph{derivation complete at}
$c \in \mn{adom}(\Dmc)$ if for every $\Sigma$-database $\Dmc'$ that is
satisfiable w.r.t.~\Omc and every $c' \in \mn{adom}(\Dmc')$, there is
a homomorphism $h$ from $\Umc^{\downarrow c'}_{\Dmc',\Omc}$ to
$\Umc_{\Dmc,\Omc}$ with $h(c')=c$. Informally, derivation completeness
at $c$ means that \Omc derives at $c$ anything that it could possibly
derive at any constant in any database. For a set
$\Gamma \subseteq \mn{adom}(\Dmc)$, we say that \Dmc is
\emph{derivation complete} at $\Gamma$ if it is derivation complete at
every $c \in \Gamma$.

%Let \Omc be an \ELIF-ontology in normal form and let $\Sigma$ be a
%signature.
%
% We say that a $\Sigma$-database \Dmc is \emph{derivation complete at}
% $c \in \mn{adom}(\Dmc)$ if for every $\Sigma$-database $\Dmc'$ that is
% satisfiable w.r.t.~\Omc and every $c' \in \mn{adom}(\Dmc')$, there is
% a homomorphism $h$ from $\Umc^{\downarrow c'}_{\Dmc',\Omc}$ to
% $\Umc_{\Dmc,\Omc}$ with $h(c')=c$. Informally, derivation
% completeness at $c$ means that \Omc derives at $c$ anything that it
% could possibly derive at any constant in any database. For a set
% $\Gamma \subseteq \mn{adom}(\Dmc)$, we say that \Dmc is
% \emph{derivation complete} at $\Gamma$ if it
% is derivation complete at every $c \in \Gamma$.
%

The databases $\Dmc_\omega$ and $\Dmc_{\mn{tree}}$ are both derivation
complete. Moreover, in all three reductions the construction of the
final databases \Dmc makes sure that for every
$c \in \mn{adom}(\Dmc_0)$ there is a simulation $S$ from
$\Dmc_{\mn{tree}}$ to \Dmc such that $(\varepsilon,c) \in S$. This
guarantees that $\Dmc$ is derivation complete at $c$.

To show all of this formally, let us define a weaker version of
derivation completeness that only pertains to concept names, c.f.\ the
construction of $\Dtree$.  A concept name~$A$ is \emph{non-empty} if
there is a $\Sigma$-database $\Dmc$ and a $c \in \mn{adom}(\Dmc)$ such
that $\Dmc,\Omc \models A(c)$.  Now, a $\Sigma$-database \Dmc is
\emph{derivation complete for concept names at}
$c \in \mn{adom}(\Dmc)$ if $\Dmc,\Omc \models A(c)$ for every
non-empty concept name $A$. % The following observation relies on the
% fact that \Omc is in normal form.

% \begin{lemma}
%   Let $\Dmc$ be a $\Sigma$-database and $c \in \mn{adom}(\Dmc)$.  If
%   $\Dmc$ is derivation complete for concept names at
%   $c$, then $\Dmc$ is derivation complete at $c$.
% \end{lemma}
% %
% \begin{proof}
%   Analyze definition of universal models; simple and short.
% \end{proof}

% We next provide a way to make a database derivation complete.
% Let \Rmc be the set of all roles $r,r^-$ with $r$ a role name from
% $\Sigma$.  The infinite tree-shaped $\Sigma$-database~$D_\omega$ has
% as its active domain $\mn{adom}(\Dmc_\omega)$ the set of all (finite)
% words over alphabet $\Rmc$ and contains the following facts:
% %
% \begin{itemize}

% \item $A(a)$ for all $w \in \mn{adom}(\Dmc_\omega)$ and concept
%   names $A \in \Sigma$;

% \item $r(w,w')$ for all $w,w' \in \mn{adom}(\Dmc_\omega)$ with $w'=wr$;

% \item $r(w',w)$ for all $w,w' \in \mn{adom}(\Dmc_\omega)$ with
%   $w'=wr^-$.

% \end{itemize}
%
\begin{lemma}
    \label{lem:derivation-complete-tree}
  $\Dmc_\omega$ is derivation complete for concept names
  at~$\varepsilon$.
\end{lemma}
\begin{proof}
  Let $A$ be a non-empty concept name.  Then there is a
  $\Sigma$-database \Dmc and a $c \in \dom(\Dmc)$ with
  $\Dmc,\Omc \models A(c)$.  Clearly, we have
  $(\Dmc,c) \preceq (\Dmc_\omega,\varepsilon)$.  Hence,
  Lemma~\ref{lem:sim} implies
  $\Dmc_\omega,\Omc\models A(\varepsilon)$, as required.
\end{proof}
It follows that, by construction, $\Dtree$ is also derivation complete
for concept names. By the following lemma, $\Dtree$ is derivation
complete in general and what is more, simulation from
$(\Dtree,\varepsilon)$ guarantees derivation completeness.
\begin{lemma}
  \label{lem:treesim}
  Let $\Dmc$ be a $\Sigma$-database that is satisfiable w.r.t.\ \Omc
  and $c \in \mn{adom}(\Dmc)$ such that
  $(\Dtree,\varepsilon) \preceq (\Dmc,c)$.
  Then $\Dmc$ is derivation complete at $c$.
  % Let $\Dmc$ be a $\Sigma$-database that is satisfiable w.r.t.\ \Omc
  % and let $\Gamma \subseteq \mn{adom}(\Dmc)$ such that
  % $(\Dtree,\varepsilon) \preceq (\Dmc,c)$ for all $c \in \Gamma$.
  % Then $\Dmc$ is derivation complete at $\Gamma$.
\end{lemma}
\begin{proof}
  Let $\Dmc$ and $c$ be as in the lemma. Further let $\Dmc'$ be any
  $\Sigma$-database and $c' \in \mn{adom}(\Dmc')$. We have to show
  that there is a homomorphism $h$ from
  $\Umc^{\downarrow c'}_{\Dmc',\Omc}$ to $\Umc_{\Dmc,\Omc}$ with
  \mbox{$h(c')=c$}. From the fact that $\Dmc_{\mn{tree}}$ is
  derivation complete for concept names at $\varepsilon$, it follows
  that $\Dmc',\Omc \models A(c')$ implies
  $\Dmc_{\mn{tree}},\Omc \models A(\varepsilon)$. From
  $(\Dtree,\varepsilon) \preceq (\Dmc,c)$ and Lemma~\ref{lem:sim}, we
  obtain $\Dmc,\Omc \models A(c)$. Lemma~\ref{lem:newintermediate}
  thus yields the desired homomorphism $h$.
\end{proof}

\newcommand{\mypath}{\pi}
\newcommand{\len}{\text{len}}
\subsection{Additional Preliminaries}
A \emph{path} $\mypath$ in a database $\Dmc$ is a sequence of database constants
$a_0, \dots, a_k$ such that for $0 \leq i < k$ there is an edge
$\{a_i,a_{i+1}\} $ in the Gaifman graph $G_{\Dmc}$.
%By $\len(\mypath)$ we denote the lenght of the path
%and we access the elements of the path by writing $\mypath[\cdot]$.
%Here, $\len(\mypath)=k+1$ and for $0 \leq i \leq k$
%$\mypath[i] = a_i$ and $\mypath[i]$ is undefined for all other values of $i$.
%there is a fact $R(a_i,a_{i+1})$\

%We say that a path $a_0, \dots, a_k$
%\emph{crosses} $a$ if $a = a_i$ for some $0 \leq i \leq k$,
%and that two paths \emph{cross} if there is $a$
%that both path cross.

We say that a path is \emph{simple} if no constant repeats, that it is
\emph{chordless} if there is no edge $\{a_i,a_j\}$ such that $|i-j|>1$
except for $\{i,j\}=\{1,k\}$, and that it is a \emph{cycle} if
$k \geq 1$ and $\{a_0,a_k\}$ is an edge or if $k=1$ and \Dmc contains
two distinct facts that both mention $a_0$ and $a_1$.  We say that
constants $c,c' \in \mn{adom}(\Dmc)$ are \emph{connected in \Dmc} if
\Dmc contains a path that starts at $c$ and ends at $c'$.

A \emph{path} in a CQ $q$ is a sequence $y_0, \dots, y_k$ of
variables from $q$ that forms a path in $\Dmc_q$. The path is
\emph{functional} if for $0\leq i <k$ there is an atom
$r(y_i,y_{i+1}) \in q $ such that $\func(r) \in \Omc$.  % It is a

\subsection{Cyclic Queries - The Case of Non-Chordality}

We finish the proof of Lemma~\ref{lemma:trianglelower} that we have
started in the main body of the paper. First, however, we give an
example of the construction of $\Dmc_0$ %, shown
in~Figure~\ref{fig:ex-cyc}. We use the query $q$ and graph $G$ shown in
the top part of the figure and assume that the role names $f_1$ and
$f_2$ are declared functional by \Omc.

\begin{figure*}[t]
    %\begin{wrapfigure}{r}{0.45\textwidth}
    \centering
    \begin{tikzpicture}[scale=.8]

    %delimiters
%%%%%%%%%%%%%%%grid
\draw[dashed,gray] (10,-1.5)--(10,2);
\draw[dashed,gray] (0,-1.5)--(20,-1.5);

%%%%%%%%%%%%query
    \node (q) at (5,2) {query $q$};

    \node (y0) at (5,-1) {$y_0$};
    \node (z1) at (8 ,-0) {$z_1$};
    \node (y1) at (5, 1) {$y_1$};
    \node (y2) at (2, 1) {$y_2$};
    \node (y3) at (2, -1) {$y_3$};

    \draw[-latex] (y0)--(z1) node[midway, below] {$f_0$};
    \draw[-latex] (z1)--(y1) node[midway, above] {$f_1$};
    \draw[-latex] (y1)--(y2) node[midway, above] {$R_{12}$};
    \draw[-latex] (y2)--(y3) node[midway, left] {$R_{23}$};
    \draw[-latex] (y3)--(y0) node[midway, above] {$R_{30}$};

%%%%%%%%%%%graph
    \node (g) at (15,2) {graph $G$};

    \node (a) at (13,-1) {$a$};
    \node (b) at (17,-1) {$b$};
    \node (c) at (17, 1) {$c$};
    \node (d) at (13, 1) {$d$};

    \draw (a)--(b)--(c)--(a)--(d);

%%%%%%%%%%%database

    \node (db) at (10,-2) {database $\Dmc_0$};

    \node (r12) at (2,-3) {$R_{12}$};
    \node (r23) at (6,-3) {$R_{23}$};
    \node (r30) at (10,-3) {$R_{30}$};
    \node (f0)  at (14,-3) {$f_0$};
    \node (f1)  at (18,-3) {$f_1$};

    \node (a1) at (0,-4) {$y_1{\mapsto} a$};
    \node (b1) at (0,-6) {$y_1{\mapsto} b$};
    \node (c1) at (0,-8) {$y_1{\mapsto} c$};
    \node (d1) at (0,-10) {$y_1{\mapsto} d$ };

    \node (a2) at (4,-4) {$y_2{\mapsto} a$};
    \node (b2) at (4,-6) {$y_2{\mapsto} b$};
    \node (c2) at (4,-8) {$y_2{\mapsto} c$};
    \node (d2) at (4,-10) {$y_2{\mapsto} d$};

    \node (a3) at (8,-4) {$y_3{\mapsto} a$};
    \node (b3) at (8,-6) {$y_3{\mapsto} b$};
    \node (c3) at (8,-8) {$y_3{\mapsto} c$};
    \node (d3) at (8,-10) {$y_3{\mapsto} d$};

    \node (az) at (16,-4) {$y_0y_1{\mapsto} ab$};
    \node (bz) at (16,-5) {$y_0y_1{\mapsto} ac$};
    \node (cz) at (16,-6) {$y_0y_1{\mapsto} ad$};
    \node (dz) at (16,-7) {$y_0y_1{\mapsto} ba$};
    \node (av) at (16,-8) {$y_0y_1{\mapsto} bc$};
    \node (bv) at (16,-9) {$y_0y_1{\mapsto} ca$};
    \node (cv) at (16,-10) {$y_0y_1{\mapsto} cb$};
    \node (dv) at (16,-11) {$y_0y_1{\mapsto} da$};

    \node (a0) at (12,-4) {$y_0{\mapsto} a$};
    \node (b0) at (12,-6) {$y_0{\mapsto} b$};
    \node (c0) at (12,-8) {$y_0{\mapsto} c$};
    \node (d0) at (12,-10) {$y_0{\mapsto} d$};

    \node (a11) at (20,-4) {$y_1{\mapsto} a$};
    \node (b11) at (20,-6) {$y_1{\mapsto} b$};
    \node (c11) at (20,-8) {$y_1{\mapsto} c$};
    \node (d11) at (20,-10) {$y_1{\mapsto} d$};

    \draw[black,-latex] (a1.east)--(a2.west);
    \draw[black,-latex] (a2.east)--(b3.west);
    \draw[black,-latex] (a3.east)--(b0.west);

    \draw[black,-latex] (b1.east)--(b2.west);
    \draw[black,-latex] (b2.east)--(a3.west);
    \draw[black,-latex] (b3.east)--(a0.west);

    \draw[black,-latex] (b2.east)--(c3.west);
    \draw[black,-latex] (b3.east)--(c0.west);

    \draw[black,-latex] (c1.east)--(c2.west);
    \draw[black,-latex] (c2.east)--(b3.west);
    \draw[black,-latex] (c3.east)--(b0.west);

%    (c1)--(a2)
    \draw[black,-latex] (c2.east)--(a3.west);
    \draw[black,-latex] (c3.east)--(a0.west);

%    (a1)--(c2)
    \draw[black,-latex] (a2.east)--(c3.west);
    \draw[black,-latex] (a3.east)--(c0.west);

%    (a1)--(d2)
    \draw[black,-latex] (d1.east)--(d2.west);
    \draw[black,-latex] (a2.east)--(d3.west);
    \draw[black,-latex] (d2.east)--(a3.west);
    \draw[black,-latex] (a3.east)--(d0.west);
    \draw[black,-latex] (d3.east)--(a0.west);

    \draw[black,-latex] (a0.east)--(az.west);
    \draw[black,-latex] (a0.east)--(bz.west);
    \draw[black,-latex] (a0.east)--(cz.west);

    \draw[black,-latex] (b0.east)--(dz.west);
    \draw[black,-latex] (b0.east)--(av.west);

    \draw[black,-latex] (c0.east)--(bv.west);
    \draw[black,-latex] (c0.east)--(cv.west);

    \draw[black,-latex] (d0.east)--(dv.west);

    \draw[black,-latex] (dz.east)--(a11.west);
    \draw[black,-latex] (bv.east)--(a11.west);
    \draw[black,-latex] (dv.east)--(a11.west);

    \draw[black,-latex] (az.east)--(b11.west);
    \draw[black,-latex] (cv.east)--(b11.west);

    \draw[black,-latex] (bz.east)--(c11.west);
    \draw[black,-latex] (av.east)--(c11.west);

    \draw[black,-latex] (cz.east)--(d11.west);

%    \node[quantified] (u2) at (4,0) {};

    \end{tikzpicture}
   \caption{Example for the construction of $\Dmc_0$. The role name
     $R_{01}$ is interpreted as the identity in $\Dmc_0$ (not shown in
     picture).
     % {\color{red} Note that elements of form $w {\mapsto} v$ are constants in the database and not facts.}
    }

    \label{fig:ex-cyc}
\end{figure*}
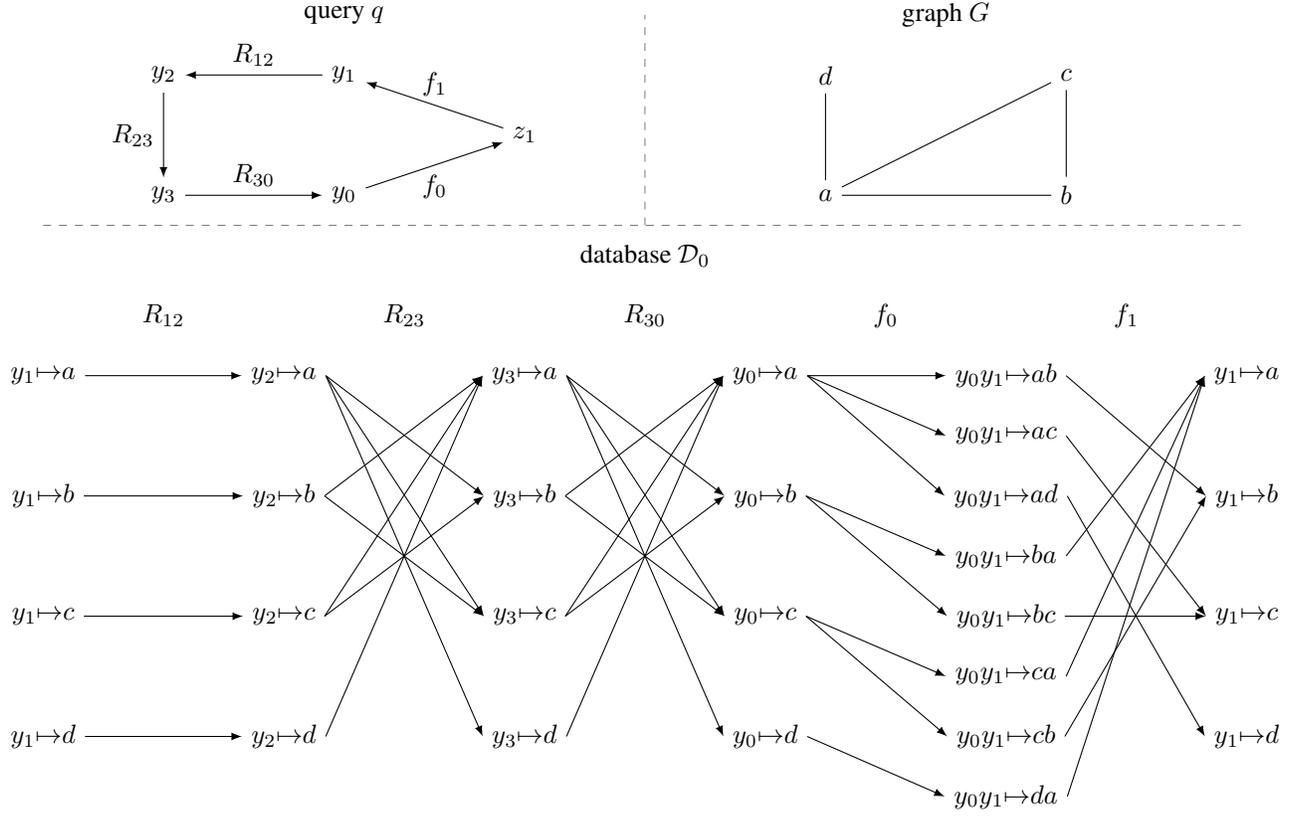

\medskip

Now back to the proof of Lemma~\ref{lemma:trianglelower}.
We first observe some basic properties of
$\Dmc_0$. % There, we started from an OMQ
% $Q(\bar x)=(\Omc,\Sigma,q) \in (\ELI,\text{CQ})$ as in
% Lemma~\ref{lemma:trianglelower}, a chordless cycle $y_0,\dots,y_k$,
% and an undirected graph $G=(V,E)$, and we have constructed a database
% $\Dmc_0$. We first provide some intuitive explanation of the
% construction of $\Dmc_0$ and then extend $\Dmc_0$ to obtain the
% database \Dmc used in the reduction.
%
% The database $\Dmc_0$ may be seen as a variation of the direct product
% of the query $q$ and the graph $G$. To make this more precise, we
% observe the following.
%
\begin{claim}
        \label{claim:cycle-small-y}
        % For every atom $R(z,z')$ in $q$
        ~\\[-4mm]
        \begin{enumerate}
        \item for every variable $x$ in
          $q$, %for all $0 \leq i,j \leq k$ we have that
%        if $y_i,y_j \in V_x$ then $|i-j|\leq 1$ or $\{i,j\} =
%        \{1,k\}$.
        we have $Y_x \subseteq \{y_i,y_{i+1}\}$ for some $0\leq i \leq k$.
%        we have $Y_z \cup Y_{z'} \subseteq \{y_i,y_{i+1}\}$ for some
%        $0\leq i \leq k$.% or $Y_x \subseteq \{y_k,y_{0}\}$

        \item $Y_{y_i}= \{ y_i\}$ for $0 \leq i \leq k$.

        \end{enumerate}
    \end{claim}
%    \cMP{do we state this for every variable?}
    %
%    Proof of this claim follows the same logic as the proof of Claim\ref{claim:MM-small-Y}.
    \begin{proof}
      Both points follow from the chordlessness of the cycle
      $y_0,\dots,y_k$.

      For Point~1, let $x$ be a variable in $q$ and assume that
      $Y_x \supseteq \{y_j,y_i\}$ and
      $y_j \notin \{ y_{i-1},y_{i+1}\}$ with $y_{-1}=y_k$. By
      definition of $Y_x$, there is then a functional path in $q$ from
      $x$ to $y_j$ and from $x$ to $y_i$. Thus, by construction of
      $q^+$ the Gaifman graph of $q^+$ contains the edge
      $\{y_j,y_i\}$, in contradiction to $y_0,\dots,y_k$ being
      chordless.

      For Point~2, first note that $y_i \in Y_{y_i}$ by definition of
      $Y_{y_i}$. By Point~1, this leaves as the only possible options
      $Y_{y_i}= \{ y_i\}$, $Y_{y_i}=\{y_i,y_{i-1}\}$ (with
      $y_{-1}=y_k$), and $Y_{y_i}=\{ y_i,y_{i+1}\}$. The latter two
      options, however, are not possible. We show this exemplarily for
      the last option.  By construction of $q^+$, the fact that
      $\{ y_i,y_{i-1}\}$ is an edge in the Gaifman graph of $q^+$
      implies that $q$ contains an atom $\alpha$ that contains
      variables $u,v$ such that $y_{i-1}$ is reachable on a functional
      path from $u$ and $y_{i}$ is reachable on a functional path from
      $y_i$. Since $y_{i+1} \in Y_{y_i}$, there is also a functional
      path from $y_i$ to $y_{i+1}$ and thus we find a functional path
      from $v$ to $y_{i+1}$. By construction of $q^+$, the Gaifman
      graph of $q^+$ contains the edge $\{ y_{i-1},y_{i+1}\}$, in
      contradiction to the chordfreeness of $y_0,\dots,y_k$.
\end{proof}

    \begin{lemma}
        \label{claim:cycle-database-is-ok}
        The database $\Dmc$
        \begin{enumerate}
            \item can be computed in time $O(||E||)$.
            \item satisfies all functionality
              assertions in~$\Omc$.
            \item is derivation complete at $\mn{adom}(\Dmc_0)$.
        \end{enumerate}
    \end{lemma}
    \begin{proof}
      Point~1 is clear from the construction. % , the second follows the
      % same proof structure as the second bullet in
      % Claim~\ref{claim:MM-database-is-ok}.

      For Point~2, take a role $R$ such that $\func(R) \in \Omc$ and
      facts $R(c,d),R(c,d') \in \Dmc$. We show that $d=d'$.
      If one of $c$, $d$ or $d'$ belongs to a copy of $\Dtree$ then
      clearly $d =
      d'$. % Now, if $d$ (or $d'$) belongs to a copy of a tree, again,
           % clearly $d = d'$.
      Thus, let us assume that $c,d,d' \in \mn{adom}(\Dmc_0)$.  Then
      the facts $R(c,d)$ and $R(c,d')$ were added to $\Dmc_0$ because
      of an atom $R(x,y)$ in $q$ and, due to the self-join freeness of
      $q$ we have $c = \langle x, f_c \rangle $,
      $d = \langle y, f_d \rangle$, and
      $d' = \langle y, f_{d'} \rangle$ for functions $f_c,f_d,f_{d'}$.
      From $\func(R) \in \Omc$ we obtain $Y_y \subseteq Y_x$.  No
      matter whether fact $R(c,d)$ was added in Step~1, 2, or~3 of the
      construction of $\Dmc_0$, there is a word $w$ such that
      $f_c=f^w_x$ and $f_d=f^w_y$. It thus follows from
      $Y_y \subseteq Y_x$ that for all $y_i \in Y_y$, we have
      $f_d(y_i) = f_c(y_i)$.  The same is true for the fact $R(c,d')$,
      and thus we obtain $f_{d} = f_{d'}$, which in turn implies $d=d'$.

      By construction of \Dmc, it is easy to see that
      $(\Dtree,\varepsilon) \preceq (\Dmc,c)$ for all
      $c \in \mn{adom}(\Dmc_0)$.  By Lemma~\ref{lem:treesim}, it
      follows that $\Dmc$ is derivation complete at $\mn{adom}(\Dmc_0)$.
  % Let $\Dmc$ be a $\Sigma$-database that is satisfiable w.r.t.\ \Omc
  % and let $\Gamma \subseteq \mn{adom}(\Dmc)$ such that
  % $(\Dtree,\varepsilon) \preceq (\Dmc,c)$ for all $c \in \Gamma$.
  % Then $\Dmc$ is derivation complete at $\Gamma$.
    \end{proof}
    %
    % Since $\Dmc$ satisfies functional assertions in $\Omc$ and can be
    % constructed in time linear in the size of $G$, we can use it as
    % input for the enumeration algorithm promised in the lemma.  We now
    % describe the relation between answers to $Q$ on $\Dmc$ and
    % triangles in $G$.

    \begin{claim}
        \label{claim:paths-in-core-triangle}
        Let $h$ be a homomorphism from $q$ to $\Umc_{\Dmc,
          \Omc}$. Then
        \begin{enumerate}

        \item for every simple path $u_0, \dots, u_m$ in $q$, $m>0$,
        %        {\color{violet} If $h(u_0) \in \dom(\Dmc)$ and $h(u_m) \in \dom(\Dmc)$, then $h(u_i) \in \dom(\Dmc)$ for
            %            $0<i\leq m$. Moreover, if}
          $h(u_0) \in \DGamma$ and $h(u_m) \in \DGamma$ implies
          $h(u_i) \in \DGamma$ for $0\leq i\leq m$;

        \item for every cycle $u_0,\dots,u_m$ in $q$, we have
          $h(u_i) \in \DGamma$ for $0 \leq i \leq m$.

        \end{enumerate}
    \end{claim}

    \begin{proof}
      We only prove Point~1, the proof of Point~2 is similar.
        Assume towards a contradiction that $h(u_j) \notin \DGamma$ for
        some $j$ with $0 \leq j \leq m$ and additionally let $j$ be
        smallest with this property. Moreover, let $\ell>j$ be smallest
        such that $h(u_\ell)\in \DGamma$. Note that $0 < j < \ell < m$.
        By definition of a path, $q$ must contain atoms
        $R_j(u_{j-1},u_j),\dots,R_\ell(u_{\ell-1},u_\ell)$. Since $q$ is
        self-join free, the relation symbols  $R_j,\dots,R_\ell$ are
        all distinct.

        Next observe that by construction of \Dmc and by definition of
        universal models, the interpretation \Imc obtained from
        $\Umc_{\Dmc, \Omc}$ by removing all edges $r(a,b)$ with
        $a,b \in \DGamma$ is a collection of trees without multi-edges.
        More precisely, these are the trees added in the second step in
        the construction of \Dmc and the trees $\Umc_{\Dmc,
            \Omc}^{\downarrow c}$
        of $\Umc_{\Dmc, \Omc}$.
        We view as the roots of these trees the elements of $\DGamma$.

        Then $h$ maps $R_j(u_{j-1},u_j),\dots,R_\ell(u_{\ell-1},u_\ell)$
        to some tree of this collection, and both $u_{j-1}$ and $u_\ell$
        are mapped to the root of the tree. This, however, clearly
        contradicts the fact that the relation symbols
        $R_j,\dots,R_\ell$ are all distinct and the tree has no multi-edges.
    \end{proof}
    The next lemma states soundness and completeness of the reduction.
    \lemcyclecorr*

%     \paragraph{The algorithm -- triangle detection}
%     Before we prove the claim, let us show how we use it to prove the lemma.
% %
% %
% %    We now describe the algorithm that detects a triangle in a graph and uses the assumed enumeration algorithm.
% %

%     Given graph $G$, we construct the database $\Dmc$ as above. Then we run the enumeration algorithm.
%     If it returns an answer, we immediately stop the algorithm and return \emph{true}.
%     If it returns no answers, we return \emph{false}.

%     Since $\Dmc$ is constructed in time $O(||G||)$, we have that $||\Dmc||$ is $O(||G||)$.
%     Hence, the preprocessing phase of the enumeration algorithm is done in time $O(||G||)$
%     and, thus, the first answer, if exists, is produced in time $O(||G||)$.

%     The correctness of the algorithm follows immediately from Claim~\ref{claim:lower-bound-partial-answers-and-cycle}.
%     This ends the proof of Lemma~\ref{lemma:lower-bound-enum-cycle}

% \medskip

    % The following observation will be important for the proof of
    % Claim~\ref{claim:lower-bound-partial-answers-and-cycle}.
    %
%
We now prove
Claim~\ref{claim:lower-bound-partial-answers-and-cycle}, which
finishes the lower bound proof.  To show
\textbf{TD2}, let $a,b,c$ be a triangle in $G$. To show that there is
a complete answer to $Q$ on \Dmc, it suffices to exhibit a
homomorphism $h$ from $q$ to $\Umc_{\Dmc, \Omc}$ such that
$h(x) \in \mn{adom}(\Dmc_0)$ for every answer variable~$x$.

% We first fix a partial mapping $h$ from $q$ to $\Umc_{\Dmc, \Omc}$
% such that $h(y_0) = \{(y_0, a)\}$, $h(y_i) = \{(y_i, b)\}$ for
% $0 < i < k$, and $h(y_k) = \{(y_k, c)\}$.  Following the same steps as
% in proof of \textbf{MM2}, we can show that this mapping can be
% extended to a homomorphism from $q$ to $\Umc_{\Dmc, \Omc}$. This
% proves \textbf{CD2}.

    Since $Q$ is non\=/empty, there is a database $\Dmc'$ that is
    satisfiable w.r.t.\ \Omc and a homomorphism $h'$ from $q$ to
    $\Umc_{\Dmc', \Omc}$ such that $h'(\bar{x})$ contains only
    elements of $\mn{adom}(\Dmc')$.  We use $h'$ to guide the
    construction of $h$.

    Let $W_0$ be the set of variables $x \in \mn{var}(q)$ with
    $h'(x) \in \mn{adom}(\Dmc')$ and for each
    $c \in \mn{adom}(\Dmc')$, let $W_c$ be the set of variables
    $x \in \mn{var}(q)$ such that $h'(x)$ is a trace that starts
    with~$c$, that is, $h'(x)$ is $c$ or located in the tree
    $\Umc_{\Dmc', \Omc}^{\downarrow c}$ in $\Umc_{\Dmc', \Omc}$ that is
    rooted in $c$.
    % More precisely, this
    % tree is the subinterpretation of $\Umc_{\Dmc', \Omc}$ that
    % consists of all traces starting with $c$ and we denote it with
    % $\Umc_c$.

    We first note that for all $c \in \mn{adom}(\Dmc')$, the
    restriction $q_c$ of $q$ to the variables in $W_c$ is a disjoint
    union of trees (without multi-edges and reflexive loops) in the
    sense that the Gaifman graph of $q_c$ is a tree. Let $p$ be a
    connected component of $q_c$. Clearly, $p$ has no multi-edges and
    reflexive loops since the same is true for
    $\Umc_{\Dmc', \Omc}^{\downarrow c}$. It remains to show that $p$
    is a tree. Assume to the contrary that $p$ contains atoms
    $R_1(x_0,x_1),\dots,R_m(x_{m-1},x_m)$ with $x_0=x_m$.  Since
    $\Umc_{\Dmc', \Omc}^{\downarrow c}$ is a tree, $h'$ must map two
    distinct atoms from this list to the same fact in
    $\Umc_{\Dmc', \Omc}^{\downarrow c}$. This contradicts the
    self-join freeness of $q$, which implies that $R_1,\dots,R_m$ are
    pairwise distinct.

    We start the construction of $h$ by setting
    $h(x)=\langle x,f_{x}^{ab^{k{-}1}c}\rangle$ for all $x \in
    W_0$. Let $q_0$ be the restriction of $q$ to the variables in
    $W_0$. It can be verified that $h$ is indeed a homomorphism from
    $q_0$ to $\Umc_{\Dmc, \Omc}$. In fact, binary atoms $R(x,x')$ in
    $q_0$ are preserved: since $h'(x),h'(x') \in \mn{adom}(\Dmc')$ and
    the construction of $\Umc_{\Dmc',\Omc}$ from $\Dmc'$ does not add
    any binary facts $R(a,a')$ with $a,a' \in \mn{adom}(\Dmc')$, we
    must have $R \in \Sigma$; we thus obtain
    $R(h(x),h(x')) \in \Dmc_0$ from the construction of $\Dmc_0$ and
    the fact that $a,b,c$ is a triangle in $G$. Unary atoms $A(x)$
    with $A \in \Sigma$ are trivially preserved by construction of
    $\Dmc$. And unary atoms $A(x)$ with $A \notin \Sigma$ are
    preserved since $\Dmc$ is derivation complete at
    $\mn{adom}(\Dmc_0)$. More precisely, if $A(x)$ is an atom in
    $q_0$, then non-emptiness of $q$ implies that $A$ is non-empty,
    and since derivation completeness (trivially) implies derivation
    completeness for concept names it follows that
    $\Omc,\Dmc \models A(c)$ for all $c \in \mn{adom}(\Dmc_0)$.

    To extend $h$ to the remaining variables in~$q$, consider each set
    $W_c$ separately (clearly, those sets are disjoint and contain all
    remaining variables).  As established above, the restriction $q_c$
    of $q$ to the variables in $W_c$ is a disjoint union of trees. Let
    $p$ be any such tree. Since $q$ is connected, $p$ contains a
    (unique) variable $x$ with $h'(x)=c$. Since \Dmc is derivation
    complete at $c$, there is a homomorphism $h_c$ from
    $\Umc_{\Dmc', \Omc}^{\downarrow c}$ to $\Umc_{\Dmc,\Omc}$ with
    $h_c(c)=h(x)$. Then set $h(y)=h_c(h'(y))$ for every
    $y \in \mn{var}(p)$. This finishes the construction of $h$.

    %
    % Using the construction of \Dmc,  it is now straightforward to
    % verify that $h$ is indeed a homomorphism from $q$ to
    % $\Umc_{\Dmc, \Omc}$ such that
    % $h(x_1)= \langle x_1, f_{x_1}\rangle$ where $f_{x_1}$ is a
    % function that satisfies $f_{x_1}(y_0) = a$ and
    % $h(x_2)= \langle x_2, f_{x_2}\rangle$ where $f_{x_2}$ is a
    % function that satisfies $f_{x_2}(y_k) = b$.

    \medskip

    To show \textbf{TD1}, assume that there is a minimap partial
    answer to $Q$ on \Dmc. Then there is a homomorphism $h$ from $q$ to
    $\Umc_{\Dmc, \Omc}$. Ideally, we would like to show that for
    $i \in \{0,1,k\}$, $h(y_i)$ is of the form
    $\langle y_i, f_i \rangle$ and $f_0(y_0), f_1(y_1), f_k(y_k)$
    forms a triangle in~$G$.  Unfortunately, it may be the case that
    $h(y_i) \notin \mn{adom}(\Dmc_0)$ and thus we need to argue in a slightly
    more complex way.
    We first show the following.
    \begin{claim}
      \label{claim:triangleui}
        There is a variable $u_i$
        such that $h(u_i) \in \mn{adom}(\Dmc_0)$ and $y_i \in Y_{u_i}$,
        for $0\leq i \leq k$.
    \end{claim}
    \begin{proof}
    % To prove the claim observe that for every $0 \leq i \leq k$,
    % there is a variable $u_i$ such that $y_i \in Y_{u_i}$
    % and $u_i$ lies on some simple cycle.
%
    Let $0\leq i \leq k$. %  Since $y_0, \dots, y_k$ is a cycle in
    % $q^+$, %by Lemma~\ref{lemma:paths-from-q-plus-to-q},
    Recall that $y_0,\dots,y_k$ is a cycle in $q^+$ and let
    $y_{k+1}=y_0$.  By construction of $q^+$, this means that for
    $0 \leq j \leq k$, $q$ contains an atom $\alpha_j$ that contains
    (possibly identical) variables $u_j,v_j$ such that $q$ contains a
    functional path $\pi_j$ from $u_j$ to $y_j$ and a functional path
    $\pi'_j$ from $v_j$ to $y_{j+1}$.\footnote{One possible and
      particularly simple case is that $\alpha_j$ takes the form
      $R(y_j,y_{j+1})$, $\pi_j=y_j$, and $\pi'_j=y_{j+1}$.}  The
   concatenation $\pi^-_0\pi_0 \cdots \pi^-_k\pi_k$ is a cycle
    $C=z_0, \dots, z_\ell$ in $q$ where $\cdot^-$ denotes path
    reversal.  It may be that $C$ is not a simple cycle, but we can
    use it to find a simple cycle in $q$ that contains a variable
    $u_i$ with $y_i \in Y_{u_i}$.

    Let $u_i$ be the last variable on the longest common prefix of the
    paths ${\pi'}_{i-1}^-\pi_{i-1}$ and $\pi_i^-\pi'_i$. This $u_i$
    must be in the $({\pi'}_{i-1})^-$ prefix of the first path and in
    the $\pi_i^-$ prefix of the second path because all variables $z$
    in the $\pi_{i-1}$ suffix of the first path satisfy
    $y_{i-1} \in Y_z$ and all variables $z$ in the $\pi'_{i}$ suffix
    of the second path satisfy $y_{i+1} \in Y_z$, c.f.\
    Claim~\ref{claim:cycle-small-y}.  This clearly implies that
    $y_i \in Y_{u_i}$.  We can now obtain a cycle $C'$ by starting
    with $C$ and
    \begin{enumerate}

    \item cutting out the prefixes of ${\pi'}^-_{i-1}$ and $\pi_i^-$
      up to (excluding) $u_i$ and

    \item removing subpaths that end in the same variable (and do not
      contain $u_i$) from the remaining
      cycle to make it simple.

    \end{enumerate}
    By Point~2 of Claim~\ref{claim:paths-in-core-triangle},
    $h(u_i) \in \mn{adom}(\Dmc_0)$.
    \end{proof}

    % \medskip

    % Now, since $|Y| \geq 3$ and for every variable $z$ in $q$ we have that
    % $Y_z \subseteq \{y_i, y_{i+1}\}$ for some $0 \leq i \leq k$,
    % there is an atom $R(x,x')$ in $q$ such that $x \neq x'$ and $h(x), h(x') \in \DGamma$.
    % This atom allows us to reason about the values of $h(x)(y_i)$, if they are defined.

    \begin{claim}
        \label{claim:cycle-well-guided-hom}
        Let $y,y'$ be variables in $q$ such that %$y \neq y'$ and
        $h(y),h(y') \in \mn{adom}(\Dmc_0)$.  Then,
        $h(y) = \langle y, f_y\rangle$ and
        $h(y) = \langle y', f_{y'}\rangle$ for some functions $f_y$
        and $f_{y'}$.  Moreover, if $y_i \in Y_{y} \cap Y_{y'}$ then
        $f_y(y_i) = f_{y'}(y_i)$.
    \end{claim}
    \begin{proof}
      Since $q$ is connected, there is a simple path
      $z_0, \dots, z_\ell$ in $q$ from $z_0=y$ to $z_\ell=y'$.  Since
      $h(y), h(y') \in \mn{adom}(\Dmc_0)$ and by Point~1 of
      Claim~\ref{claim:paths-in-core-triangle},
      all variables on this path are mapped
      to $\mn{adom}(\Dmc_0)$. Thus, there are atoms $R(y, z_1)$ and
      $S(z_{\ell-1}, y')$ such that $R(h(y), h(z_1)) \in \Dmc_0$
      and $S(h(z_{\ell-1}), h(y'))\in \Dmc_0$. It now follows from
      the construction of $\Dmc_0$ and the self-join freeness
      of $q$ that
      $h(y) = \langle y, f_y\rangle$ and
      $h(y) = \langle y', f_{y'}\rangle$.

        For the second part of the claim observe that if $y_i \in Y_{y} \cap Y_{y'}$
        then there are a simple functional path from $y$ to $y_i$
        and a simple functional path from $y'$ to $y_i$.
        Hence, there is a simple path from $y$ to $y'$
        such that for every variable $u$ on this path we have $y_i \in Y_u$.
        By Claim~\ref{claim:MM-functional-database}, this implies $f_y(y_i) = f_{y'}(y_i)$.
        \end{proof}
    Let $f = \bigcup_{0 \leq i \leq k} h(u_i)$. By
    Claim~\ref{claim:cycle-well-guided-hom}, $f$ is a function. We
    argue that $a=f(y_0)$, $b=f(y_1)$, and $c=f(y_k)$ form a triangle
    in $G$.

    First for the edge $\{a,b\}$. We start with observing that $q$
    contains a path $\pi$ from $u_0$ to $u_1$ such that
    $Y_z \cap \{ y_0,y_1\} \neq \emptyset$ for all variables $z$ on
    $\pi$. We can find $\pi$ by concatenating the functional path from
    $u_0$ to $y_0$ with a path from $y_0$ to $y_1$ that contains only
    variables $t$ with $Y_z \cap \{ y_0,y_1\} \neq \emptyset$ (which
    exists since $\{y_0,y_1\}$ is an edge in $G_{q^+}$ and then again
    with the reverse of the functional path from $u_1$ to $y_1$. We
    can convert $\pi$ into a simple path by cutting out subpaths.
    Since $h(u_0), h(u_1) \in \mn{adom}(\Dmc_0)$ and by Point~1
    of Claim~\ref{claim:paths-in-core-triangle}, all variables
    on this path are mapped to $\mn{adom}(\Dmc_0)$. It now follows from
    $y_0 \in Y_{u_0}$, $y_1 \in Y_{u_1}$, and Claim~\ref{claim:cycle-small-y}
    that there are two consecutive variables $z_0,z_1$ on $\pi$ such
    that $y_0 \in Y_{z_0}$ and $y_1 \in Y_{z_1}$. In summary, $q$
    contains an atom $R(z_0,z_1)$ with $y_0 \in Y_{z_0}$,
    $y_1 \in Y_{z_1}$, and $h(z_0),h(z_1) \in \mn{adom}(\Dmc_0)$.
    Then $R(h(z_0),h(z_1)) \in \Dmc_0$ and
    Claim~\ref{claim:cycle-well-guided-hom} yields
    $h(z_0)= \langle z_0,f_{z_0} \rangle$ with $f_{z_0}(y_0)=f(y_0)$
    and $h(z_1)=\langle z_1,f_{z_1} \rangle$ with
    $f_{z_1}(y_1)=f(y_1)$.  Since $q$ is self-join free, the fact
    $R(h(z_0),h(z_1))$ must have been added to $\Dmc_0$ due to the
    atom $R(z_0,z_1) \in q$.  Since
    $Y_{z_1} \cup Y_{z_2} = \{y_0,y_1\}$, the fact must have been
    added in Case~1. Thus $\{a,b\} \in E$, as required.

    The case of the edge $\{ b,c \}$ is analogous, with $u_0,y_0$
    replaced by $u_{k-1}$,$y_{k-1}$ and $u_1,y_1$ by
    $u_k,y_{k}$. Moreover, it uses Case~2 of the construction of
    $\Dmc_0$ in place of Case~1.

    The case of the edge $\{c,b\}$ is also analogous, with $u_0,y_0$
    replaced by $u_k,y_k$ and $u_1,y_1$ by $u_0,y_0$. The fact
    $R(h(z_0),h(z_1))$ was then added by Case~1 and Case~2
    simultaneously (they simply produce the same fact), which
    yields $\{c,b\} \in E$.

    \bigskip

\subsection{Cyclic Queries - The Case of Non-Conformity}

%%%%%%%%%%%%%%%%%%%%%%%%%%%%%%%%%%%%%%%%%%%%%%%%%%%%%%%%%%%%%%%%%%%%%%%%%%
We now consider the second case in the proof of Point~1 of
Theorem~\ref{lemma:enumeration-lower-bound}. The basic
idea of our reduction goes back to \cite{BraultBaron}, but
the details differ. Recall that a CQ~$q$ is
\emph{conformal} if for every clique $C$ in the Gaifman graph of~$q$,
there is an atom in $q$ that contains all variables in $C$.
\begin{lemma}
    \label{lemma:lower-bound-enum-hyper}
Let $Q(\bar x)=(\Omc,\Sigma,q) \in (\ELI,\text{CQ})$ be non-empty such
that $q$ is self\=/join free and connected and the hypergraph of $q^+$
is non\=/conformal.
Then enumerating complete answers to $Q$ is not in \dlc unless the
hyperclique conjecture fails.  The same is true for least partial
answers, both with a single wildcard and with multi-wildcards.
    % Let $Q(\bar{x}) = (\Omc, \Sigma, q)$ be a non\=/empty connected
    % OMQ such that the Gaifman graph of $q^{+}$ has a clique that is
    % not contained in an edge of $q^{+}$.  If there is an algorithm
    % that given an $\Sigma$\=/database $\Dmc$ after linear
    % preprocessing enumerates the set $Q(\Dmc)$ with constant delay,
    % then there is $k>2$ and an algorithm that given a $k$\=/regular
    % hypergraph $G$ detects a $(k+1)$\=/hyperclique in $G$, i.e. a set
    % of $k{+}1$ vertices where every subset of size $k$ forms a
    % hyperedge, in time $O(n^k)$, where $n$ is the number of vertices
    % in $G$.
\end{lemma}
Let $Q(\bar x)=(\Omc,\Sigma,q) \in (\ELI,\text{CQ})$ be as in
Lemma~\ref{lemma:lower-bound-enum-hyper}, and let
$Y = \{y_0, \dots, y_k\}$, $k \geq 2$, be a clique in the Gaifman
graph of $q^+$. Further assume that every proper subset of $Y$ is
contained in an edge in~$q^{+}$.  \cMP{what if k=2?}
% As before, for
% every variable $x$ in $q$ we use $Y_x$ to denote the set of variables
% $y \in Y$ such that $q$ contains a functional (possibly empty) path
% from $x$ to~$y$.

Let $G = (V, E)$ be a $k$-regular hypergraph, that is,
$V= \{v_1, \dots, v_n\}$ is a set of vertices and $E$ is a set of
subsets of $V$ of size $k$ called the hyperedges. Our aim is to
construct a database \Dmc, in time $n^k$, such that $G$ contains a
hyperclique of size $k+1$ if and only if $Q(\Dmc) \neq
\emptyset$. Clearly, a \dlc enumeration algorithm for $Q$ lets us
decide the latter in linear time and thus we have found an algorithm
for the $(k+1,k)$-hyperclique problem that runs in time $O(n^k)$,
refuting the hyperclique conjecture.

The reduction shares a lot of ideas, intuition, and notation with the
previous one. In particular, for every variable $x$ in $q$ the set
$Y_x$ is defined as in the previous reduction and we again define the
desired database \Dmc by starting with a `core part' $\Dmc_0$ that
encodes the hyperclique problem and then adding trees in a second
step.

We start with the construction of $\Dmc_0$, where we use the same
functions $f^w_x$ as in the previous reduction. For every atom
$r(x,y)$ in $q$ with $r \in \Sigma$, add the following facts to
$\Dmc_0$:
\begin{itemize}
\item $r(\langle x, f^{w}_x \rangle , \langle y, f^w_y \rangle)$ for
  every hyperedge $e = \{a_1, \dots, a_k\} \in E$ and every
  serialization $w = a_1 \cdots a_k$ of $e$.
\end{itemize}
% where, as before, $f^w_z$ is a function from the set $Y_z$ to the set $V$ such that for $f^w_z({y_i})=(w_i)$ for $y_i \in Y_z$.
%
%
% \paragraph{Adding trees.}

The construction of \Dmc from $\Dmc_0$ is exactly as in the previous
reduction, based on the database $\Dmc_{\mn{tree}}$ and its fragments
$\Dmc_R$. More precisely, we extend the database $\Dmc_0$ to the
desired database \Dmc as follows: for every $c \in \mn{adom}(\Dmc_0)$
and every role $R \in \{ r, r^-\}$ with $r \in \Sigma$ such that there
is no fact $R(c, c') \in \Dmc_0$, add a disjoint copy of $\Dmc_{R}$,
glueing the copy of $\varepsilon$ to $c$.

% Moreover, for every relation $r \in \Sbf$ and every constant $c$ used in one of the facts added above, we add a fresh copy of $t_K$
%anchored to $c$, i.e. \begin{itemize}
%    \item if there is no fact $r(c, c')$ in $\Dmc$ then we add to $\Dmc$  a copy of $t_K$ and fact $r(c,c')$ where $c'$
%    is the constant $\varepsilon$ in the fresh copy of $t_K$;
%    \item if there is no fact $r(c', c)$ in $\Dmc$ then we add to $\Dmc$ a copy of $t_K^{-}$ and fact $r(c',c)$ where $c'$
%    is the constant $\varepsilon$ in the fresh copy of $t_K^{-}$.
%\end{itemize}

Finally, for every unary relation $A$ symbol in $\Sigma$ and every constant $c$ in the so\=/far constructed database
we add fact $A(c)$.
\begin{lemma}
    \label{claim:hyper-database-is-ok} The database $\Dmc$
    \begin{enumerate}
        \item can be computed in time $O(n^k)$,
        \item satisfies all functional assertions in~$\Omc$,
        \item is derivation complete at $\mn{adom}(\Dmc_0)$.
    \end{enumerate}
\end{lemma}

Point~1 is clear from the construction and Points~2 and~3 can be
proved in the same way as the corresponding points in
Claim~\ref{claim:MM-database-is-ok}.

We observe the following counterpart of
Claim~\ref{claim:paths-in-core-triangle}.
The proof is omitted.
    \begin{claim}
        \label{claim:paths-in-core-hyperclique}
        Let $h$ be a homomorphism from $q$ to $\Umc_{\Dmc,
          \Omc}$. Then
        \begin{enumerate}

        \item for every simple path $u_0, \dots, u_m$ in $q$, $m>0$,
        %        {\color{violet} If $h(u_0) \in \dom(\Dmc)$ and $h(u_m) \in \dom(\Dmc)$, then $h(u_i) \in \dom(\Dmc)$ for
            %            $0<i\leq m$. Moreover, if}
          $h(u_0) \in \DGamma$ and $h(u_m) \in \DGamma$ implies
          $h(u_i) \in \DGamma$ for $0\leq i\leq m$;

        \item for every cycle $u_0,\dots,u_m$ in $q$, we have
          $h(u_i) \in \DGamma$ for $0 \leq i \leq m$.

        \end{enumerate}
    \end{claim}
    To finish the proof of Lemma~\ref{lemma:lower-bound-enum-hyper},
it remains to show the following.
\begin{lemma}

    \label{claim:lower-bound-partial-answers-and-hyper}
    \
    \begin{description}
        \item[HD1] if there is a minimal partial answer to $Q$ on $\Dmc$ then there is a $(k+1)$\=/hyperclique in the Gaifman graph of~$G$,
        \item[HD2] if there is a $(k+1)$\=/hyperclique in the Gaifman graph of $G$ then there is a complete answer to $Q$ on $\Dmc$.
    \end{description}
\end{lemma}
To show \textbf{HD2}, let $v_0, \dots, v_k$ be a hyperclique in
$G$. To show that there is a complete answer to $Q$ on \Dmc, it
suffices to exhibit a homomorphism $h$ from $q$ to $\Umc_{\Dmc, \Omc}$
such that $h(x) \in \mn{adom}(\Dmc_0)$ for every answer variable~$x$.
The construction of $h$ is very similar to the proof of the \textbf{TD2} part
of Claim~\ref{claim:lower-bound-partial-answers-and-cycle} and we
provide only an outline.

% We first fix a partial mapping $h$ from $q$ to $\Umc_{\Dmc, \Omc}$
% such that $h(y_0) = \{(y_0, a)\}$, $h(y_i) = \{(y_i, b)\}$ for
% $0 < i < k$, and $h(y_k) = \{(y_k, c)\}$.  Following the same steps as
% in proof of \textbf{MM2}, we can show that this mapping can be
% extended to a homomorphism from $q$ to $\Umc_{\Dmc, \Omc}$. This
% proves \textbf{CD2}.

Since $Q$ is non\=/empty, there is a database $\Dmc'$ that is
satisfiable w.r.t.\ \Omc and a homomorphism $h'$ from $q$ to
$\Umc_{\Dmc', \Omc}$ such that $h'(\bar{x})$ contains only elements of
$\mn{adom}(\Dmc')$.  Let $W_0$ be the set of variables
$x \in \mn{var}(q)$ with $h'(x) \in \mn{adom}(\Dmc')$.

We may construct the desired homomorphism $h$ by setting
$h(x) = \langle x, f^{v_0\cdots v_k}_x\rangle$ for all $x \in W_0$, and then extending
to the remaining variables in $q$ (which are part of tree-shaped
subqueries of $q$) exactly as in the proof of the \textbf{TD2} part
of Claim~\ref{claim:lower-bound-partial-answers-and-cycle}.

\medskip

% \cMP{do we want something like $Y_{R}$ for $\bigcup_{x \in R(\bar{x})} Y_x$ where $R(\bar{x}) \in q$?}

% Before we prove the claim, observe that it is possible that
% no homomorphism matches the variables from $Y$ to $\dom(\Dmc_0)$.
% Moreover, it is possible that for no relational symbol $R \in \Sigma$
% there is a variable $z$ such that $R(y_1,z)$ is an atom in $q$.
% Thus, we cannot ``set'' nor ``access'' the values that any homomorphisms maps $y_1$
% to by simply using atoms containing $y_1$.
% The following observation will allow us to do it implicitly.

To prove \textbf{HD1}, we start with observing an analogue of
Claim~\ref{claim:triangleui}.
\begin{claim}
  \label{claim:hypercliqueui}
        There is a variable $u_i$
        such that $h(u_i) \in \mn{adom}(\Dmc_0)$ and $y_i \in Y_{u_i}$,
        for $0\leq i \leq k$.
 \end{claim}
%
%To prove this, it is enough to show that there is a variable $u_i$ in $q$
%such that $y_i \in Y_{u_i}$ and $u_i$ belongs to a simple cycle in $q$.
%Indeed, since $y_i$ belongs to a simple cycle in $q$ and $\Umc_{\Dmc, \Omc} \setminus \Dmc$
%has a forest\=/like structure with no multiedges, this cycle has to be mapped to $\dom(\Dmc_0)$ and, thus,
%$h(u_i) \in \mn{adom}(\Dmc_0)$. In fact, we can also show an analogue of Claim\ref{claim:paths-in-core}.
%
%\begin{claim}
%    \label{claim:paths-in-core-hyper}
%    Let $z_0, \dots, z_\ell$ be a simple path in $q$.
%    If $h(z_0),h(z_\ell) \in \DGamma$ then $h(z_i)\in \DGamma$
%    for all $0 \leq i \leq \ell$.
%\end{claim}
%
%The proof of the claim mirrors that of Claim~\ref{claim:paths-in-core}.
\noindent
\begin{proof}
  We prove the claim for the case $y_i=y_1$, the other cases are
  symmetric. By Point~2 of Claim~\ref{claim:paths-in-core-triangle}, it suffices
  to show that $q$ contains a simple cycle and a variable $u_i$ on the
  cycle with $y_i \in Y_{u_i}$.

%let $y_{i_1} \in Y$. Take any $y_{i_2} \in Y$
% different from $y_{i_1}$.  Since there is no atom in $q^+$ that contains
% all variables in $Y$, there is a variable $y_{i_3} \in Y$ such that there is
% no atom in $q^+$ using all of $y_{i_1}$, $y_{i_2}$, and $y_{i_3}$.  Thus,
% there is a cycle in $q$ using all of $y_{i_1}$, $y_{i_2}$, and $y_{i_3}$.
% Intuitively, we use this cycle to extract a simple cycle crossing a
% variable $u$ such that $y_{i_1} \in Y_{u}$.  The construction of this
% cycle is not difficult, but technical.
%
%
%By symmetry, without loss of generality we assume that $y_{i_1}
%=y_1$, $y_{i_2} =y_2$, and $y_{i_3} =y_3$.
Consider the variables $y_1,y_2,y_3$. We first show that $q$
contains atoms $R_1(v^2_1,v^3_1)$,
$R_2(v^1_2,v^3_2)$, and $R_3(v^1_3,v^2_3)$ such that
for every variable $v^j_i$:
\begin{enumerate}

\item $y_j \in Y_{v^j_i}$ via a functional path $\pi^j_i$ in $q$ from
  $v^j_i$ to $y_j$,

\item $y_{i} \notin Y_{v^j_i}$ (and  thus $y_i \notin Y_x$ for any
  variable $x$ on path~$\pi^j_i$),

\item the concatenated path $(\pi^{j_1}_i)^-\pi^{j_2}_i$ is simple
  where $\{i,j_1,j_2\}=\{1,2,3\}$ and $\cdot^-$ denotes path reversal.

\end{enumerate}
In fact, by choice of $Y$, $q^+$ contains atoms
$\alpha_1,\alpha_2,\alpha_3$ such that for $i,j \in \{1,2,3\}$
$\alpha_i$ contains $y_j$ iff $i \neq j$. This implies in particular
that $y_i \notin Y_z$ for any variable $z$ in $\alpha_i$.  By
construction of $q^+$, we find in $q$ (unary or binary atoms)
$\beta_1,\beta_2,\beta_3$ such that $\beta_i$ contains only variables
from $\alpha_i$ and every variable in $\alpha_i$ is reachable in a
functional path in $q$ from a variable in $\beta_i$.  Thus
$y_i \notin Y_z$ for any variable $z$ in $\beta_i$. Moreover, for
$\{i,j_1,j_2\}=\{1,2,3\}$ the atom $\beta_i$ contains (possibly
identical) variables $w_i,w'_i$ such that $q$ contains a functional
path $\pi$ from $w_i$ to $y_{j_1}$ and $\pi'$ from $w'_i$
to~$y_{j_2}$. If these two paths share no variables, then we can
choose $R_i(v^{j_1}_i,v^{j_2}_i)$ to be $\beta_i$, $\pi^{j_1}_i$ to be
$\pi$, and $\pi^{j_2}_i$ to be $\pi'$. Otherwise, we choose
$v^{j_1}_i$ to be the last variable shared by the two paths,
$R_i(v^{j_1}_i,v^{j_2}_i)$ to be some atom that connects $v^{j_1}_i$
with its successor node on $\pi'$, $\pi^{j_1}_i$ to be the
suffix of $\pi$ starting at $v^{j_1}_i$, and $\pi^{j_2}_i$ to be the
suffix of $\pi'$ starting at $v^{j_2}_i$.

% Note that it follows from Point~2 that $y_{i} \notin Y_{u}$ for every
% variable $u$ on $\pi^j_i$.
Let $z_{1}$ be the first variable shared
by paths $\pi^1_3$ and $\pi^1_2$, $z_{2}$ the first variable on
$\pi^2_1$ and $\pi^2_3$, and $z_{3}$ the first variable on $\pi^3_2$
and $\pi^3_1$. This is illustrated in Figure~\ref{fig:ex-hypercyc}.
\begin{figure*}[t]
    %\begin{wrapfigure}{r}{0.45\textwidth}
    \centering
    \begin{tikzpicture}[scale=.8]

        %delimiters
        %%%%%%%%%%%%%%%grid
%        \draw[dashed,gray] (10,-1.5)--(10,2);
%        \draw[dashed,gray] (0,-1.5)--(20,-1.5);

        %%%%%%%%%%%%query
%        \node (q) at (5,2) {query $q$};

        \node (y1) at (3.5,7) {$y_1$};
        \node (y2) at (0 ,0) {$y_2$};
        \node (y3) at (7, 0) {$y_3$};
        \node (z1) at (3.5,6) {$z_1$};
        \node (z2) at (1 ,1) {$z_2$};
        \node (z3) at (6, 1) {$z_3$};
        \node (v12) at (3, .5) {$v_1^2$};
        \node (v13) at (4, .5) {$v_1^3$};
        \node (v21) at (5, 4) {$v_2^1$};
        \node (v23) at (5.5, 3) {$v_2^3$};
        \node (v31) at (2, 4) {$v_3^1$};
        \node (v32) at (1.5, 3) {$v_3^2$};

        \draw[red, -latex] (z1) to [out=150,in=210] (y1);
        \draw[blue, -latex] (z1) to [out=30,in=330] (y1);

        \draw[-latex] (z2) to [out=150,in=120] (y2);
        \draw[blue, -latex] (z2) to [out=270,in=330] (y2);

        \draw[red, -latex] (z3) to [out=270,in=180] (y3);
        \draw[-latex] (z3) to [out=0,in=60] (y3);

        \draw[red, dashed, -latex] (v21) to  (z1);
        \draw[blue, dashed, -latex] (v31) to  (z1);

        \draw[dotted, -latex] (v12) to  (z2);
        \draw[blue, dotted, -latex] (v32) to  (z2);

        \draw[-latex] (v13) to  (z3);
        \draw[red, -latex] (v23) to  (z3);

        \draw[] (v12)--(v13);
        \draw[red] (v21)--(v23);
        \draw[blue] (v31)--(v32);
%        \draw[-latex] (z1)--(y1) node[midway, below] {$f_0$};
%        \draw[-latex] (z1)--(y1) node[midway, above] {$f_1$};
%        \draw[-latex] (y1)--(y2) node[midway, above] {$R_{12}$};
%        \draw[-latex] (y2)--(y3) node[midway, left] {$R_{23}$};
%        \draw[-latex] (y3)--(y0) node[midway, above] {$R_{30}$};

    \end{tikzpicture}
    \caption{Illustration of the proof of existence of a simple cycle containing $u_i$.
    The directed edges symbolise functional paths between variables.
    Paths composed of the same colour sections ({\color{red} red}, {\color{blue} blue}, {\color{black} black})
    or of the same structure sections (solid, dashed, dotted) are simple.
%    \cMP{there are 3 situations: it can be $z_i=y_i$, the paths from $z_i$ to $y_i$ coincide
%    or the paths are different.}
    }

    \label{fig:ex-hypercyc}
\end{figure*}
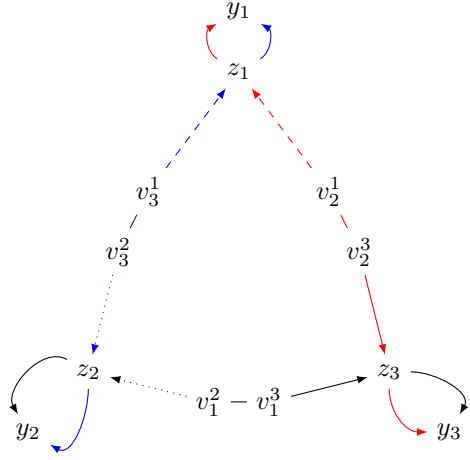 Now consider the following cycle
$C$:
\begin{itemize}
    \item path $\pi_3^1$ from $v_3^1$ to $z_1$, reverse of path $\pi_2^1$ from $z_1$ to $v_2^1$,
    \item path $\pi_3^2$ from $v_2^3$ to $z_3$, reverse of path
      $\pi_1^3$ from $z_3$ to~$v_1^3$,
    \item path $\pi_1^2$ from $v_1^2$ to $z_2$, reverse of path $\pi_3^2$ from $z_2$ to $v_3^2$.
\end{itemize}
If $C$ is simple, then we can choose $u_i=v^1_3$ and are done.
Otherwise, we can extract from $C$ a simple cycle that still
contains $u_i=v^1_3$, as follows.

Let $v$ be the first variable on $C$ (regarding the order in which $C$
is presented above) that occurs more than once on~$C$, let its
first appearance be on (possibly reversed) $\pi_i^j$ and the last
appearance on (possibly reversed)~$\pi_\ell^m$. Note that
\begin{enumerate}

\item[(a)] $\ell \neq j$ because $v$ being on $\pi_i^j$ implies that
  $y_{j} \in Y_v$ while $v$ being on $\pi_\ell^m$ implies
    $y_{\ell} \notin Y_v$ by Point~2 above;

\item[(b)] $\ell \neq i$ due to Point~3 above.

\end{enumerate}
Points~(a) and~(b) imply that the first appearance of $v$ on $C$ is
before~$z_3$. It also implies that if the first appearance is
before~$z_3$, then any other appearance (including the last) is after
$z_3$. We can thus remove the subcycle between the first and last
occurrence without removing
$u_i=v^1_3$. % Repeating this argument yields a simple cycle.
% \cMP{no repeat needed!}
% and finishes the proof of \textbf{HSC}.

The resulting cycle is simple.  To see this, assume to the contrary
that there is a variable $w$ with repeated appearance. Then $w \neq v$
and the first appearance of $w$ cannot be before the (only) appearance
of $v$ by choice of $v$. It can also not be after the appearance of
$v$ since then all appearances of $w$ would be after $z_3$ in the
original cycle $C$ and Points~(a) and~(b) above imply that this is
impossible.
\end{proof}

Now we can prove \textbf{HD1}. Let $h$ be a homomorphism from $q$ to
$\Umc_{\Dmc, \Omc}$. By Claim~\ref{claim:hypercliqueui} for
$0\leq i \leq k$ there is a variable $u_i$ in $q$ such that
$y_i \in Y_{u_i}$ and $h(u_i) \in \DGamma$. The following claim
is proved in exactly the same way as
Claim~\ref{claim:cycle-well-guided-hom},
details are omitted.
The proof is identical and omitted.
    \begin{claim}
        \label{claim:cycle-well-guided-hom-hyperclique}
        Let $y,y'$ be variables in $q$ such that %$y \neq y'$ and
        $h(y),h(y') \in \mn{adom}(\Dmc_0)$.  Then,
        $h(y) = \langle y, f_y\rangle$ and
        $h(y) = \langle y', f_{y'}\rangle$ for some functions $f_y$
        and $f_{y'}$.  Moreover, if $y_i \in Y_{y} \cap Y_{y'}$ then
        $f_y(y_i) = f_{y'}(y_i)$.
    \end{claim}
    Let $f = \bigcup_{0 \leq i \leq k} f_{u_i}$. It follows from
    Claim~\ref{claim:cycle-well-guided-hom-hyperclique} that $f$ is a
    function from $Y$ to $V$.  For brevity, let $v_i = f(y_i)$ for
    $0 \leq i \leq k$. To end the proof it suffices to show that for
    every set $I \subseteq \{0, \dots, k\}$ of size $k$, the set
    $V_I = \{v_i \in V \mid i \in I\}$ is a hyperedge in $G$, i.e.\
    $V_I \in E$.
%$0\leq i < j \leq k$
%we have that $\{v_i,v_j\} \subseteq e$ for some hyperedge $e\in G$. This will end the proof.

    Since $Y_I = \{y_i \in Y \mid i \in I\} \subseteq Y$, there is an
    atom $R^{+}_I(\bar{z})$ in $q^+$ that contains all variables in
    $Y_I$.  We show that $q$ contains an atom $R(z,z')$ such that
    $Y_I \subseteq Y_{z} \cup Y_{z'}$ and $h(z), h(z') \in \DGamma$.
    Before we prove the existence of this atom, we show how this
    implies that $V_I \in E$.  It follows from
    $h(z), h(z') \in \DGamma$ that $R(h(z), h(z')) \in \Dmc_0$.  By
    Claim~\ref{claim:cycle-well-guided-hom-hyperclique} and the
    definition of $\Dmc_0$, there is thus a word $w=b_1 \cdots b_k$
    such that $h(z) = \langle z , f_z^w\rangle$,
    $h(z') = \langle z',f_{z'}^w\rangle$, and $w$ is a permutation of
    vertices in an edge in $G$.  The definition of $f^w_z$ and
    $f^w_{z'}$ together with the fact that
    $Y_I \subseteq Y_{z} \cup Y_{z'}$ yields $v_i = f(y_i) = b_i$ for
    all $i \in I$.  Thus,
    $V_I = \{v_i\}_{i \in I} = \{b_1 ,\dots,b_k\}$.  This proves that
    $V_I$ is an edge in $G$, as desired.

\smallskip

Now we prove the existence of the atom $R(z,z')$.  By the construction
of $q^+$, $q$ contains a (unary or binary) atom $\alpha$ and two (not
necessarily distinct) variables $w,w'$ in $\alpha$ such that for every
$y_i \in Y_I$, there is a functional path in $q$ from $w$ to $y_i$ or
from $w'$ to $y_i$.  Since $q$ is connected, every variable occurs in
a binary atom, and thus we may actually assume that $\alpha$ is of the
form $R'(z_0,z'_0)$. Then
$Y_I \subseteq Y_{z_0} \cup Y_{z'_0}$, that is, for
every variable $y_i \in Y_I$ there is a simple functional path $\pi_i$
in $q$ from $z_0$ to $y_i$ or there is a simple functional path
$\pi'_i$ in $q$ from $z'_0$ to $y_i$.

We next show that there is a variable $v$ such that
% belongs to all paths
% $\pi_i$ and satisfies
$h(v) \in \DGamma$ and $q$ contains a %simple
path from $z_0$ to $v$
such that for all variables $x$ on the path, $Y_{z_0} \subseteq Y_x$.
Likewise, there is a variable $v'$ such that $h(v') \in \DGamma$ and
$q$ contains a %simple
path from $z'_0$ to $v'$ such that for all
variables $x$ on the path, $Y_{z'_0} \subseteq Y_x$.  We only do the
case of $v$ explicitly, the case of $v'$ is symmetric.

If $Y_{z_0}=\{y_i\}$ is a singleton, we choose $v=u_i$ and are done.
Note that there is a functional path $\pi$ from $z_0$ to $y_i$ and a
functional path $\pi'$ from $u_i$ to $y_i$, and the required path from
$z_0$ to $v$ is $\pi{\pi'}^-$.

Otherwise, let $Y_{z_0}=\{y_{i_1},\dots,y_{i_m}\}$, $j \geq 2$. We
prove the following by induction on $i$.

\begin{claim}
  For $2 \leq i \leq m$, there is variable $w_i$ in $q$ such that
  $\{y_{i_1},\dots,y_{i_i} \} \subseteq Y_{w_i}$,
  $h(w_i) \in \mn{adom}(\Dmc_0)$, and there is a simple functional
  path $\pi_i$ from $z_0$ to $w_i$.
\end{claim}
We may then choose
$v=w_m$.
\begin{proof}
  The induction start is $i=2$. Since $y_{i_1},y_{i_2} \in Y_{z_0}$,
  we find a simple functional path $\tau_1$ from $z_0$ to $y_{i_1}$
  and a simple functional path $\tau_2$ from $z_0$ to
  $y_{i_2}$. Clearly, all variables $x$ on $\tau_1$ satisfy
  $y_{i_1} \in Y_x$ and all variables $x$ on $\tau_2$ satisfy
  $y_{i_2} \in Y_x$.  We choose $w_2$ to be the last variable on the
  longest common prefix of $\tau_1$ and $\tau_2$ (note that these
  paths end in $y_{i_1}$ and $y_{i_2}$, which are distinct), and
  $\pi_2$ to be this prefix. Clearly,
  $\{y_{i_1},y_{i_2} \} \subseteq Y_{w_2}$ and thus it remains to show
  that $h(w_2) \in \mn{adom}(\Dmc_0)$.

  Let $\omega_1,\omega_2$ be the parts of $\tau_1$ and $\tau_2$ that
  remain after the common prefix $\pi_2$. Further recall that there is
  a simple functional path $\rho_1$ from $u_{i_1}$ to $y_{i_1}$ and a
  simple functional path $\rho_2$ from $u_{i_2}$ to $y_{i_2}$. Now
  consider the path $\mu= \rho_1^-\omega_1^-\omega_2\rho_2$ from
  $u_{i_1}$ to $u_{i_2}$.  We can make the $\rho_1^-\omega_1^-$ prefix
  simple by removing subpaths and we can make the $\tau_2\omega_2$
  suffix simple in the same way.

  First assume that the resulting path $\mu'$, which must still
  contain $w_2$, is simple. Since
  $h(u_{i_1}), h(u_{i_2}) \in \mn{adom}(\Dmc_0)$, Point~1 of
  Claim~\ref{claim:paths-in-core-hyperclique} yields
  $w_2 \in \mn{adom}(\Dmc_0)$ and we are done.

  Now assume that $\mu'$ is not simple. Then some variable
  $x \neq w_2$ must occur on the path $\omega_1\rho_1$ and on the path
  $\omega_2\rho_2$. This gives rise to a cycle $C$ in $q$, starting
  from $w_2$ on $\omega_1\rho_1$ to $x$ and back on $\rho_2\omega_2$
  to $w_2$. Note that this must indeed be a cycle since it contains
  $w_2$ as well as the successors of $w_2$ on the paths $\tau_1$ and
  $\tau_2$, and these three variables must be pairwise distinct by
  choice of $w_2$ and since $\tau_1,\tau_2$ are simple
  paths.\footnote{If the alledged cycle had only length two, then we
    would need more than one atom between the two involved variables,
    c.f.\ the definition of cycles.} % We may now drop subpaths from $C$
  % to obtain a simple cycle $C'$ that still contains $w_2$.
Point~2 of
  Claim~\ref{claim:paths-in-core-hyperclique} yields
  $w_2 \in \mn{adom}(\Dmc_0)$ and we are done.

  The induction step, where $i > 2$, is essentially identical to the
  induction start except that we now start from the simple functional
  path $\tau_1=\pi_{i-1}$ from $z_0$ to $w_{i-1}$ in place of the path
  $\tau_1$ from $z_0$ to $y_{i_1}$. All the remaining arguments are identical.
\end{proof}
Hence, $q$ contains a simple path $\pi$ from $z_0$ to $v$ such that
$h(v) \in \mn{adom}(\Dmc_0)$ and all variables $x$ on the path satisfy
$Y_{z_0} \subseteq Y_x$, and a simple path $\pi'$ from $z'_0$ to $v'$
such that $h(v') \in \mn{adom}(\Dmc_0)$ and all variables $x$ on the
path satisfy $Y_{z'_0} \subseteq Y_x$.  First assume that $\pi$ and
$\pi'$ share a variable and let $z$ be the first variable on $\pi^-$
shared with $\pi'$. Consider the path $\tau$ obtained by traveling on
$\pi^-$ from $v$ to $z$ and then on $\pi$ from $z$ to $v'$.  Clearly,
$\tau$ is simple and thus Point~1 of
Claim~\ref{claim:paths-in-core-hyperclique} together with the fact
that $h(v),h(v') \in \mn{adom}(\Dmc_0)$ implies that
$h(z) \in \DGamma$. Moreover we have $Y_I \subseteq h(z)$. Let $z'$ be
the successor of $z$ on $\tau$. Then $q$ must contain an atom
$R(z,z')$ which is the atom that we have been looking for.  Now assume
that $\pi$ and $\pi'$ are disjoint. Then $\pi^-\pi'$ is a simple path
from $v$ to $v'$ that contains both $z_0$ and $z'_0$. Point~1 of
Claim~\ref{claim:paths-in-core-hyperclique} yields
$R(h(z_0), h(z'_0)) \in \Dmc_0$ and $R(z_0,z'_0)$ is the promised
atom.

\subsection{Acyclic Queries that are not Free-Connex Acyclic}

We now prove Point~2 of Theorem~\ref{lemma:enumeration-lower-bound}.

\begin{lemma}
    \label{lemma:lower-bound-enum-MM}
    Let $Q(\bar{x}) = (\Omc, \Sigma, q) \in (\fEL,\text{CQ})$ be a
    self-join free, non\=/empty, connected OMQ such that $q^{+}$ is
    acyclic but not free\=/connex acyclic. Then enumerating the
    complete
    answers to $Q$
    is not in \dlc unless
    %
    % If there is an algorithm
    % that given an $\Sbf$\=/database after linear preprocessing
    % enumerates the set $Q(D)$ with constant delay then
    one can multiply Boolean matrices $M_1,M_2$ in time
    $O(|M_1| + |M_2| + |M_1M_2|)$.  The same is true for least partial
    answers, both with a single wildcard and with multi-wildcards.
\end{lemma}
Let $Q(\bar{x}) = (\Omc, \Sigma, q)$ be as in
Lemma~\ref{lemma:lower-bound-enum-MM}.  It is well known that if a CQ
is acyclic but not free\=/connex acyclic, then it contains a bad path,
that is, a path $y_0,\dots,y_k$ that is simple, chordless, and not a
cycle, and such that $y_0,y_k$ are answer variables while
$y_1, \dots, y_{k-1}$ are quantified variables.  See e.g.\
\cite{berkholz-enum-tutorial} for a proof. Thus $q^+$ contains a bad path
$y_0,\dots,y_k$. Note that $y_0$ and $y_k$ are answer variables in
$q^+$, but not necessarily in $q$.  By construction of~$q^+$, however,
there is an answer variable in $q$ that has a functional path to $y_0$
and a different one that has a~functional path to~$y_k$. We can assume
w.l.o.g.\ that these answer variables in $q$ are the first two answer
variables $x_1$ and $x_2$ in the tuple $\bar x$.

Let $M_1,M_2$ be $n {\times} n$-matrices. Our aim is to construct a
database \Dmc, in time $O(|M_1|+|M_2|)$, such that a \dlc algorithm
for enumerating the (complete or partial) answers to $Q$ on \Dmc
allows us to compute $M_1M_2$ in time $O(|M_1| + |M_2| + |M_1M_2|)$.
The reduction shares ideas, intuition, and notation with the previous
two reductions. In particular, for every variable $x$ in $q$ the set
$Y_x$ is defined as in the previous reduction and we again define the
desired database \Dmc by starting with a `core part' $\Dmc_0$ that
encodes matrix multiplication and then add trees in a second step.
Before we give details, we observe the following properties of the
sets $Y_x$.
% Let the CQ $q'$ be
% obtained from $q$ by making $y_0$ and $y_n$ answer variables. Due
% to the following lemma, we can replace $q$ with $q'$ in the remaining
% proof, that is, we can simply assume that $y_0$ and $y_n$ are answer
% variables in $q$
% %
% \begin{lemma}
%   If $Q(\Dmc)$ can be enumerated in \dlc, then so can $Q'(\Dmc)$.
% {\color{blue}This fails for partial answers. :(}
% \end{lemma}
% %
% \begin{proof}
%   Since $y_0,\dots,y_n$ is a bad path in $q^+$, it follows from
%   the construction of $q^+$ that
% %
%   \begin{enumerate}

% \item $q$ contains a path from $y_0$ to $y_n$;

% \item $q$ contains a functional path from some answer variable $x_0$
%   to $y_0$;

% \item $q$ contains a functional path from some answer variable $x_1$
%   to $y_n$.

% \end{enumerate}
% %
% {\color{blue}Needs to be fleshed out:} Thus $q$ contains a path $p$
% from $x_0$ to $x_1$ that goes through $y_0$ and $y_n$.  Let \Dmc be an
% \Sbf-database. Since $q$ is self-join free, every homomorphism from
% $q$ to $\Umc_{\Dmc,\Omc}$ must map all variables on $p$, including
% $y_0$ and $y_n$, to database constants. This can be used to show that
% there is a one-to-one correspondence between the answers in $Q(\Dmc)$
% and $Q'(\Dmc)$ where $Q'=(\Omc,\Sbf,q')$.
% \end{proof}
% %
% For every variable $x$ in $q$, we use $V_x$ to denote the set of
% variables $y$ such that $q$ contains a functional (possibly empty) path from $x$
% to~$y$. Further, let $Y = \{y_0,\dots, y_k\}$ and set
% $Y_x = Y \cap V_x$.

%    By the construction of $q^{+}$ we observe the following.
    \begin{claim}
      \label{claim:MM-small-Y}
      ~\\[-4mm]
      \begin{enumerate}
      \item For every variable $x$ in $q$,
        $Y_x \subseteq \{y_{i}, y_{i+1}\}$ for some $i$ with
        $0\leq i < k$.  .
      \item For every atom $R(x,y)$ in $q$, $Y_x \cup Y_y \subseteq
        \{y_{i}, y_{i+1}\}$
        for some $i$ with $0\leq i < k$.
      \item $Y_{y_i} = \{y_i\}$ for $0\leq i \leq k$.
      \item For every answer variable $x'$ in $q$,
        $Y_{x'} \in \{ \emptyset, \{y_0\}, \{y_k\} \}$.

      \end{enumerate}

    \end{claim}
    \begin{proof}
     All points follow from the fact that the bad path
      $y_0,\dots,y_k$ is chordless and not a cycle.

      For Point~1, let $x$ be a variable in $q$ and assume that
      $Y_x \supseteq \{y_j,y_i\}$ and
      $y_j \notin \{ y_{i-1},y_{i+1}\}$ with $y_{-1}=y_k$. By
      definition of $Y_x$, there is then a functional path in $q$ from
      $x$ to $y_j$ and from $x$ to $y_i$. Thus, by construction of
      $q^+$ the Gaifman graph of $q^+$ contains the edge
      $\{y_j,y_i\}$, in contradiction to $y_0,\dots,y_k$ being
      chordless.

      The proof of Point~2 is similar.

      For Point~3, first note that $y_i \in Y_{y_i}$ by definition of
      $Y_{y_i}$. By Point~1, this leaves as the only possible options
      $Y_{y_i}= \{ y_i\}$, $Y_{y_i}=\{y_i,y_{i-1}\}$ (with
      $y_{-1}=y_k$), and $Y_{y_i}=\{ y_i,y_{i+1}\}$. The latter two
      options, however, are not possible. We show this exemplarily for
      the last option.  By construction of $q^+$, the fact that
      $\{ y_i,y_{i-1}\}$ is an edge in the Gaifman graph of $q^+$
      implies that $q$ contains an atom $\alpha$ that contains
      variables $u,v$ such that $y_{i-1}$ is reachable on a functional
      path from $u$ and $y_{i}$ is reachable on a functional path from
      $y_i$. Since $y_{i+1} \in Y_{y_i}$, there is also a functional
      path from $y_i$ to $y_{i+1}$ and thus we find a functional path
      from $v$ to $y_{i+1}$. By construction of $q^+$, the Gaifman
      graph of $q^+$ contains the edge $\{ y_{i-1},y_{i+1}\}$, in
      contradiction to the chordfreeness of $y_0,\dots,y_k$.

         For Point~4 observe that the variables
         $y_1,\dots,y_{k-1}$ are quantified. Thus, by construction of
         $q^+$  (which is
         $q^+(\bar x^+)$),
         $q$ contains no functional path from an answer variable in
         $q$ to any of these variables. Consequently, if $x'$ is an
         answer variable in $q$, then $Y_{x'} \subseteq
         \{y_0,y_k\}$. But by construction of $q^+$,
         $Y_{x'} = \{y_0,y_k\}$ would imply that the Gaifman graph of
         $q^+$ contains the edge $\{y_0,y_k\}$, in contradiction to
         $y_1,\dots,y_k$ not being a cycle in~$q^+$.  Thus
         $Y_{x'} \subsetneq \{y_0,y_k\}$ and we are done.
%    	If there was a functional path from $y_0$ (or from $y_k$) to any of $y_i$, $0<i<k$, then by the assumption
%    	in Claim~\ref{claim:well-formed-query} $y_i$ would be an answer variable which contradicts the requirements of a bad path.
%
%    	Now, let $0 \leq i  \leq k$ be an index. If there was a functional path to $y_i$ from any of $y_j$, where $j \neq i$, then by construction of $q^{+}$,
%    	$y_i$ and $y_j$ would be neighbours in the Gaifman graph of $q^{+}$ and, thus, the claim follows.
%
%        For the last part of the claim fix a variable $x$ and observe that if $y_i, y_j \in V_x$. Then, by definition, there is a functional path from $x$ to $y_i$
%        and a functional path from $x$ to $y_j$. Thus, there is an atom in $q^{+}$ that contains both $y_i$ and $y_j$ and, thus, they have to be two consecutive variables
%        on the bad path.
    \end{proof}

    % We will use the bad path $y_0, \dots, y_k$ to encode the matrix multiplication.
    %

    To give an intuition of the reduction, consider the CQ
    $q_0(y_0,y_2)= \exists y_1 \, R_1(y_0,y_1) \wedge R_2(y_1,y_2)$,
    which represents a `paradigmatic' bad path, and assume that the
    ontology is empty. To compute the product $M_1 M_2$, we can use
    the database
$$
\begin{array}{rcl}
\Dmc &=& \{ R_1(\langle y_0,a\rangle,\langle y_1,b\rangle) \mid (a,b) \in M_1 \}
  \; \cup\\[1mm]
  && \{ R_2(\langle y_1,b\rangle,\langle y_2, c\rangle) \mid (b,c) \in M_2 \}.
\end{array}
$$
Clearly, the answers to $q_0$ on \Dmc, projected to the second
component, are exactly the 1-entries in $M_1 M_2$.  The above database
\Dmc may be viewed as the direct product of $q_0$ with the database
$\{ R_i(a,b) \mid (a,b) \in M_i, \ i \in \{1,2\} \}$.

Also in general (and as in the previous two reductions), the
constructed database \Dmc may be viewed as a variation of the direct
product. A homomorphism $h$ from $q$ to $\Dmc$ that maps every
$y_i \in Y$ to a constant of the form $\langle y_i,f_{y_i} \rangle$
gives rise to a pair $(a,b) \in M_i$ as follows. By Point~3 of
Claim~\ref{claim:MM-small-Y}, the domain of the function $f_{y_i}$ is
$\{y_i\}$, and by construction of \Dmc $f_{y_i}(y_i)$ will be an
element of $[n]$. The homomorphism $h$ thus identifies a sequence of
elements $a_0,\dots,a_k \in [n]$, with $a_i=f_{y_i}(y_i)$. The
construction of $\Dmc$ ensures that then $(a_0,a_1) \in M_1$,
$a_1=\cdots=a_{k-1}$, and $(a_{k-1},a_k) \in M_2$, and thus
$(a_0,a_k) \in M_1M_2$. As in the previous two reductions, however, we
cannot ensure that $h$ maps every $y_i$ to a constant of the form
$\langle y_i,f_{y_i} \rangle$. This leads to similar technical
complications as in the previous reductions. Conversely, all pairs
$(a,c) \in M_1$ and $(c,b) \in M_2$ give rise to a homomorphism
from $q$ to \Dmc.

\smallskip

We now construct the database \Dmc.

   \paragraph{Step~1: Encoding the product of $M_1$ and $M_2$.}

   We use the same functions $f^w_x$ % In this step, we use constants of the form $\langle x,f \rangle$
   % where $x$ is a variable and $f$ is a function from $Y_x$ to
   % $[n]$. For every variable $x$ in $q$ and word
   % \mbox{$w = a_0a_1 \dots a_{k} \in [n]^{\ast}$}, let
   also employed in the two previous reductions.
   % $f^w_x$ denote the
   % function that maps each variable $y_i \in Y_x$ to $a_i$.
   The database $\Dmc_0$ contains the following facts,
   %
    % The domain $\dom(\Dmc)$ of the database
    % \Dmc consists of pairs $\langle x,f \rangle$ where $x$ is a
    % variable and $f$ is a function from $Y_x$ to $[n]$.partial
    % functions from the set $Y$ to the set $\{1, \dots, n\}$; and some
    % fresh database constants that will be used as a ``padding'' that
    % assures that for every relation symbol in the schema we have that
    % every constant has both an ingoing and an outgoing edge labelled
    % by that symbol.
%
    % We first add the facts, with relation symbols in the schema, that encode the matrix multiplication
    % and satisfy the functional dependencies enforced by the ontology.
   %
   for every atom $r(x,y)$ in $q$ with $r\in \Sigma$:
     \begin{description}
                \item[MR1] if $y_0 \in Y_x \cup Y_y$, the fact
                  $r(\langle x, f^{ab^k}_x\rangle, \langle y ,
                  f^{ab^k}_y \rangle)$ for all $(a,b) \in M_1$,
                \item[MR2] if $y_k \in Y_x \cup Y_y$, the fact $r(\langle x, f^{b^kc}_x\rangle, \langle y ,
                  f^{b^kc}_y \rangle)$ for all $(b,c) \in M_2$,
                \item[MR3] otherwise, the fact $r(\langle x, f^{b^{k+1}}_x\rangle, \langle y ,
                  f^{b^{k+1}}_y \rangle)$ for all $b$ such that $(a,b) \in M_1$ or $(b,c) \in M_2$.
    \end{description}
    It also contains $A(c)$ for every concept name $A \in \Sigma$ and
    every constant $c \in \mn{adom}(\Dmc_0)$. %  For easy reference, we
    % use $\Gamma$ to denote the set of constants in $\Dmc_0$.
    % Before
    % we proceed, let us give an example of the construction.
    % \begin{example}

    %  \end{example}

     % The database $\Dmc_0$ constructed in Step~1 may not deliver any
     % answers to $q_0$ because it contains only $\Sigma$-facts whereas
     % $q_0$ may also use symbols from outside of $\Sigma$. While we are
     % not allowed to include symbols from outside of $\Sigma$ in \Dmc,
     % we can extend \Dmc to make sure that `as many non-$\Sigma$-facts
     % as possible' are derived via \Omc.  This is exactly what is
     % achieved
     % by derivation completeneness.

%    If the schema is not full, we need to augment our definition of the database.
%    We say that a constant in a database is \emph{relevant} if it is used in some fact.
%    Let $p_r^K$ be an $\Sbf$\=/database, called a~\emph{rod}, of domain $[K]$ with the following set of facts
%    \begin{itemize}
%        \item $r(i,i+1)$  for $i \in [K-1]$,
%        \item $s(i,i)$ for $s \in \Sbf \setminus \{r\}$ and $i \in [K]$.
%    \end{itemize}
%
%    For every relation $r \in \Sbf$ and every constant $c$ used in one of the facts added above, we add a fresh copy of $p_r^K$
%    anchored to $c$, i.e. \begin{itemize}
%        \item if there is no fact $r(c, c')$ in $\Dmc$ then we add to $\Dmc$  a copy of $p_r^K$ and fact $r(c,c')$ where $c'$
%        is the constant $1$ in the fresh copy of $p_r^K$;
%        \item if there is no fact $r(c', c)$ in $\Dmc$ then we add to $\Dmc$ a copy of $p_r^K$ and fact $r(c,c')$ where $c'$
%        is the constant $K$ in the fresh copy of $p_r^K$.
%    \end{itemize}.

    \paragraph{Step~2: Adding trees.}
    This step is the same as in the previous two reductions.  We
    extend the database $\Dmc_0$ to the desired database \Dmc as
    follows: for every $c \in \mn{adom}(\Dmc_0)$ and every role
    $R \in \{ r, r^-\}$ with $r \in \Sigma$ such that there is no fact
    $R(c, c') \in \Dmc_0$, add a disjoint copy of $\Dmc_{R}$, glueing
    the copy of $\varepsilon$ to $c$.
    %
    % We further add, for every constant  $c \in \mn{adom}(\Dmc)$,
    %
    % set $\Gamma = \mn{adom}(\Dmc)$.
    % Let $\Gamma$ be the set of database constants added in the above construction.
    % For every unary relation symbol $A$ in $\Sbf$ and every constant $c \in \Gamma$ we add fact $A(c)$.

    % Intuitively, we add trees only to database constants $(x, f)$ where $f$ is not a function from $Y_x$ to $[n]$
    % or to constants that create ``bridges'' between answer variables $y_0, y_k$ and the remaining variables in $Y$.

%    The database restricted to constants in $\Gamma$ satisfies the
%    functionalities.
    The following is immediate from the construction of $\Dmc_0$, by a
    straightforward analysis of the Cases~1-3 in the construction of
    $\Dmc_0$ that may introduce an atom $R(a,b)$.
    \begin{claim}
      \label{claim:MM-functional-database}
%       Let $R(x,y) \in q$ be such that % an atom that uses a relation symbol from
%       % $\Sigma$ and satisfies
%       $Y_x \subseteq Y_y$ and let
%       \mbox{$R(a,b) \in
%       \Dmc_0$},  % set $F = \{(a,b) \in \Gamma^2 \mid R(a,b) \in D\}$ is a function.
%       Then $a= \langle x,f_x\rangle $ and
%       $b=\langle y,f_y \rangle$ where $f_x,f_y$ are functions
%       $\fun{f_x}{Y_x}{[n]}$, $\fun{f_y}{Y_y}{[n]}$ such that
%       $f_x = {f_y}|_{Y_x}$.
% %        Moreover, if $Y_x = Y_y$ then $F \subseteq \{(a,a) \mid a \in \Gamma\}$.
%
%       Let $R(x,y) \in q$ and \mbox{$R(a,b) \in
%         \Dmc_0$}, % set $F = \{(a,b) \in \Gamma^2 \mid R(a,b) \in D\}$ is a function.
%       Then $a= \langle x,f_x\rangle $ and $b=\langle y,f_y \rangle$
%       for functions $f_x,f_y$ such that $f_x(y_i)=f_y(y_i)$ for all
%       $y_i \in Y_x \cap Y_y$.
%
      Let $R(a,b) \in \Dmc_0$ with $a= \langle x,f_x\rangle $ and
      $b=\langle y,f_y \rangle$.  Then $f_x(y_i)=f_y(y_i)$ for all
      $y_i \in Y_x \cap Y_y$.
    \end{claim}
    We next summarize some important properties of the database \Dmc.
    \begin{lemma}
    \label{claim:MM-database-is-ok}
    % The following hold.
    The database \Dmc
    \begin{enumerate}
        \item can be computed in time $O(|M_1| + |M_2|)$.
%        \item Thus, $|\Dmc|$ is $O(|M_1| + |M_2|)$.
        \item satisfies the functional assertions in~$\Omc$.
        \item is derivation complete at $\mn{adom}(\Dmc_0)$.

        \end{enumerate}
\end{lemma}
%
%\begin{proof}
Point~1 should be clear by construction of \Dmc and Points~2 and~3 can
be proved in the same way as the corresponding points in
Claim~\ref{claim:MM-database-is-ok}.
  % , observe that $|\Dmc_0|$ is
  % $O(|\Sigma| \cdot (|M_1| + |M_2|))$, thus $|\dom(\Dmc_0)|$ is
  % $O(|\Sigma| \cdot (|M_1| + |M_2|))$.  Moreover, since $\Dtree$
  % depends only on the OMQ and in Step~2 we add no more than
  % $|\Sigma|\cdot |\Dmc_0|$ copies of $|\Dtree|$, the overall
  % construction can be done in the claimed time.

%   For Point~2, take a role name $r$ such that $\func(r) \in \Omc$, and
%   take a constant $c$ such that there are $d$ and $d'$ such that
%   $r(c,d) \in \Dmc$ and $r(c,d') \in \Dmc$. We show that $d=d'$.

%   If one of $c$, $d$ or $d'$ belongs to a copy of $\Dtree$ then
%   clearly $d =
%   d'$. % Now, if $d$ (or $d'$) belongs to a copy of a tree, again, clearly $d = d'$.
%   Thus, let us assume that $c,d,d' \in \DGamma$.  Then, there is an
%   atom $r(x,y)$ in $q$, for some variables $x$ and $y$.  Moreover, by
%   construction of $\Dmc_0$ there are functions $f_c,f_d,f_{d'}$ such
%   that $c = \langle x, f_c \rangle $, $d = \langle y, f_d \rangle$,
%   and $d' = \langle y, f_{d'} \rangle$.

%   Since $\func(r) \in \Omc$ we have that $Y_y \subseteq Y_x$ .  Thus,
%   by Claim~\ref{claim:MM-small-Y} for all $y_i \in Y_y$ we have that
%   $f_d(y_i) = f_c(y_i) = f_{d'}(y_i)$.  This implies that
%   $f_{d} = f_{d'}$ and, thus, $d=d'$.

%   The case where $\func(r^{-}) \in \Omc$ is proven analogously.
% \end{proof}
%
The following lemma makes precise the connection between answers to
$Q$ and the matrix product $M_1M_2$. Recall that $x_1$
and $x_2$ are the first two answer variables of $q$. For any complete
or partial answer $\bar a$, we use $a_1$ and $a_2$ to denote the first
two constants in $\bar a$.
    \begin{lemma}
        \label{lemma:lower-bound-partial-answers-and-matrix}
        % The following holds.
        ~\\[-4mm]
        \begin{description}
            \item[MM1] For every minimal partial answer $\bar{a}$
                to $Q$ on~$\Dmc$,
                if $a_1 = \langle x_1,
                f_{x_1}\rangle \in \DGamma$ and $a_2 = \langle
                x_2,f_{x_2}\rangle \in \DGamma$ then
             $(f_{x_1}(y_0),f_{x_2}(y_k)) \in M_1M_2$.
        \item[MM2] %For any pair $(a, b) \in [n] \times [n]$,
            For all $(a,b) \in M_1M_2$, there is a complete answer $\bar{a}$ to $Q$ on $\Dmc$
                such that $a_1 = \langle x_1, f_{x_1}\rangle $, $a_2 = \langle y_x,f_{x_2}\rangle$ with $f_{x_1}(y_0) = a$ and  $f_{x_2}(y_k) = b$.

            \item[MM3] There are no more than $O(|M_1| + |M_2| + |M_1M_2|)$ minimal partial answers to $Q$ on $\Dmc$.
            \end{description}
            This holds both for minimal partial answers with a single
            wildcard and with multi-wildcards.
    \end{lemma}

%    The first bullet is not trivial, neither is the second.
%
%    For the third bullet is ok.

    Before we prove the lemma, we let us explain how it is used to prove
    Lemma~\ref{lemma:lower-bound-enum-MM}.

    \medskip

%    \paragraph{The algorithm}
%    We can finally prove Lemma~\ref{lemma:lower-bound-enum-MM}.
    Assume that the complete or minimal partial answers to $Q$ (with a single
    or multiple wildcards) answers can be enumerated in \dlc.  We can
    then multiply matrices $M_1$ and $M_2$ in time
    $O(|M_1|+ |M_2| +|M_1M_2|)$ as follows.
    We build the database $\Dmc$ as described above.  Then, we
    enumerate the answers to $Q$ on $\Dmc$.  For every answer
    $\bar{a}$ with $a_1 = \langle x_1,f_{x_1}\rangle$ and
    $a_2 = \langle x_2, f_{x_2}\rangle$ we extract the pair
    $(f_{x_1}(y_0), f_{x_2}(y_k)) \in [n] \times [n]$. By \textbf{MM1}
    and \textbf{MM2}, the result of this extraction is exactly the set
    of pairs in $M_1 M_2$.
    Since $\Dmc$ can be constructed in time $O(|M_1|+|M_2|)$ and by
    ~\textbf{MM3} there are no more than $O(|M_1|+ |M_2| +|M_1M_2|)$
    answers, this algorithm produces $M_1M_2$ in time
    $O(|M_1|+|M_2|+|M_1M_2|)$. % The additional manipulations on the
    % answer set can be carried out in time $O(|M_1|+|M_2|+|M_1M_2|)$
    % and, thus, the overall time is $O(|M_1|+|M_2|+|M_1M_2|)$.

     % The correctness follows directly from the first two bullets of Claim~\ref{lemma:lower-bound-partial-answers-and-matrix}.

\medskip

To prepare for proof of
Lemma~\ref{lemma:lower-bound-partial-answers-and-matrix},
we make some technical observations.

    % As in Example~\ref{ex:simple-bad-path} we want to discover a function
    % $\fun{f}{Y}{[n]}$ such that $(f(y_0), f(y_1)) \in M_1$, $(f(y_1), f(y_k)) \in M_2$ with
    % $f(y_0) = f_{x_1}(y_0)$ and $f(y_k) = f_{x_2}(y_k)$.
    % The natural candidates for the values $f(y_i)$ would be $f'(y_i)$,
    % where $f'$ is a function such that $h(y_i) = \langle y', f'\rangle$.
    % Unfortunately, it may be the case that $y_i \notin \dom(f')$
    % or, even, that $h(y_i) \notin \dom(\Dmc_0)$ and, thus, $f'$ is undefined.

    % Nevertheless, it is possible to define the function $f$.
    % To do so, we start with two observations
    % concerning the homomorphisms from $q$ to $\Umc_{\Dmc, \Omc}$.
    % The first one states that for every simple path in $q$ with both ends mapped to $\dom(\Dmc_0)$
    % the image of  the path is all contained in $\dom(\Dmc_0)$.
    %
    \begin{claim}
        \label{claim:paths-in-core}
        Let $h$ be a homomorphism from $q$ to $\Umc_{\Dmc, \Omc}$ and
        $u_0, \dots, u_m$ a simple path in $q$, $m>0$.
        %        {\color{violet} If $h(u_0) \in \dom(\Dmc)$ and $h(u_m) \in \dom(\Dmc)$, then $h(u_i) \in \dom(\Dmc)$ for
            %            $0<i\leq m$. Moreover, if}
        If
        $h(u_0) \in \DGamma$ and $h(u_m) \in \DGamma$, then $h(u_i) \in \DGamma$ for
        $0\leq i\leq m$.
    \end{claim}

    \begin{proof}
        Assume towards a contradiction that $h(u_j) \notin \DGamma$ for
        some $j$ with $0 \leq j \leq m$ and additionally let $j$ be
        smallest with this property. Moreover, let $\ell>j$ be smallest
        such that $h(u_\ell)\in \DGamma$. Note that $0 < j < \ell < m$.
        By definition of a path, $q$ must contain atoms
        $R_j(u_{j-1},u_j),\dots,R_\ell(u_{\ell-1},u_\ell)$. Since $q$ is
        self-join free, the relation symbols  $R_j,\dots,R_\ell$ are
        all distinct.

        Next observe that by construction of \Dmc and by definition of
        universal models, the interpretation \Imc obtained from
        $\Umc_{\Dmc, \Omc}$ by removing all edges $r(a,b)$ with
        $a,b \in \DGamma$ is a collection of trees without multi-edges.
        More precisely, these are the trees added in the second step in
        the construction of \Dmc and the trees $\Umc_{\Dmc,
            \Omc}^{\downarrow c}$
        of $\Umc_{\Dmc, \Omc}$.
        We view as the roots of these trees the elements of $\DGamma$.

        Then $h$ maps $R_j(u_{j-1},u_j),\dots,R_\ell(u_{\ell-1},u_\ell)$
        to some tree of this collection, and both $u_{j-1}$ and $u_\ell$
        are mapped to the root of the tree. This, however, clearly
        contradicts the fact that the relation symbols
        $R_j,\dots,R_\ell$ are all distinct and the tree has no multi-edges.
    \end{proof}

    % Intuitively, the above claim will allow us to assign the value $f(y_i)$
    % for every $y_i \in Y$, as we will show that for every such $y_i$ there is
    % a variable $y$ such that $h(y) = \langle y, f_y\rangle$ and $y_i \in Y_i = \dom(f_y)$.
    % To assure that this assignment does not depend on the choice of $y$
    % we make the following observation.

    \begin{claim}
        \label{claim:MM-well-guided-hom}
        Let $y,y'$ be variables in $q$ such that $y \neq y'$ and $h(y),h(y') \in \DGamma$.
        Then, $h(y) = \langle y, f_y\rangle$ and $h(y) = \langle y', f_{y'}\rangle$
        for some functions \fun{f_y}{Y_y}{[n]} and \fun{f_{y'}}{Y_{y'}}{[n]}.
        Moreover, if $y_i \in Y_{y} \cap Y_{y'}$ then $f_y(y_i) = f_{y'}(y_i)$.
    \end{claim}

    \begin{proof}
        Since $q$ is connected, there is a simple path $z_0, \dots, z_\ell$ in $q$ from $z_0=y$ to $z_\ell=y'$.
        Since $h(y), h(y') \in \DGamma$, by the previous claim, all variables on this path
        are mapped to $\DGamma$. Thus, there are atoms $R_0(y, z_1)$ and $R_2(z_{\ell-1}, y')$
        such that $R_0(h(y), h(z_1)) \in \Dmc_0$ and $R_2(h(z_{\ell-1}), h(y'))\in \Dmc_0$.
        Those facts have been added by one of the rules \textbf{MR1}, \textbf{MR2}, or \textbf{MR3}
        in the construction of $\Dmc_0$ and, thus, $h(y) = \langle y, f_y\rangle$ and $h(y) = \langle y', f_{y'}\rangle$.

        For the second part of the claim observe that if $y_i \in Y_{y} \cap Y_{y'}$
        then there are a simple functional path from $y$ to $y_i$
        and a simple functional path from $y'$ to $y_i$.
        Hence, there is a simple path from $y$ to $y'$
        such that for every variable $u$ on this path we have that $y_i \in Y_u$.
        By Claim~\ref{claim:MM-functional-database}, this implies that $f_y(y_i) = f_{y'}(y_i)$.
    \end{proof}

    We now prove
    Lemma~\ref{lemma:lower-bound-partial-answers-and-matrix},
    starting with \textbf{MM1}.
    In fact, we show the following slightly more general
    statement that will be useful also in the proof of {\bf MM3}.
    \begin{claim}
        \label{claim:mapping-to-core}
        Let $h$ be a homomorphism from $q$ to $\Umc_{\Dmc,\Omc}$ and
        $z,z' \in \mn{var}(q)$ such that $y_0 \in Y_z$,
        $y_k \in Y_{z'}$, $h(z) = \langle z,f_z\rangle$, and
        $h(z') = \langle z',f_{z'}\rangle $. Then
        $(f_{z}(y_0),f_{z'}(y_k)) \in M_1M_2$.
    \end{claim}
    \noindent
    We may infer {\bf MM1} by applying the claim for $z = x_1$ and
    \mbox{$z' = x_2$}.
%     Note that, by Claim~\ref{claim:MM-small-Y},
%    $Y_{x_1}=\{y_0\}$ and $Y_{x_2}=\{y_k\}$ so the application of
%    Claim~\ref{claim:mapping-to-core} is indeed possible. % We know that
    % $y_0 \in Y_{x_1}$ and $y_k \in Y_{x_2}$, so we have to argue that
    % there are no other variables in these sets. Assume that
    % $y_i \in Y_{x_1}$ with $i \neq 0$. Claim~1 yields $i=1$. But then
    % $y_1$ is an answer variable in $q^+$, contradicting the fact
    % that $y_0,\dots,y_k$ is a bad path in $q^+$.

    \smallskip

    \noindent
    To prove Claim~\ref{claim:mapping-to-core}, let $h$ be a
    homomorphism from $q$ to $\Umc_{\Dmc,\Omc}$ and
    $z,z' \in \mn{var}(q)$ such that $y_0 \in Y_z$, $y_k \in Y_{z'}$,
    $h(z) = \langle z,f_z\rangle$, and
    $h(z') = \langle z',f_{z'}\rangle $.

    Recall that $y_0,\dots,y_k$
    is a bad path in $q^+$.
    Thus, % by Lemma~\ref{lemma:paths-from-q-plus-to-q},
    for every pair
    $y_i,y_{i+1}$, $0 \leq i <k$,  there is an atom $R^+(\bar{z})$ in $q^+$
    that uses both $y_i$ and $y_{i+1}$.
    Hence, there is a simple path $\pi_i$ from $y_i$ to $y_{i+1}$ in $q$
    such that for every variable $x$ on the path we have that $Y_x \neq \emptyset$.
    Indeed, since there is an atom in $q^+$ that uses both  $y_i$ and $y_{i+1}$,
    by definition of $q^+$ there is an atom $R(v,v')$ in $q$ such that
    there are a simple functional path $\nu$ from $v$ to $y_i$ in $q$
    and a simple functional path $\nu'$ from $v'$ to $y_{i+1}$ in $q$ .
    Let $v''$ be the last variable on the path $\nu$ that also belongs to $\nu'$.
    The desired simple path is the concatenation of the reverse of the suffix from $v''$ to $y_i$ of the path $\nu$
    and the suffix from $v''$ to $y_{i+1}$ of the path $\nu'$ .

%    by definition of $q^+$ there are a functional path $\pi_i$ in $q$ from some variable $u_i$ to $y_i$
%    and a functional path $\nu_i$ from some variable $u'_i$ to $y_{i+1}$ such that the path $\pi_i^-\nu_i$, i.e. concatenation of the reverse of $\pi_i$ and $\nu$,
%    is a simple functional path. Notice that for every
%    variable $u$ on this path $\emptyset \subsetneq Y_u$.
%    More precisely, for every
%    variable $u$ that $\pi_i$ crosses we have that $y_i \in Y_u$
%    and for every
%    variable $u$ that $\nu_i$ crosses we have that $y_{i+1} \in Y_u$. % \subseteq \{y_i,y_{i+1}\}$.

    Moreover, since $y_0 \in Y_z$ and $y_k \in Y_{z'}$, there are a
    simple functional path $\pi_{-1}$ from $z$ to $y_0$ in $q$ and a simple functional path $\pi_{k+1}^-$ from $z'$ to~$y_k$ in $q$.
    Concatenating paths $\pi_i$, $-1 \leq i \leq k+1 $, we obtain a path from $z$ to $z'$
    such that  $Y_u \neq \emptyset$ for every variable $u$ on this path.
    Since $z \neq z'$, we can refine it to a simple path
    $z_0, \dots, z_\ell$ from $z_0=z$ to $z_\ell = z'$ in $q$, by removing internal cycles.
    We call this path $\pi$.
    Notice that $y_0 \in Y_{z_0}$, $y_k \in Y_{z_\ell}$,
    and for every variable $u$ on this path the set $Y_u$ is not empty.

    Let $R_j(z_j,z_{j+1})$, $0 \leq j < \ell$, be a sequence of atoms on this path.
    Then, for every $0\leq i < k$ there is an atom $R_i(v_i,v'_i)$
    such that $\{y_i,y_{i+1}\} \subseteq Y_{v_i} \cup Y_{v'_i}$.

    To show this let $S_j = Y_{z_j} \cup Y_{z_{j+1}}$ for $0 \leq j < \ell$.
    Now, recall that by Claim~\ref{claim:MM-small-Y} for every $0 \leq j < \ell$ there is
    $0 \leq o < k$ such that $ S_j \subseteq \{y_o,y_{o+1}\}$ and,
    by the choice of $\pi$, we have that $S_{j} \cap S_{j+1} \neq \emptyset$.
    Moreover $y_0 \in S_0$ and $y_k \in S_{\ell -1}$.
    Now assume that there is $0 \leq i < k$ such that $\{y_i, y_{i+1}\} \not \subseteq S_j$ for all $0 \leq j < \ell$.
    Then, by a simple inductive argument on the length of the sequence of sets $S_j$, we can show that $y_o \notin S_j$ for all $o > i$.
    Thus, since $y_k \in S_\ell$, if $\{y_i, y_{i+1}\} \not \subseteq S_j$ for all $0 \leq j < \ell$
    then $i > k$. This implies that for every $0\leq i < k$ there is an atom $R_i(v_i,v'_i)$
    such that $\{y_i,y_{i+1}\} \subseteq Y_{z_j} \cup Y_{z_{j+1}}$.

%    Thus, as a consequence of , for every $0\leq i < k$ there is an atom $R_i(v_i,v_i')$
%    such that $v_i=z_j$ and $v_{i}'=z_{j+1}$ are consecutive variables on the path
%    and $Y_{v_i} \cup Y_{v_i'} = \{y_i,y_{i+1}\}$.
%    Indeed, by Claim~\ref{claim:MM-small-Y} we know that for every two consecutive variables $v,v'$
%    on the path we have that $Y_{v} \cup Y_{v'} \subseteq \{y_i,y_{i+1}\}$ as there is an atom $R(v,v') \in q$.
%    On the other hand, if for some pair $y_i, y_{i+1}$, $0 \leq i < k$ we have that
%    $Y_{v} \cup Y_{v'} \neq \{y_i,y_{i+1}\}$ for all consecutive variables $v,v'$ on the path
%    then it is not possible that $y_0 \in Y_{z_0}$ and $y_k \in Y_{z_\ell}$.

%    Let us, for a given $i$ fix a pair of the variables $z_j, z_{j+1}$ as above
%    and let us denote them as $v_i$ and $v'_i$, respectively.
%
    Since the path $\pi$ is simple and its ends are mapped to $\DGamma$, i.e. $h(z_0), h(z_\ell) \in \DGamma$,
    by Claim~\ref{claim:paths-in-core},
    we have that $h(u) \in \DGamma$ for every variable $u$ in $\pi$.
    Thus, by Claim~\ref{claim:MM-well-guided-hom}, we know that $h(v_i) = \langle v_i, f_{v_i} \rangle$
    and $h(v_i') = \langle v_i', f_{v_i'} \rangle$ for some partial functions $f_{v_i'},f_{v_i}$
    from $Y$ to $[n]$.

    Let $f = \bigcup_{0\leq i <k} f_{v_i} \cup f_{v_i'}$.
    We claim that $f$ is a function from $Y$ to $[n]$.
    The existence of atoms $R_i(v_i,v_i')$ shows that for every $y_i \in Y$
    there is $d \in [n]$ such that $(y_i,d) \in f$.
    On the other hand, Claim~\ref{claim:MM-well-guided-hom}
    shows that the value $d$ is uniquely defined. Thus, $f$ is a function from $Y$ to $[n]$.

        Let $f(y_0)=a$, $f(y_1)=c$, and $f(y_k)=b$. Then
    \begin{description}
        \item[ME1] $f(y_0) = f_z(y_0) = a$,
        \item[ME2] for $0<i<k$ we have that $f(y_i) = c$,
        \item[ME3] and $f(y_k) = f_{z'}(y_k) = b$.
        \item[ME4] Finally, $(a,c) \in M_1$, $(c,b)\in M_2$.
    \end{description}

    For \textbf{ME1} recall that $y_0 \in Y_{z}$. Hence, there is a simple functional path %$v_0=z, \dots, v_\ell=y_0$
    from $z$ to $y_0$. For all vertices $u$ on this path we have that $y_0 \in Y_u$ and, moreover,
    $h(u) \in \DGamma$. Thus $f(y_0) = f_z(y_0)$.

%    Thus, for all $0 \leq i < \ell$ and every atom $R(v_i,v_{i+1})$ in $q$
%    we have that $R(h(v_i), h(v_{i+1})) \in \Dmc_0$. Thus

    \textbf{ME3} is shown in the same way.
    For \textbf{ME2} observe that for $0<i<k-1$ we have that $Y_{v_i} \cup Y_{v_i'} = \{y_i,y_{i+1}\}$
    thus fact $R_i\big(\langle v_i, f_{v_i} \rangle,\langle v_i', f_{v_i'} \rangle \big)$
    was added by the third rule \textbf{MR3} in construction of $\Dmc_0$.
    Thus, there is a word $w=c_1c^{k-1}c_2$, for some $c$, $c_1$, $c_2 \in [n]$, such that $f_{v_i}^{w} = f_{v_i}$ and $f_{v_i'}^w = f_{v_i'}$.
    By the definition of $f_{x}^w$ we obtain $f(y_i) = f(y_{i+1})= c$.

    To show $(a,c) \in M_1$ observe that $Y_{v_0} \cup Y_{v_0'} = \{y_0,y_{1}\}$.
    Thus, the fact $R_0\big(\langle v_0, f_{v_0} \rangle,\langle v_0', f_{v_0'} \rangle \big)$
    was added by the first rule (\textbf{MR1}) in construction of $\Dmc_0$.
    Thus, $f_{v_0}^{w} = f_{v_0}$ and $f_{v_0'}^w = f_{v_0'}$ for $w=ac^{k}$.
    By definition of \textbf{MR1}, this implies that $(a,c) \in M_1$.
    A similar argument shows that $(c,b)\in M_2$.

    From \textbf{ME4} we infer immediately that $(a,b) \in M_1M_2$,
    which shows that $(f_z(y_0), f_{z'}(y_k)) \in M_1M_2$ and ends the proof of \textbf{MM1}.

    \paragraph{MM2}
%     To show \textbf{HD2}, let $v_0, \dots, v_k$ be a hyperclique in
% $G$. To show that there is a complete answer to $Q$ on \Dmc, it
% suffices to exhibit a homomorphism $h$ from $q$ to $\Umc_{\Dmc, \Omc}$
% such that $h(x) \in \mn{adom}(\Dmc_0)$ for every answer variable~$x$.
% The construction of $h$ is very similar to the proof of the \textbf{TD2} part
% of Claim~\ref{claim:lower-bound-partial-answers-and-cycle} and we
% provide only an outline.

% % We first fix a partial mapping $h$ from $q$ to $\Umc_{\Dmc, \Omc}$
% % such that $h(y_0) = \{(y_0, a)\}$, $h(y_i) = \{(y_i, b)\}$ for
% % $0 < i < k$, and $h(y_k) = \{(y_k, c)\}$.  Following the same steps as
% % in proof of \textbf{MM2}, we can show that this mapping can be
% % extended to a homomorphism from $q$ to $\Umc_{\Dmc, \Omc}$. This
% % proves \textbf{CD2}.

% Since $Q$ is non\=/empty, there is a database $\Dmc'$ that is
% satisfiable w.r.t.\ \Omc and a homomorphism $h'$ from $q$ to
% $\Umc_{\Dmc', \Omc}$ such that $h'(\bar{x})$ contains only elements of
% $\mn{adom}(\Dmc')$.  Let $W_0$ be the set of variables
% $x \in \mn{var}(q)$ with $h'(x) \in \mn{adom}(\Dmc')$.

% We may construct the desired homomorphism $h$ by setting
% $h(x) = f^{v_0\cdots v_k}_x$ for all $x \in W_0$, and then extending
% to the remaining variables in $q$ (which are part of tree-shaped
% subqueries of $q$) exactly as in the proof of the \textbf{TD2} part
% of Claim~\ref{claim:lower-bound-partial-answers-and-cycle}.

    Let $(a,b) \in M_1M_2$ and choose a $d$ such that $(a,d) \in M_1$
    and $(d,b) \in M_2$.  To prove MM2, it suffices to identify a
    homomorphism $h$ from $q$ to ${\Umc_{\Dmc, \Omc}}$ such that
    $h(\bar{x}) = \bar{a}$.

    Since $Q$ is non\=/empty, there is a database $\Dmc'$ that is
    satisfiable w.r.t.\ \Omc and a homomorphism $h'$ from $q$ to
    $\Umc_{\Dmc', \Omc}$ such that $h'(\bar{x})$ contains only
    elements of $\mn{adom}(\Dmc')$.
    Let $W_0$ be the set of variables $x \in \mn{var}(q)$ with
    $h'(x) \in \mn{adom}(\Dmc')$.

    We may construct the desired homomorphism $h$ by setting
    $h(x) = \langle x, f^{ad^{k-1}b}_x\rangle$ for all $x \in W_0$,
    and then extending to the remaining variables in $q$ (which are
    part of tree-shaped subqueries of $q$) exactly as in the proof of
    the \textbf{TD2} part of
    Claim~\ref{claim:lower-bound-partial-answers-and-cycle}.

    \paragraph{MM3}
%     We want to prove that the database $\Dmc$ provides not too many answers to $Q$ on $\Dmc$.
%     To do so, we will estimate all the possible wildcard tuples that can be partial answers to $Q$ on $\Dmc$.
%
% %    Let $m$ be the number of answer variables.
%     For a $\bar{a}$ partial answer, let us fix a homomorphism $h$ from $q$ to $\Umc_{\Dmc, \Omc}$ such that for all $1 \leq i \leq m$
%     if $a_i$ is not a wildcard then $h(x_i) = a_i$.
%
%     To analis
%
%     \paragraph{MM3a}
    Let $\bar a \in Q^*(\Dmc)$ be a minimal partial answer and
    $x \in \mn{var}(q)$.
    We say that $\bar a$ is \emph{misguided} and
    that $x$ is \emph{misguided to $c \in \mn{adom}(\Dmc_0)$ for}
    $\bar a$ if there is a homomorphism $h$ from $q$ to
    $\Umc_{\Dmc, \Omc}$ such that $h(\bar{x}) = \bar{a}$ and
    $h(x) = c$ has the form $\langle x', f_x\rangle$ with
    $x'\neq x$. We first show the following.
    \begin{claim}
    There are no more than
    $$|\bar x| \cdot |\DGamma| \cdot (|\mn{adom}(\Dtree)| +1)^{|\bar
      x|}$$ misguided minimal partial answers.
        \end{claim}
        Note that the above count is
    $O(|M_1|+|M_2|)$ since $|\DGamma|$ is $O(|M_1|+|M_2|)$ and
    $|\bar x|$ and $|\mn{adom}(\Dtree)|$ are constants.
    \begin{proof}
    It clearly suffices to show that for every $x \in \mn{var}(q)$ and
    $c \in \mn{adom}(\Dmc_0)$, there are at most
    $(|\mn{adom}(\Dtree)| +1)^{|\bar x|}$ minimal partial answers
    $\bar a \in Q^*(\Dmc)$ such that $x$ is misguided to $c$
    for~$\bar a$. Thus fix $x$ and $c=\langle x', f_x\rangle$ with
    $x'\neq x$. It suffices to show that if $h$ is a homomorphism from
    $q$ to $\Umc_{\Dmc, \Omc}$ such that $h(\bar{x}) = \bar{a}$ and
    $h(x) = c$, then for every answer variable~$x_i$, the
    corresponding constant $a_i$ in $\bar a$ is either a constant from
    a tree attached in Step~2 of the construction of \Dmc to $c$
    (which includes the possibility that $h(x_i)=c$) or it is
    `$\ast$'.

    Let $x_i$ be an answer variable.  Since $q$ is connected, $q$
    contains a simple path $z_1,\dots,z_\ell$ from $x=z_1$ to
    $x_i=z_\ell$. There are thus atoms
    $R_1(z_1,z_2),\dots,R_{\ell-1}(z_{\ell-1},z_\ell) \in q$.
    Consider the fact $R_1(h(x),h(z_2)) \in \Dmc$. We argue that
    $R_1(h(x),h(z_2)) \notin \Dmc_0$.  Assume for a contradiction that
    the opposite is true. %$R_1(h(x),h(z_2)) \in \Dmc_0$.
    From $h(x) = \langle x', f_x\rangle$ and the construction of
    $\Dmc_0$, it follows that $q$ contains a fact of the form
    $R_1(x',y)$. This contradicts $q$ being self-join free and
    containing the atom $R_1(x,z_2)$ with $x \neq x'$.  From
    $R_1(h(x),h(z_2)) \notin \Dmc_0$, it follows that
    $h(z_2) \notin \DGamma$, that is, $h(z_2)$ is a constant from a
    tree attached in Step~2 of the construction of \Dmc to $c$ or an
    element of
    $\mn{adom}(\Umc_{\Dmc,\Omc}) \setminus \mn{adom}(\Dmc)$.

    To show that $a_i$ is a constant from a tree attached in Step~2 of
    the construction of \Dmc to $c$ or `$\ast$', it suffices to show
    that $h(z_j) \notin \mn{adom}(\Dmc_0)$ for $2 \leq j \leq \ell$,
    or in other words: after leaving $\DGamma$ in the first step,
    entering a tree-shaped part of $\Umc_{\Dmc,\Omc}$, the path
    $R_1(z_1,z_2),\dots,R_{\ell-1}(z_{\ell-1},z_\ell)$ never leaves
    that tree-shaped part again. Assume to the contrary that
    $h(z_j) \in \mn{adom}(\Dmc_0)$ with $2 \leq j \leq \ell$. Then
    Claim~\ref{claim:paths-in-core} applied to the simple path
    $z_1,\dots,z_j$ yields $h(z_1) \in \DGamma$, a contradiction.
    \end{proof}

%    \paragraph{MM3b} % Therefore, let us assume that for every variable $x$ in $q$ if $h(x) \in \DGamma$ then $h(x) = (x,f_x)$
    % for some function $\fun{f_x}{Y_x}{[n]}$. All the remaining cases are covered by MM3a.
    It remains to analyse the number of minimal partial answers
    $\bar a$ that are not misguided.  % Let $h$ be a homomorphism from
    % $q$ to $\Umc_{\Dmc, \Omc}$ such that $h(\bar{x}) = \bar{a}$. Then
    % for all answer variables $x$, if $h(x) \in \DGamma$ then
    % $h(x) = \langle x, f_x\rangle$. %Let $L = 2 |\Sigma| \cdot |\Dtree|$, which is a constant.

    \smallskip
    We first establish a technical claim that analyses the facts
    $R(a,b) \in \Dmc$, which `cross' from the $\Dmc_0$-part of \Dmc
    into the tree part added in the second step of the construction.
    \begin{claim}
        \label{claim:few-trees}
        Let $R(x,y) \in q$, $R(a,b) \in \Dmc$,
%        If $a=\langle x,f\rangle \in \DGamma$ and $b \notin \DGamma$ then $Y_y \not \subseteq Y_x$.
%
        % {\color{blue}Let $r \in \Sigma$ be such that  $R(x,y)$ is an atom in $q$ and $R(a,b)$ is a fact in $\Dmc$, for $R \in \{r,r^-\}$.
            % If $a=\langle x,f\rangle \in \DGamma$ and $b \notin \DGamma$ then $Y_y \not \subseteq Y_x$.}
         $a=\langle x,f\rangle \in \DGamma$, and $b \notin \DGamma$. Then
        \begin{enumerate}
            \item $\{y_0, y_k \} \cap Y_x = \emptyset$ and
            \item $\{y_0, y_k \} \cap Y_y \neq  \emptyset$.
%            \item and $\{y_0, y_k \} \cap (Y_x\cup Y_y) \neq  \emptyset$.
        \end{enumerate}
%        Moreover, $\{y_0, y_k \} \cap Y_x = \emptyset$, but  $\{y_0, y_k \} \cap (Y_x\cup Y_y) \neq  \emptyset$.
    \end{claim}
    \begin{proof}
      Let $R(x,y)$ and $R(a,b)$ be as in the claim.
 %        \cMP{fix the proof, done}
%
 % {\color{blue}Is this Claim~\ref{claim:MM-functional-database}? then it should be earlier!} \cMP{Not Claim~\ref{claim:MM-functional-database} -- it's analysis of construction of $\Dmc$; we don't state that if $R(a,b) \in \Dmc$ then $a,b$ are in relation, but we say that a fact was added in the construction.
 %            still changing proof and moving the claim here, the comment is void.}
      Since $b \notin \DGamma$, the fact $R(a,b)$ was added to the
      database $\Dmc$ in Step~2 of its construction. Since
      $a \in \DGamma$, it follows that there is no $d \in \DGamma$
      such that $R(a,d)\in \Dmc_0$.

        Since $a \in \DGamma$, there is a fact $R'(a,c) \in \Dmc_0$.
        Thus, by construction of $\Dmc_0$ there is a a word $w = w_0\cdots w_k \in [n]^\ast$ such that $f = f_x^w$
        and $(w_0,w_1) \in M_1$ or $(w_{k-1},w_k) \in M_2$.

        For Point~1, assume to the contrary of what is to be shown
        that $y_0 \in Y_x$. From $R(x,y) \in q$ and the construction
        of $\Dmc_0$ it then follows that
        $R(\langle x, f_x^w \rangle, \langle y, f_y^w \rangle) \in
        \Dmc_0$.  This is a contradiction to that fact that there is
        no $d \in \DGamma$ such that $R(a,d) \in \Dmc_0$. We can show
        similarly that $y_k \notin Y_x$.

        For Point~2, assume to the contrary that $\{y_0, y_k \} \cap (Y_x\cup Y_y) = \emptyset$.
        Then, by Point~3 of the construction of $\Dmc_0$ we have
        $R(\langle x,  f_x^w \rangle, \langle y, f_y^w \rangle) \in
        \Dmc_0$ which again yields a  contradiction.
    \end{proof}

    \newcommand{\treeroot}{\textmd{root}}

    Recall that the restriction of $\Umc_{\Dmc, \Omc}$ to domain
    $\dom(\Umc_{\Dmc, \Omc}) \setminus \DGamma$ is a disjoint union of
    trees. These trees consists of the constants that have been added
    in Step~2 of the construction of \Dmc and of the nulls that have
    been added in the construction of the universal model.
    Let $\mathfrak{T}'$ be the set of all those trees and let
    $\mathfrak{T}$ be the subset of trees in $\mathfrak{T}'$ that
    contain at least one constant in $\Dmc$, i.e.
    $\mathfrak{T} = \{\Tmc \in \mathfrak{T}' \mid \dom(\Tmc) \cap
    \dom(\Dmc) \neq \emptyset\}$.

    For a constant $c \in \dom (\bigcup \Tmf)$, let $\treeroot(c)$ be
    the root of the tree $\Tmc \in \Tmf$ that contains $c$,
    i.e.~$\treeroot(c)$ is the unique constant in $\Tmc$ such that
    there is a fact $R(\treeroot(c),d) \in \Dmc$ with $d \in \DGamma$.
    Moreover, let $T^{\Umc}_c$ be the domain of the tree that contains
    $c$.
   %
%    Let $T_c$ be the subset of all constants in $\dom(\Umc_{\Dmc, \Omc})$ in the tree rooted in $\treeroot(c)$.
%    That is, if $c \in \DGamma$ then
%
%    Let $T^{\Umc}_c$ be the subset of  $\dom(\Umc_{\Dmc, \Omc})$ that is the maximal connected component of $\Umc_{\Dmc, \Omc} \setminus \{R(c,c')\}$
%    that contains $c$. That is, $T^{\Umc}_c$ is the tree shaped connected part of the universal model that contains $c$ and does not contain the core of the database.
    Clearly, for all $c \in \Umc_{\Dmc, \Omc}$ we have $|T^{\Umc}_c \cap \dom(\Dmc)|\leq |\mn{adom}(\Dmc_{\mn{tree}})|$.

    Further let $\DGamma^\Umc=\bigcup_{c \in \DGamma} \dom(\Umc_{\Dmc, \Omc}^{\downarrow c})$.

    It is not hard to verify that the set $\DGamma^{\Umc}$ and the
    sets $T^{\Umc}_c$, for $c \in \dom(\Dmc) \setminus \DGamma$, form
    a partition of $\mn{adom}(\Umc_{\Dmc, \Omc})$.

    % Consider an answer variable $x_i$ in $q$. By
    % Claim~\ref{claim:MM-small-Y}, $Y_{x_i}$ is $\{y_0\}$,
    % $\{y_k\}$, or $\emptyset$. We also know that $Y_{x_1} = \{y_0\}$
    % and $Y_{x_2} = \{y_k\}$.
    % We will now inspect four cases that depend on whether some of the
    % answer variables $x$ with $Y_x \subsetneq \{y_0,y_k\}$ are mapped to $\DGamma^{\Umc}$.

%    As an immediate corollary of Clai\ref{claim:few-trees}, we get the following.
%    \begin{claim}
%        Fix $y \in \{y_0,y_k\}$. If there is a variable $x$ such that $Y_X = \{y\}$ and $h(x) \in \Gamma$ then
%        for every $x'$ such that $Y_{x'} = \{y\}$ we have that $h(x') \in \Gamma$ or $h(x') \notin \dom(\Dmc)$.
%    \end{claim}

    The next claim analyses the relationship between the homomorphism
    targets of answer variables in $q$. Recall that by
    Claim~\ref{claim:MM-small-Y}, $Y_{x_i}$ is $\{y_0\}$, $\{y_k\}$,
    or $\emptyset$ for each answer variable~$x_i$. Also recall that
    $Y_{x_1} = \{y_0\}$ and $Y_{x_2} = \{y_k\}$.
    Intuitively, part \textbf{MH1} implies that for homomorphisms from
    $q$ to $\Umc_{\Dmc, \Omc}$ that are not misguided either (\textbf{MH1a}) every variable $x$ such that $y_0 \in Y_{x}$
    is mapped to a trace below a constant in $\DGamma$ or (\textbf{MH1b}) there is a single constant $d \in \dom(\Dmc) \setminus \DGamma$
    such that every every variable $x$ such that $y_0 \in Y_{x}$ is mapped to the tree $T_{d}^{\Umc}$.
    A similar observation can be made for variables $x$ such that $y_k \in Y_x$.
    Finally, for variables $x$ such that $Y_x = \emptyset$, the claim implies, part \textbf{MH2}, that
    either x is mapped to a trace below a constant from $\Umc_{\Dmc, \Omc}$ or is mapped to one of the trees above.

%
%     From this observation, we may infer the relation between images of answer variables
%     for a fixed homomorphism.
%     Informally, the following claim states that the answer variables $x$ with $Y_{x}=\{y_0\}$
% %    are either all mapped to a single copy of $\Dtree$ or they are all mapped below $\Dmc_0$.
%     The same holds for variables $x$ with $Y_{x}=\{y_k\}$.
%     For the answer variables $x$ with $Y_{x}=\emptyset$, they are either mapped with
%     the $\{y_0\}$ variables, with the $\{y_k\}$ variables or mapped to $\langle x, \emptyset \rangle$.
    %
    \begin{claim}
    \label{claim:MM-matching-avars-to-trees}
    Let $h$ be a homomorphism from $q$ to $\Umc_{\Dmc, \Omc}$ that
    is  not misguided. Then the following hold:

    \begin{description}
    \item[{\bf MH1}]
    Let $x$ be an answer variable such that $Y_x = \{y_0\}$ or $Y_x =
    \{ y_k \}$. %$Y_x \neq \emptyset$.
    For every variable $x'$ such that $Y_x= Y_{x'}$:
     ~\\[-4mm]
     \begin{description}
         \item[\bf{MH1a}] if $h(x) \in \DGamma^{\Umc}$ then $h(x') \in \DGamma^{\Umc}$,
         \item[\bf{MH1b}] if $h(x) \notin \DGamma^{\Umc}$ then $h(x') \in T_{\treeroot(h(x))}^{\Umc}$.
    \end{description}

    \item[{\bf MH2}]
    Let $x$ be an answer variable such that $Y_x = \emptyset$.
    If $h(x) \notin \DGamma^{\Umc}$
    then there is an answer variable $x'$ with $Y_{x'} \neq \emptyset$ such that $h(x) \in T_{\treeroot(h(x'))}^{\Umc}$.
    \end{description}
    \end{claim}

    \begin{proof}
        For {\bf{MH1}},
        let $x,x'$ be as in the claim.
        Since $Y_x = Y_x' = \{y\}$ for some $y \in \{y_0,y_k\}$, there are a simple functional path from $x$ to $y$ and a simple functional path from $x'$ to $y$.
        Moreover, since the paths are functional and $y \in Y_y$, for any variable $u$ on either of those paths we have that  $y \in Y_u$.
        Thus, there is a simple path $z_1, \dots, z_\ell$ from $z_1=x$ to $z_\ell = x'$ such that for every variable $z_i$ we have that $y \in Y_{z_i}$.

        For {\bf{MH1a}},
        let us assume, by contradiction,  that $h(x) \in \DGamma^{\Umc}$ and $h(x') \notin \DGamma^{\Umc}$.
        Thus, $\treeroot(h(x'))$ is well defined. Let $d = \treeroot(h(x'))$.
        By construction of $\Dmc$, there is $z_j$ such that $h(z_j) = d$.
        Thus, there are $z_{j-1}$ and an atom $R(z_{j-1}, z_{j})$ such that
        $h(z_{j-1}) = c\in \DGamma$, $h(z_{j})=d \notin \DGamma$, and $R(c,d) \in \Dmc$.
        By Claim~\ref{claim:few-trees}, we have that $y \notin Y_{z_{j-1}}$,
        but this is impossible as we know that for every $z_i$ on the path
        we have that $y \in Y_{z_i}$. This proves {\bf{MH1a}}.

        For {\bf{MH1b}}, let $h(x) \notin \DGamma^{\Umc}$.
        Thus, $\treeroot(h(x))$ is well defined. Let $d = \treeroot(h(x))$.
        By contradiction, let us assume that  $h(x') \notin T_{d}^{\Umc}$.
        Then, $h(z_1), \dots, h(z_\ell)$ is a path in $\Dmc$.
        Moreover, since $h(z_\ell) \notin T_{d}^{\Umc}$, by construction of $\Dmc$ there is a constant $h(z_{j+1}) = c \in \DGamma$
        and an atom $R(z_j, z_{j+1}) \in q$ such that $h(z_j)=d$.
        Since $c,d \in \dom{\Dmc}$, we have $R(c,d) \in \Dmc$ and, by Claim~\ref{claim:few-trees},
        we infer that $y \notin Y_{z_{j+1}}$. Again, this is impossible as we know that for every $z_i$ on the path
        we have that $y \in Y_{z_i}$. This proves {\bf{MH1b}}.

        For {\bf MH2},
        let $x$ be an answer variable with $Y_x = \emptyset$ and  $h(x) \notin \DGamma$.
        Then, $\treeroot(h(x))$ is well defined. Let $d = \treeroot(h(x))$.
        If $h(x_1) \in T^{\Umc}_d$, then we are done, so let us assume that $h(x_1) \notin T^{\Umc}_d$.
        Let $z_1, \dots, z_\ell$ be a simple path from $z_1=x$ to $z_\ell = x_1$.
        Since $h(x_1) \notin T^{\Umc}_d$ and $h(x) \in T^{\Umc}_d$, by construction of $\Dmc$ there
        is a variable $z_i$ on this path such that $h(z_i) = d$
        and $h(z_{i+1}) \in \DGamma$. Thus, there is an atom $R(z_i, z_{i+1}) \in q$
        and a fact $R(h(z_i),h(z_{i+1}))$. Thus, by Claim~\ref{claim:few-trees},
        we have that $\{y_0,y_k\} \cap Y_{h(z_i)} \neq \emptyset$.
        By {\bf MH1a} $y_0 \notin Y_{h(z_i)}$.
        Thus $y_k \in Y_{h(z_i)}$ and, by {\bf MH1b}, $h(x_2) \in T^{\Umc}_d$.
        This ends the proof of {\bf MH2}.
    \end{proof}

%    Intuitively, the above claim states that all answer variables $x$
%    with $Y_x = \{y_0\}$ are either mapped to the same copy of (or the
%    chase below) a tree added in Step 2. of the construction, or all
%    are mapped to $\Dmc_0$ (or to the chase below).  The same holds
%    for answer variables $x$ with $Y_x = \{y_k\}$.  Finally, answer
%    variables $x$ with $Y_x = \emptyset$ are mapped ``near'' the
%    answer variables $x'$ with $Y_x \neq \emptyset$.

    \medskip For every answer $\bar a \in Q(\Dmc)^\Wmc$, choose a
    homomorphism $h$ from $q$ to $\Umc_{\Dmc,\Omc}$ with
    $h(\bar x)=\bar a$ that is not misguided.  We analyse four cases
    according to where exactly $h$ maps the answer variables in~$q$,
    and determine the maximum number of minimal partial answers in
    each case. Note that by Claim~\ref{claim:MM-small-Y}, the cases
    are exhaustive.

    \medskip \textit{Case 1}:  There is no variable $x$ such that
    $Y_x = \{y_0\}$ or $Y_x = \{y_k\}$ and
    $h(x) \in \DGamma^{\Umc}$.
    We will show that there are $O(|M_1| + |M_2|)$ minimal partial answers
    consistent with this condition.

    \smallskip
    By Claim~\ref{claim:MM-matching-avars-to-trees} part \textbf{MH1b},
    %this \textit{Case 1.} assumption implies that
    there is a constant $d_1$ such that all answer variables $x$ with
    $Y_x = \{y_0\}$ are mapped to the tree $T_{d_1}^{\Umc}$, and
    likewise for variables with $Y_x = \{y_k\}$ and some constant
    $d_2$.  To be more precise, let $d_1 = \treeroot(h(x_1))$ and
    $d_2 = \treeroot(h(x_2))$. Since $Y_{x_1}=\{y_0\}$ and
    $Y_{x_2}=\{y_k\}$, Point~{\bf MH1b} of
    Claim~\ref{claim:MM-matching-avars-to-trees} implies that
    $h(x) \in T^{\Umc}_{d_1}$ for every answer variable $x$ with
    $Y_x = \{y_0\}$ and $h(x) \in T^{\Umc}_{d_2}$ for every answer
    variable $x$ with $Y_x = \{y_k\}$. Moreover, by {\bf MH2}, by
    Claim~\ref{claim:MM-small-Y}, and by fact that $h$ is not
    misguided, for every answer variable $x$ with $Y_x = \emptyset$ we
    have
    $h(x)\in T^{\Umc}_{d_1} \cup T^{\Umc}_{d_2} \cup \{\langle x,
    \emptyset \rangle\}$.  This implies that, for every pair of
    constants $d_1,d_2 \in \dom(\Dmc)$, for any answer variable $x$
    with $h(x) \in \dom(\Dmc)$ we have
    $h(x) \in T^{\Umc}_{d_1} \cup T^{\Umc}_{d_2} \cup \{\langle x,
    \emptyset \rangle\}$. This gives us
    $2 |\mn{adom}(\Dmc_{\mn{tree}})|+1$ possible values of $h(x)$.
    Since for every answer variable $x_i$ such that
    $h(x_i) \notin \dom(\Dmc)$ the value $a_i$ is one of
    $|\bar{x}|$ wildcards, for a fixed pair $(d_1,d_2)$ we have no
    more than
    $(2 |\mn{adom}(\Dmc_{\mn{tree}})|+1 + |\bar{x}|)^{|\bar{x}|}$
    minimal partial answers, which is independent
    of $||\Dmc||$.  Thus, to show that the number of possible minimal
    partial answers is $O(|M_1| + |M_2|)$, % in the first case
    it is enough to show that the number of possible pairs $(d_1,d_2)$
    is $O(|M_1| + |M_2|)$.

    By definition of the $\treeroot$ function, $\Dmc$ contains unique
    facts $R_1(c_1,d_1)$ and $R_2(c_2,d_2)$ such that
    $c_1,c_2 \in \DGamma$. Observe that there are no more than
    $|\Sigma|\cdot|\DGamma| \in O(|\DGamma|)$ pairs $(d_1,d_2)$ such
    that $c_1=c_2$. Thus, let us assume that $c_1 \neq c_2$. Since
    $c_1,c_2 \in \DGamma$ and the homomorphism $h$ is not misguided,
    there are variables $z_1, z_2$ in $q$ such that
    $c_1 = \langle z_1, f_{z_1} \rangle$ and
    $c_2 = \langle z_2, f_{z_2} \rangle$.  As $q$ is connected, there
    are variables $v_1,v_2$ in $q$ such that $h(v_1) = d_1$,
    $h(v_2) = d_2$ and $R_1(z_1,v_1), R_2(z_2,v_2)$ are atoms in $q$.
    By Claim~\ref{claim:few-trees} we have that
    $Y_{v_1} \cap \{y_0, y_k\} \neq \emptyset$. Since all variables
    $x$ with $y_k \in Y_x$ are mapped to the tree $T^{\Umc}_{d_2}$,
    which by assumption is different from $T^{\Umc}_{d_1}$, we infer
    that $y_0 \in Y_{v_1}$.  Thus, by Claim~\ref{claim:few-trees} and
    the fact that there is an $i$ with
    $Y_x \cup Y_{x'} \subseteq \{y_i, y_{i+1}\}$ for every atom
    $R(x,x') \in q$, cf.\ Point~1 of Claim~\ref{claim:MM-small-Y}, we
    infer that $Y_{z_1} \subseteq \{y_1\}$. Indeed, since
    $y_0 \in Y_{v_1}$ and $R_1(z_1,v_1) \in q$,
    $Y_{v_1} \cup Y_{z_1} \subseteq \{y_0, y_{1}\}$.  Thus, Point
    1 of Claim~\ref{claim:few-trees} implies that
    $Y_{z_1} \subseteq \{y_1\}$.  By a similar argument,
    $Y_{z_2} \subseteq \{y_{k-1}\}$.

    Now, following a similar argument as in the \textbf{ME2} statement
    of \textbf{MM1}, we can show that
    $f_{z_1}(y_1) = f_{z_2}(y_{k-1})$.  Thus, if $c_1 \neq c_2$ then
 %   $c_1 = \langle z_1, f_{z_1} \rangle$ and
   % $c_2 = \langle z_2, f_{z_2} \rangle$ such that
    $f_{z_1} \subseteq \{(y_0, a)\}$ and
    $f_{z_2} \subseteq \{(y_k, a)\}$ for some $a$ used in $M_1$ or in
    $M_2$, i.e. there is a pair $(a,b) \in M_1 \cup M_2$ or a pair
    $(b,a) \in M_1 \cup M_2$.

    Therefore, there are no more than $8|\mn{var}(q)|^2(|M_1|+|M_2|)$
    possible pairs $c_1\neq c_2$. Indeed, there is no more than
    $2(|M_1|+|M_2|)$ possible values of $a$ and, thus, no more than
    $4\cdot 2(|M_1|+|M_2|)$ possible pairs $(f_{z_1}, f_{z_2})$. Since
    there is no more than $|\mn{var}(q)|^2$ possible pairs $(z_1,z_2)$
    we can bound the number of pairs $(c_1,c_2)$ as above. Now,
    observe that for $i=1,2$ and every possible constant $c_i$, there
    is no more than $|\Sigma|$ different possible candidates for
    $d_i$. Thus, if $c_1 \neq c_2$ then there is no more than
    $$8|\mn{var}(q)|^2(|M_1|+|M_2|)|\Sigma| = O(|M_1|+|M_2|)$$ possible
    pairs $(d_1,d_2)$.
    Combining the $c_1 \neq c_2$ and the $c_1 = c_2$ case, in total there
    are no more than
    $$|\Sigma|\cdot|\DGamma| +
    8|\Sigma|\cdot|\mn{var}(q)|^2(|M_1|+|M_2|)$$ possible pairs
    $(d_1,d_2)$. Since, by Claim~\ref{claim:MM-database-is-ok}
    Point~1, $|\DGamma| = O(|M_1|+ |M_2|)$, we infer that there are no
    more than $O(|M_1|+ |M_2|)$ possible minimal partial answers in
    \textit{Case 1}.

    \medskip \textit{Case 2:} There is a variable $x$ such that
    $Y_x = \{y_0\}$ and $h(x) \in \DGamma^{\Umc}$, but no
    variable $y$ such that $Y_x = \{y_k\}$ and
    $h(y) \in \DGamma^{\Umc}$.

    \smallskip

    Let $d= \treeroot(h(x_2))$.
    We show that there is a pair $(a,b) \in M_1$ such that the following holds:
%    There is no $j$ such that $(j,a) \in M_1$ and there is a fact $R(d, \langle u, f_u\rangle)$
%    in $\Dmc$ such that $y_0 \notin Y_u$ and $f_u(y_1) = a$.
 \begin{enumerate}\itemsep=0pt
        \item for all answer variables $x'$ such that $Y_{x'} = \{y_k\}$ we have $h(x') \in T^\Umc_{d}$;
        \item for all answer variables $x'$ such that $Y_{x'} = \{y_0\}$ we have either $h(x') \notin \mn{adom}(\Dmc)$
        or $h(x')=\langle x',f_{x'}\rangle$ with $f_{x'}(y_0) = a$;
        note that in the latter case, $f_{x'}$ is only defined on the
        variable $y_0$ and thus $h(x')$ is uniquely determined;
        \item for all answer variables $x'$ such that $Y_{x'} = \emptyset$ we have either $h(x') \in \DGamma^{\Umc}$ or $h({x'}) \in T^\Umc_{d}$.
        \end{enumerate}
    It follows that in this case there are no more than $$|M_2|\cdot (|\mn{adom}(\Dmc_{\mn{tree}})|+1)^{|\bar{x}|}$$ possible partial answers.

    \smallskip Point~1 follows from {\bf MH1b} in
    Claim~\ref{claim:MM-matching-avars-to-trees}, which implies that for
    all answer variables $x'$ with $y_k \in Y_{x'}$ we have
    $h(x') \in T^\Umc_{d}$.

%    By the same argument, if there was a variable $u$ such that $y_0 \in  Y_{u}$ and $h(u) \notin \DGamma^{\Umc}$,
%    then we would have that for all variables $u'$ with $y_0 \in Y_{u'}$  holds $h(u') \in T^{\Umc}_{h(u)}$.
%    In particular, this would imply that $h(x) \in T^{\Umc}_{h(u)}$.
%    Since $h(x) \in \DGamma$ and for all $t \in \Tmc$ we have that $\dom(t) \cap \DGamma^{\Umc} = \emptyset$, this is impossible.
%    Hence, for all variables $u$ such that $y_0 \in  Y_{u}$ we have that $h(u) \in \DGamma^{\Umc}$.

    Now for Point~2.  By {\bf MH1a} in
    Claim~\ref{claim:MM-matching-avars-to-trees}, we have
    $h(x') \in \DGamma^{\Umc}$ for all answer variables $x'$ with
    $y_0 \in Y_{x'}$. If $h(x') \notin \mn{adom}(\Dmc)$ for all such
    $x'$, then we are done. Thus assume otherwise, that is, for some
    answer variable $x'$ with $y_0 \in Y_{x'}$ we have
    $h(x') \in \mn{adom}(\Dmc_0)$.

    Since $y_0 \in Y_{x'}$, there is a simple functional path
    $z_1,\dots,z_\ell$ from $x'=z_1$ to $y_0=z_\ell$.  Thus, there are
    atoms $R_1(z_1,z_2),\dots,R_{\ell-1}(z_{\ell-1},z_\ell) \in q$.
    Moreover, there is a variable $z = z_j$ such that
    $h(y_0) \in \dom(\Umc_{\Dmc, \Omc}^{\downarrow h(z)})$.  Thus,
    $h(z) = \langle z, f_z\rangle \in \DGamma$ and, by construction of
    $\Dmc_0$, $f_z(y_0) = a$ for some $(a,b) \in M_1$.

    Since $h(y_0) \in \dom(\Umc_{\Dmc, \Omc}^{\downarrow h(z)})$, for
    every answer variable $x'$ with $y_0 \in Y_{x'}$ and for every
    simple functional path $v_1,\dots,v_m$ from $x'=z_1$ to $y_0=z_m$,
    which has to exists as $y_0 \in Y_{x'}$, there is a variable
    $v = v_j$ on this path such that
    $h(v) = \langle v, f_v \rangle = \langle z, f_z \rangle = h(z)$.
    Since $y_0 \in Y_{v_i} $ or all variables $v_i$ on this path, by
    Claim~\ref{claim:MM-functional-database} we infer that
    $f_{x'}(y_0) = f_z(y_0) = a$.  This proves Point~2.

    Point~3 follows immediately from  {\bf MH2} in
    Claim~\ref{claim:MM-matching-avars-to-trees}.

%    For the third bullet,
%    let $u$ be an answer variable with $Y_u = \emptyset$. If $h(u) \in \DGamma$ then $h(u) = \langle u, \emptyset \rangle$.
%    If $h(u) \in T^{\Umc}_{d'}$ for some $d'$ then there is $c' \in \DGamma$ such that $R(c',d') \in \Dmc$.
%    Moreover, $c' = \langle z, f_z\rangle$ for some variable $z$.
%
%    If $d' = d$ then we are done, so let us assume that $d \neq d'$.
%    Then, there is a variable $v$ and an atom $R(z,v) \in q$ such that $h(v) = d'$.
%    This follows immediately from the construction of $\Dmc$ and the fact that $q$ is connected.
%
%    By Claim~\ref{claim:few-trees}, we have that $\{y_0,y_k\} \cap Y_z = \emptyset$.
%    Moreover, since $d' \neq d$, we have that $Y_v \cap \{y_0,y_k\}  = \emptyset$ as all
%    variables $v'$ such that $Y_{v'} \cap \{y_0,y_k\}  \neq \emptyset$ are mapped to $T^{\Umc}_{d_1} \cup \DGamma{\Umc}$.
%    Thus, $(Y_z \cup Y_v) \cap \{y_0,y_k\} = \emptyset$, which is impossible by Claim~\ref{claim:few-trees}.

    \medskip

    \textit{Case 3}:
    There is a variable $x$ such that $Y_x = \{y_k\}$ and $h(x) \in \DGamma^{\Umc}$,
    but no variable $y$ such that $Y_x = \{y_0\}$ and
    $h(y) \in \DGamma^{\Umc}$.

    \smallskip

    By an argument smmetric to the one for Case~2, we may show that there are
    no more than $|M_2|\cdot (|\mn{adom}(\Dmc_{\mn{tree}})|+1)^{|\bar{x}|}$ possible partial
    answers in this case.

\medskip

    \textit{Case 4};
   There is a variable $x$ such that $Y_x = \{y_0\}$ and $h(x) \in \DGamma^{\Umc}$
    and a variable $y$ such that $Y_y = \{y_k\}$ and $h(y) \in
    \DGamma^{\Umc}$.

\smallskip

%    Then, by connectedness of $q$ and Claim~\ref{claim:few-trees}, for every answer variable $z$ we have that $h(z) \in \DGamma^{\Umc}$.
     Then, by {\bf MH1a} of Claim~\ref{claim:MM-matching-avars-to-trees}, for every answer variable $z$ we have $h(z) \in \DGamma^{\Umc}$.

     If $h(x) \in \DGamma$ for some answer variable $x$ with
     $Y_{x} = \{y_0\}$ then we can argue as in Case~2 that
     for every
     answer variable $x'$ with $Y_{x'} = \{y_0\}$ and
     $h(x') \in \DGamma$ we have $h(x') = \langle x', f_x\rangle$.

     Similarly, if $h(y) \in \DGamma$ for some answer variable $y$
     with $Y_{y} = \{y_k\}$, then we can argue as in Case~3 that for
     every answer variable $y'$ with $Y_{y'} = \{y_k\}$ and
     $h(y') \in \DGamma$ we have $h(y') = \langle y', f_x\rangle$.

     Moreover, if both such $x$ and $y$ exist then
     Claim~\ref{claim:mapping-to-core} gives    $(f_x(y_0), f_y(y_k)) \in M_1M_2$.
     It follows that we have no more than
     $$O(|M_1| + |M_2| + |M_1M_2|) \cdot 2^{|\bar{x}|}$$ possible
     partial answers of this kind.

    \cMP{the below should be here, true?}
    To end the proof of \textbf{MM3} we recall that every minimal partial answer is either misguided or
    falls under one of the four cases \textit{Case 1-4.} above.
    We have shown that there is no more than $O(|M_1|+|M_2|)$ misguided minimal partial answers.
    Furthermore, every of the four cases
    admits $O(|M_1| + |M_2| + |M_1M_2|)$ minimal partial answers. Thus, \textbf{MM3} holds.

\section{Proofs for Section~\ref{sect:combined}}

\propcombinedptime*
\begin{proof} We first establish Theorem~\ref{prop:combinedptime} in
  the version where `minimal partial answers' are replaced with
  `partial answers'. Let an OMQ
  $Q(\bar x)=(\Omc,\Sigma,q) \in (\mathcal{ELH},\text{CQ})$ be given with $q$
  acyclic, as well as a $\Sigma$-database~$\Dmc$ and a tuple
  $\bar a^\ast \in (\mn{adom}(\Dmc) \cup \{ \ast \})^{|\bar x|}$.
%
  % We first observe that we can check in \PTime, given a tuple
  % $\bar b^\ast \in (\mn{adom}(\Dmc) \cup \{ \ast \})^{|\bar x|}$,
  % check whether $\bar a^\ast$ is a (not necessarily minimal)
  % partial answer to $Q$ on $\Dmc$. To do this, we
  We
  simultaneously
  modify $q$ and $\bar a^\ast$ as follows:
  \begin{itemize}

      \item let $q'(\bar x)$ be obtained from $q(\bar x)$ by quantifying
        all variables that have a `$\ast$' in the corresponding position
        of $\bar a^\ast$;

      \item let $\bar a$ be obtained from $\bar a^\ast$ by removing all
        wildcards.

      \end{itemize}
  We then check whether $\bar a \in Q'(\Dmc)$ where
  $Q'=(\Omc,\Sigma,q')$ and return the answer. This can be done in
  \PTime~\cite{DBLP:conf/ijcai/BienvenuOSX13}.

  \smallskip

  For the case of minimal partial answers, we first check using the
  above procedure whether $\bar a^\ast$ is a partial answer and return
  `no' if this is not the case. We then consider all tuples
  $\bar b^\ast \in (\mn{adom}(\Dmc) \cup \{ \ast \})^{|\bar x|}$ that
  may be obtained from $\bar a^\ast$ by replacing a single occurrence
  of `$\ast$' with a constant from $\mn{adom}(\Dmc)$. For each such
  tuple, we check whether it is a partial answer using the above
  procedure. If any of the checks succeeds, return~`no'.  Otherwise,
  return `yes'.
\end{proof}

\thmsinglemulti*
\begin{proof}
  The upper bound is established essentially as in the proof of
  Theorem~\ref{prop:combinedptime}, with two differences. First, when
  constructing $q'$ we now introduce one quantified variable per
  wildcard $\ast_i$ in $\bar a^\ast$, and thus the resulting CQ is not
  necessarily acyclic. Checking whether $\bar a^\ast$ is a partial
  answer is thus only possible in \NPclass now, rather than in \PTime.
  And second, we now consider all tuples tuples
  $\bar b^\ast \in (\mn{adom}(\Dmc) \cup \{ \ast \})^{|\bar x|}$ such
  that one of the following holds:
  \begin{enumerate}

      \item $\bar b^\ast$ is obtained from $\bar a^\ast$ by choosing a
        wildcard $\ast_i$ and replacing all occurrences of `$\ast_i$' with
        a constant from $\mn{adom}(\Dmc)$;

      \item $\bar b^\ast$ is obtained from $\bar a^\ast$ by choosing two
        distinct wildcard $\ast_i$ and $\ast_j$, replacing all occurrences
        of $\ast_j$ with $\ast_i$, and renaming the remaining wildcard
        to ensure that they are numbered consecutively.

      \end{enumerate}
  In summary, we then obtain an \NPclass upper
  bound.

  \smallskip

  For the lower bound, we use a reduction from positive 1-in-3-SAT. % ,
  % \NPclass-hardness proof is in Schaefer's famous 1978 paper.
  Let
  $\vp$ be a 3-formula with clauses $c_1,\dots,c_n$ and variables
  $z_1,\dots,z_m$. All literals are positive and we are interested in
  whether $\vp$ has a 1-in-3 assignment, that is, an assignment that
  satisfies exactly one variable per clause.
% Further let $x_{i,j}$ denote the variable used in
  % the $j$-th literal of $c_i$, for $1 \leq i \leq n$ and
  % $1 \leq j \leq 3$.
  % A \emph{clause type} is a word $t \in \{0,1\}^3$
  % that describes the `negation pattern' of a clause. For instance,
  % $t=100$ means that the first literal is positive while the second
  % and third literal are negative.

  Let \Omc be the \EL-ontology $\{ V \sqsubseteq \exists r . \top \}$
  and $q$ the acyclic CQ that contains one connected component for
  each clause $c_i$ of $\vp$, as follows:
  $$
  \begin{array}{l}
      v_1(y_i,y_{i,1}) \wedge v_2(y_i,y_{i,2}) \wedge
      v_3(y_i,y_{i,3}) \; \wedge \\[1mm]
      r(y_{i,1},x_{i,1}) \wedge
      r(y_{i,2},x_{i,2}) \wedge
      r(y_{i,3},x_{i,3}).
      \end{array}
  $$
  The answer variables are the variables $x_{i,j}$ with
  $1 \leq i \leq n$ and $1 \leq j \leq 3$.  The OMQ $Q$ is now
  $(\Omc,\Sigma,q)$ where $\Sigma$ is the set of relation symbols in
  \Omc and $q$.

  The tuple $\bar a^\Wmc$ to be tested uses wildcards
  $\ast_1,\dots,\ast_m$. In the position for each answer variable
  $x_{i,j}$, it has wildcard $\ast_\ell$ if the $j$-th variable in
  clause $i$ is $z_\ell$. Thus every wildcard is used multiple
  times, which ensures that occurrences of the same variable
  in different clauses receive the same truth value.

  The database $\Dmc$ contains the following facts:
  $$
  \begin{array}{lll}
        V(a_T) & V(a_F) \\[1mm]
        v_1(a_1,a_T) &     v_2(a_1,a_F) &     v_3(a_1,a_F) \\[1mm]
        v_1(a_2,a_F) &     v_2(a_2,a_T) &     v_3(a_2,a_F) \\[1mm]
        v_1(a_3,a_F) &     v_2(a_3,a_F) &     v_3(a_3,a_T)
      \end{array}
  $$
  Informally, each variable $y_i$ from each clause may be
  mapped to each of $a_1, a_2, a_3$, representing the
  options of making the first, second, or third variable in
  $c_i$ true while making all other variables false.

  It can be verified that $\vp$ has a 1-in-3 assignment
  if and only if $\bar a^\Wmc \in Q(\Dmc)^\Wmc$.
\end{proof}

\thmDP*

\begin{proof}
  For the upper bound, we use essentially the same strategy as in the
  proofs of Theorems~\ref{prop:combinedptime}
  and~\ref{thm:singlemulti}.  We first observe that in
  $(\mathcal{ELH},\text{CQ})$, single-testing (not necessarily
  minimal) partial answers with multi-wildcards is in~\NPclass.  Let
  us call this problem $P_1$. To implement the strategy followed in
  the proofs of Theorems~\ref{prop:combinedptime}
  and~\ref{thm:singlemulti} in \DPclass, we consider the intersection
  of $P_1$ with the following \coNP-problem $P_2$: given an OMQ
  $Q=(\Omc,\Sigma,q) \in (\mathcal{ELH},\text{CQ})$, a database~\Dmc,
  and a list of tuples
  $\bar a_1^\ast,\dots, \bar a_n^\ast \in (\mn{adom}(\Dmc) \cup \{
  \ast \})^{|\bar x|}$, whether $(Q,\Dmc,\bar a_i^\ast)$ is a
  no-instance of $P_1$, for some $i$ with $1 \leq i \leq n$.

  \medskip To prove hardness, we provide a polynomial time reduction
  of the \DPclass-complete problem 3COL-no3COL, defined as follows:
  given a pair $(G_1, G_2)$ of undirected graphs, decide whether $G_1$
  is 3-colorable and $G_2$ is not.

  Let $(G_1,G_2)$ be an instance of 3COL-no3COL.  We construct in
  polynomial time an OMQ $Q \in (\EL, \text{CQ})$, a database $\Dmc$,
  and a wildcard tuple $\bar{a}^\ast$ such that
  $\bar{a}^\ast \in Q(\Dmc)^\ast$ if and only if $G_1$ is
  3\=/colorable and $G_2$ is not.  Assume that $G_1 = ([n], E_1)$ and
  $G_2 = ([m], E_2)$.  Define $Q(x) = (\Omc, \Sigma, q)$ where
  $$
  \begin{array}{r@{\;}c@{\;}l}
        \Omc &=& \{ B \sqsubseteq \exists R.A \} \\[1mm]
        \Sigma &=& \{A,B,r,s,t\}\\[1mm]
        q(x) &=&\!\!\!\!\!\displaystyle\bigwedge_{(i,j) \in E_1} r(x_i,x_j) \land
                 \!\!\!\!\!\bigwedge_{(i,j) \in E_2} s(y_i,y_j) \land t(y_1,x) \land
                 A(x).
      \end{array}
  $$
    \noindent
  Further, define a $\Sigma$\=/database $\Dmc$ with domain
  $\{a,b,c,d,e\}$ and facts
    \begin{itemize}
            \item $r(a,b), r(b,c), r(c,a), r(b,a), r(c,b), r(a,c)$,
            \item $s(a,b), s(b,c), s(c,a), s(b,a), s(c,b), s(a,c)$,
            \item $t(a,d), A(d),$
            \item $s(e,e), B(e)$.
        \end{itemize}
    It should be clear that the first conjunction in $q$ checks for
    3-colorability of $G_1$ as the $r$-part of \Dmc is simply
    the CSP-template for 3-colorability. The remaining conjunction of
    $q$ check for non-3-colorability of $G_2$ in the following way:
    the answer variable $x$ can be mapped to the constant $d$ if and only if
    $G_2$ is 3-colorable and it can always be mapped to an anonymous
    element generated by $\Omc$. It is this easy to prove the
    following
    (for the multi-wildcard version, `$\ast$' is simply replaced with
    `$\ast_1$').
    \begin{claim}
        $Q(\Dmc)^{\ast} = \{\ast\}$ if and only if $G_1$ is 3\=/colorable and $G_2$ is not.
        \end{claim}
%
    % Let
    % \[q_1()  \leftarrow \bigwedge_{(i,j) \in E_1} r(x_i,x_j)\]
    % and
    % \[q_2(x) \leftarrow \bigwedge_{(i,j) \in E_2} s(y_i,y_j) \land t(y_1,x) \land A(x).\]
    %
    % It is folklore that $G_1$ is 3\=/colorable if and only if $q_1(\Dmc) = \{()\}$,
    % that is if there is a homomorphism from $q_1$ to $\Dmc$.
    % Moreover, it is clear that $q_1(\Dmc) = q_1(\mn{ch}_{\Omc}(\Dmc))$.
    %
    % Let $f \in \mn{ch}_{\Omc}(\Dmc)$ be a null added by the chase of $\Omc$.
    % Notice that there is a homomorphism $h$ from $q_2$ to $\mn{ch}_{\Omc}(\Dmc)$
    % such that $h(x) = f$. Thus, $*$ is always a partial answer to $q_2$ on $\mn{ch}_{\Omc}(\Dmc)$
    % with respect to nulls.
    % It is easy to check that if $G_2$ is 3\=/colorable then also $d$
    % is a partial answer to $q_2$ on $\mn{ch}_{\Omc}(\Dmc)$ with respect to nulls.
    % Since $d \prec \ast$, $*$ is a minimal partial answer to $q_2$
    % on $\mn{ch}_{\Omc}(\Dmc)$ with respect to nulls
    % if and only if $G_2$ is not 3\=/colorable.
    %
    % Thus, by the construction the only possible partial answers to $Q$ on $\Dmc$ are $\ast$ and $d$ and the following holds.
    %
    % \begin{itemize}
        %     \item If $Q(\Dmc)^{\ast} = \{\ast\}$ then
        %     $G_1$ is 3\=/colorable and $G_2$ is not.
        %     \item If $Q(\Dmc)^{\ast} = \{d\}$ then  $G_1$ is 3\=/colorable and so is $G_2$.
        %     \item If $Q(\Dmc)^{\ast} = \emptyset$ then  $G_1$ is not 3\=/colorable.
        % \end{itemize}
    %
    % Hence, $\ast \in Q(\Dmc)^{\ast}$ if and only if $G_1$ is 3\=/colorable and $G_2$ is not.
    % Thus, taking $\bar{a}^{\ast} = \ast$ we have constructed $Q$, $\Dmc$, and $\bar{a}^{\ast}$
    % as previously demanded, which ends the proof.
    \end{proof}

    \thmourelcomb*
    \begin{proof}
      We may once more use the strategy from the proofs of
      Theorems~\ref{prop:combinedptime} and~\ref{thm:singlemulti}.
      It amounts to polynomially many calls to an \ExpTime problem,
      thus resulting in \ExpTime overall complexity.
    \end{proof}

\end{document}